\documentclass[11pt,a4paper]{article}

\usepackage[a4paper,margin=1in]{geometry}
\usepackage{amsmath,amssymb,amsthm,mathtools,bm}
\usepackage{array,booktabs}
\usepackage{algorithm}
\usepackage{algpseudocode}
\usepackage{microtype}
\usepackage{enumitem}
\usepackage[numbers,sort&compress]{natbib}
\usepackage[hidelinks]{hyperref}
\usepackage{xcolor}
\usepackage{titlesec}

\allowdisplaybreaks
\setlist[itemize]{topsep=2pt,itemsep=1pt,leftmargin=2em}
\setlist[enumerate]{topsep=2pt,itemsep=1pt,leftmargin=2em}

\titleformat{\paragraph}[runin]{\normalfont\normalsize\bfseries}{}{0pt}{}
\titlespacing*{\paragraph}{0pt}{0.8ex plus .2ex minus .1ex}{0.6em}

\newtheorem{theorem}{Theorem}[section]
\newtheorem{proposition}[theorem]{Proposition}
\newtheorem{lemma}[theorem]{Lemma}
\newtheorem{corollary}[theorem]{Corollary}
\newtheorem{definition}[theorem]{Definition}
\newtheorem{remark}[theorem]{Remark}
\newtheorem{example}[theorem]{Example}

\newcommand{\succeqB}{\succeq_{\mathrm B}}
\newcommand{\simB}{\sim_{\mathrm B}}

\newcommand{\R}{\mathbb R}

\newcommand{\PD}{\mathbb S_{++}}

\newcommand{\OPT}{\operatorname{OPT}}
\newcommand{\Rfull}{R_{\mathrm{full}}}

\newcommand{\email}[1]{\texttt{#1}}

\title{Arbitrage-free Data Pricing}
\author{
Yihang Wu\\
 Zhejiang University\\
 \email{yhwu\_is@zju.edu.cn}
\and
Zhengyu Jin\\
 Zhejiang University\\
 \email{clovers2333@zju.edu.cn}
\and
Yicheng Fu\\
 Zhejiang University\\
 \email{fuycc@zju.edu.cn}
\and
Jinfei Liu\thanks{Corresponding author.}\\
 Zhejiang University\\
 \email{jinfeiliu@zju.edu.cn}
\and
Kui Ren\\
 Zhejiang University\\
 \email{kuiren@zju.edu.cn}
}
\date{}

\begin{document}
\maketitle
\thispagestyle{empty}

\begin{abstract}

We study optimal pricing of versioned data products when buyers can combine multiple purchases. A monopoly seller offers a menu of data products, and a buyer’s value for data is the improvement to their expected utility in a Bayesian decision problem. Since a buyer may purchase any finite bundle of products, including repeated copies of the same product, versioning creates arbitrage opportunities: a bundle of cheaper products may be more valuable than a product with a higher price. We formulate the arbitrage-free data selling problem which has infinite arbitrage-free constraints in general and possibly infinite state space, and prove its computational intractability: the problem admits no PTAS even for instances in which the state space is finite, and the problem admits no polynomial time constant factor approximation for succinct high-dimensional instances. On the positive side, when the numbers of buyer types and actions are constant, we give an additive FPTAS that handles possibly infinite state spaces and infinitely many arbitrage-free constraints, and empirically validate the algorithm on realistic synthetic data trading scenarios. We also analyze the posted pricing algorithm for selling only complete information and prove a tight approximation factor. We further identify a threshold utility regime in which arbitrage-freeness reduces to Blackwell dominance, which unifies known arbitrage-free conditions for dataset query and machine learning model pricing. Under this regime, we design efficient algorithms for several structured data menus common in practice.

\end{abstract}

\newpage
\setcounter{page}{1}

\section{Introduction} \label{sec:introduction}

The growing demand for data in advertising, finance, and machine learning has given rise to active markets for data products. Two prominent data product forms are datasets and machine learning models \cite{zhang2024surveydatamarkets}. Dataset markets like AWS Data Exchange, Snowflake Marketplace, and Dawex broker subscriptions and licenses for large catalogs of data products. Model markets like AWS SageMaker and Microsoft Azure offer pre-trained models and model training services.

From an economic perspective, buyers purchase data because it improves their decisions. An advertiser uses demographic and behavioral datasets to allocate a campaign budget; a retailer uses demand forecasts to choose inventory and staffing. In both cases, the value of data comes from its effect on a downstream decision problem. This suggests a natural viewpoint: selling data is fundamentally the sale of information that is useful for decision making.

Data differs from ordinary goods in its production cost structure. Collecting, cleaning, labeling, and training on data can be expensive, but once a useful product has been created, additional copies are often nearly costless. This makes versioning an effective selling strategy \cite{shapiro1998versioning}. A dataset seller may offer different queries instead of selling only the full database \cite{koutris2015query,koutris2013toward,deep2017design}. A model seller may offer versions with different accuracy \cite{chen2019towards,liu2021dealer}. Such versioning allows the seller to serve heterogeneous buyers and extract more revenue than selling a single product.

Versioning, however, creates \emph{arbitrage}. If a menu is not designed carefully, a buyer can combine multiple cheap products and recover the information contained in an expensive one.

\begin{example}[Query arbitrage] \label{ex:query-arbitrage}
  An advertiser wants to purchase a complete table $D$ of target users for a marketing campaign. Let $D = D_{\mathrm{male}} \cup D_{\mathrm{female}}$, where \(D_{\mathrm{male}}\) and \(D_{\mathrm{female}}\) contain the male and female target users, respectively. Suppose the seller offers three query products: $Q_{\mathrm{male}}(D) = D_{\mathrm{male}}$, $Q_{\mathrm{female}}(D) = D_{\mathrm{female}}$, $Q_{\mathrm{all}}(D) = D$. If the posted prices are $p(Q_{\mathrm{male}}) = p(Q_{\mathrm{female}}) = 1000$, $p(Q_{\mathrm{all}}) = 3000$, then the advertiser would purchase \(Q_{\mathrm{male}}\) and \(Q_{\mathrm{female}}\) at total cost \(2000 < 3000\). Hence the menu admits arbitrage.
\end{example}

\begin{example}[Model arbitrage] \label{ex:model-arbitrage}
 A buyer wants a model for predicting next month's sales. Let \(y^*\in\R\) denote the prediction produced by the seller's best trained model. The seller versions this model by adding Gaussian noise and offers $M_1 = y^* + Z_1$, where $Z_1\sim\mathcal N(0, 1)$, at price \(500\), and $M_2 = y^* + Z_2$, where $Z_2\sim\mathcal N(0, 2)$, at price \(200\). Suppose repeated purchases are independent conditional on \(y^*\). Then if the buyer purchases two copies of \(M_2\), their average is $y^* + \frac{Z_2^{(1)} + Z_2^{(2)}}{2} \sim \mathcal N(y^*, 1)$. Thus the average of two \(M_2\) purchases has exactly the same distribution as \(M_1\), but costs only \(200 + 200 = 400 < 500\). This is again arbitrage.
\end{example}

Existing work on arbitrage-free pricing for query products \cite{chawla2019revenue} and model products \cite{chen2019towards,liu2021dealer} typically studies seller revenue maximization under the assumption that the buyer's value for each data product is exogenously given. However, in practice the seller only has incomplete information about the buyer's downstream task or utility function. Moreover, once value is tied to decision making, the arbitrage-free pricing rules studied in previous work (as the two examples above) need not be sufficient: they only consider the cases of exact reconstruction where the cheaper products yield the expensive one. However, if a bundle of products cannot reconstruct another product, it may still be more valuable under some decision problems.

A parallel literature on selling information starts from Bayesian decision problems \cite{babaioff2012optimal,bergemann2018design,cai2020sell,chen2020selling,li2025sell}, but it usually assumes that a buyer selects one product from the menu. However, a buyer can create multiple accounts or repeatedly query an API. Example~\ref{ex:screening-arbitrage-gap} shows that if arbitrage constraints are ignored, a menu that is optimal under standard information selling constraints may still allow buyers to arbitrage, thereby reducing the seller's revenue.

\subsection{Our contributions}

Accordingly, we investigate the seller's optimal data menu design and pricing problem where pricing satisfies arbitrage-free (AF) constraints and where the value of data is derived from the improvement to the buyer's decision utility. We first unify versioned datasets and models as experiments, and formulate AF constraints for information selling. In this formulation, buyers may purchase several products (a bundle), including repeated copies of the same product, and each bundle induces a composite experiment. We prove that incentive compatibility and individual rationality are special cases of AF constraints, thus we strictly generalize the previous information selling setting. We also establish a revelation principle adapted to AF constraints: without loss of generality, one may restrict attention to direct AF and responsive menus.

We then formulate the general AF problem, which has infinitely many AF constraints in general and possibly infinite state space, and prove its computational intractability: even explicit finite instances admit no PTAS, and succinct high-dimensional instances admit no polynomial time constant factor approximation. On the positive side, if we treat the numbers of buyer types and actions as constant, we can give an additive FPTAS that addresses the two main sources of infinitude: it groups the state space into finite cells with similar utility profiles and truncates the bundle space by removing low price products. Then the problem is a finite convex program over cellwise recommendation profiles that can be solved by the ellipsoid method. We empirically validate the algorithm on realistic synthetic data trading scenarios. Finally, we use the revelation principle to get the direct AF menu. We also prove a tight approximation guarantee for a simple complete information posted pricing algorithm.

We further identify that under threshold utilities, where a constant threshold value is obtained exactly when an experiment suffices for a target task, AF reduces to Blackwell AF: no cheaper bundle may Blackwell dominate the assigned experiment. Thus we can unify the AF conditions for query and model pricing in previous works as special cases of Blackwell AF, and indicate that the AF conditions in previous works are not sufficient for data value generated from Bayesian decision problems.

Finally, we study optimal pricing for structured data menus under threshold utilities. We first show that the Blackwell AF pricing problem remains hard, and the main difficulty is that pricing chooses which buyer types receive nonempty information. We give solutions from three perspectives: (1) we find two practical conditions under which full buyer surplus extraction is possible; (2) we give a polynomial time logarithmic approximation when Blackwell AF violations can be detected in polynomial time; (3) we find two realistic menu structures that yield dynamic programs, one of which generalizes the algorithms in \cite{chen2019towards,liu2021dealer}. We relate all the results to real world data products and thus our results can be applied to practical data pricing problems.

\subsection{Related work}

We focus on the three strands of literature most closely related to this paper: arbitrage-free pricing, information selling and multiproduct monopoly and bundling. For surveys on other data and information pricing topics, see \cite{pei2020survey,cong2022data,zhang2024surveydatamarkets,bi2024when,bergemann2019markets,bergemann2021information}.

\paragraph{Arbitrage-free data pricing.} Arbitrage-free data pricing was first introduced by \cite{balazinska2011data}. \cite{koutris2015query} was the first to formalize arbitrage-freeness in query markets and studied the computational complexity of a special class of arbitrage-free pricing functions. Building on this line of work, \cite{li2014theory} studied arbitrage-free pricing when noise is added to queries. \cite{lin2014arbitrage,deep2017design} considered more general query pricing functions, and in particular \cite{deep2017design} provided necessary and sufficient conditions for arbitrage-free pricing. \cite{chawla2019revenue} studied seller revenue guarantees for heuristic arbitrage-free pricing algorithms when buyer utilities are given. \cite{chen2019towards} was the first to introduce arbitrage-freeness into machine learning model pricing; they characterized the necessary and sufficient conditions for arbitrage-free pricing and studied algorithms for maximizing seller revenue given buyers' utilities. \cite{liu2021dealer} used differential privacy to add noise thus protecting model privacy and further discussed the revenue allocation problem among data providers based on the Shapley value.

Our work unifies query and model pricing within the general information pricing problem. We show that the arbitrage-free conditions proposed in these prior works are special cases of Blackwell dominance, and therefore are not sufficient in the general Bayesian decision setting.

\paragraph{Selling information.} \cite{babaioff2012optimal} was the first to study the sale of information with multiple rounds of interaction between buyer and seller. They focused on revelation principles under both independent and correlated buyer-seller information structures and studied the associated computational complexity. \cite{chen2020selling} further incorporated budget constraints and interpreted optimal information selling mechanisms as practical consulting mechanisms. \cite{liu2021optimal} showed that, in a special class of decision problems and buyer utility specifications, the optimal mechanism for selling information takes the form of a natural threshold mechanism. \cite{bergemann2018design} studied the menu-based sale of information, in which buyer and seller interact only once, and characterized the optimal menu when the state space is binary. \cite{cai2020sell} studied the corresponding optimal information pricing problem from an algorithmic perspective. \cite{li2025sell} proposed a sampling-based algorithm to compute the optimal menu when the state space is large. \cite{bergemann2022selling} examined the approximate optimality of selling full information. \cite{li2022selling} studied mechanisms for selling information when data buyers possess endogenous information.

Our work builds on the menu-based model of selling information in \cite{bergemann2018design} and introduces arbitrage-free constraints into the information pricing problem, which is a non-trivial generalization of previous works and brings new computational challenges.

\paragraph{Multiproduct monopoly and bundling.} Our work is also related to multiproduct monopoly and bundling. The classical literature mainly studies how a monopolist should package existing ordinary goods and choose prices, including pure and mixed bundling \cite{adams1976commodity,mcafee1989multiproduct}, multiproduct nonlinear pricing and multidimensional screening \cite{wilson1993nonlinear,armstrong1996multiproduct,rochet1998ironing}, and structural or approximation results for multi-item revenue maximization and simple bundling mechanisms \cite{manelli2006bundling,daskalakis2015strong,hart2017approximate,haghpanah2021pure,ghili2023characterization,yang2025nested}. A smaller line explicitly considers arbitrage problem induced by the failure of purchase monitoring: \cite{mcafee1989multiproduct} distinguish monitored and unmonitored purchases, and \cite{carvajal2010bundling} study bundling with resale opportunities, where buyers may choose arbitrary combinations of offered bundles.

Our setting differs since we do not optimize how to bundle a fixed set of ordinary goods. Our focus is arbitrage-freeness of data product menus against arbitrary finite bundle deviations, rather than optimal package design for existing goods.

\section{Preliminaries} \label{sec:preliminary}

\subsection{Model} \label{subsec:model}

\paragraph{Basic setup.} We consider a monopoly data seller selling data to a data buyer who faces a decision problem under uncertainty. The state of the world $\omega$ lies in a measurable state space $\Omega$, is drawn according to a publicly known common prior $\mu \in \Delta(\Omega)$, and can be observed by the data seller. The buyer has a private type $\theta$ drawn from a finite set $\Theta$, and chooses an action $a$ from a finite set $A$. The buyer's ex-post utility from taking action $a$ when the state is $\omega$ and the type is $\theta$ is a measurable function $u_\theta(\omega, a) \in [0, 1]$. Let the prior mass of type \(\theta\) be \(\lambda_\theta > 0\).

\paragraph{Experiment.} The data seller offers a menu of experiments (or signaling schemes) about the state $\omega$. An experiment \(E = (S_E, K_E)\) is a measurable signal space \(S_E\) and a Markov kernel $K_E(\cdot\mid\omega) \in \Delta(S_E), \forall \omega \in \Omega$. When the state of the world is \(\omega\), the experiment draws \(s\sim K_E(\cdot\mid\omega)\) and sends \(s\) to the buyer.

\cite{bergemann2018design,cai2020sell} treat the experiments as decision information that directly affects the buyer's payoff under uncertainty, and we will prove in Section~\ref{subsec:arbitrage} that query answers and model predictions can also be represented in the same framework. Thus the experiment provides a unifying framework for all different types of data products.

Denote $\mathcal{M} = \{(E_j, t_j)\}_{j = 1}^N$ as the menu posted by the seller, where $E_j$ is an experiment and $t_j$ is its price. In our setting, the buyer may purchase several experiments simultaneously (including repeated purchases of the same experiment), or equivalently create multiple accounts and buy experiments separately under different identities. When a buyer purchases multiple experiments, we assume that conditional on the true state of the world, the seller sends independent signals across purchases. This independence assumption is natural: because the buyer can always split purchases across different accounts, each purchase can be viewed as a separate draw from the corresponding experiment. The following definition formalizes the information generated by an arbitrary finite bundle of purchases.

\begin{definition}
  Let \(B = (i_1, \ldots, i_\ell) \in [N]^{\ell}\) be a list of indices of experiments in the menu. The composite experiment associated with \(B\), denoted by $E_B := \bigotimes_{r = 1}^{\ell} E_{i_r}$ has signal space $S_B := \prod_{r = 1}^{\ell} S_{E_{i_r}}$ and kernel
  \[ K_{E_B}(s_1, \ldots, s_\ell \mid \omega) = \prod_{r = 1}^{\ell}K_{E_{i_r}}(s_r \mid \omega), \ \forall (s_1, \ldots, s_\ell) \in S_B, \omega \in \Omega. \] 
 Besides, the empty bundle gives the uninformative experiment.
\end{definition}

Based on the notation above, the interaction between the buyer and the seller proceeds as follows:
\begin{enumerate}
 \item The seller posts a menu $\mathcal{M}$.
 \item The true state $\omega$ and the buyer's type $\theta$ are realized.
 \item The buyer selects a composite experiment $E_B$ from the menu and pays the corresponding total price $t_B = \sum_{r = 1}^{\ell} t_{i_r}$.
 \item The seller sends the buyer a signal $s \in S_B$ from the experiment $E_B$ according to the distribution $K_{E_B}(\cdot\mid\omega)$.
 \item The buyer chooses an action $a$ and obtains utility $u_\theta(\omega, a)$.
\end{enumerate}

\paragraph{The value of a composite experiment.} Before purchasing any experiment, type $\theta$ chooses the action that maximizes expected utility under the prior $\mu$ and obtains \footnote{If $\Omega$ or $S_E$ is infinite, then the summation can be written as an integral. We use the summation notation for simplicity, and the results in the paper hold in the more general setting as well.}
\[ U_\theta^0 = \max_{a\in A} \sum_{\omega\in\Omega} \mu(\omega) u_\theta(\omega, a). \]
After receiving a signal $s$ from experiment $E_B$, the buyer updates his belief to 
\[ \mu(\omega \mid s) = \frac{\mu(\omega) K_{E_B}(s \mid \omega)}{\sum_{\omega' \in \Omega} \mu(\omega') K_{E_B}(s \mid \omega')}, \]
and chooses the action that maximizes expected utility based on the posterior belief $\mu(\cdot \mid s)$, obtaining expected utility
\[ u_\theta(s) = \max_{a \in A} \sum_{\omega \in \Omega} \frac{\mu(\omega) K_{E_B}(s \mid \omega)}{\sum_{\omega' \in \Omega} \mu(\omega') K_{E_B}(s \mid \omega')} u_\theta(\omega, a). \]
Since the probability of receiving signal $s$ is $\sum_{\omega \in \Omega} \mu(\omega) K_{E_B}(s \mid \omega)$, the expected utility of experiment $E_B$ to type $\theta$ is
\[ W_\theta(E_B) = \sum_{s \in S_{E_B}} \max_{a \in A} \sum_{\omega \in \Omega} \mu(\omega) K_{E_B}(s \mid \omega) u_\theta(\omega, a). \]
Thus the value of a composite experiment \(E_B\) to type \(\theta\) is the increase in expected utility from purchasing \(E_B\) compared to purchasing nothing:
\[ V_\theta(E_B) = W_\theta(E_B) - U_\theta^0. \]

\paragraph{AF and responsive menu.} A direct menu $\mathcal{M} = \{E_\theta, t_\theta\}_{\theta \in \Theta}$ assigns an experiment $E_\theta$ and a price $t_\theta$ to each type $\theta$. We call a direct menu \(\mathcal M\) \textit{arbitrage-free (AF)} if for every type \(\theta\), the intended experiment \(E_\theta\) is the optimal composite experiment for type \(\theta\):
\begin{equation} \label{eq:AF}
 V_\theta(E_\theta) - t_\theta \geqslant V_\theta(E_B) - t_B, \forall B \neq (\theta), \theta \in \Theta.
\end{equation}

Note that AF is a natural extension of incentive compatibility (IC, where $B$ is a singleton) to the setting with composite experiments. Besides, if $B = \varnothing$, then the AF constraint reduces to individual rationality (IR), which requires that the buyer's utility from purchasing the intended experiment is at least the utility from purchasing nothing.

For \(\alpha \geqslant 0\) and a collection \(\mathcal B\) of bundles, we say that a menu \(\mathcal M\) is \emph{\(\alpha\)-AF against \(\mathcal B\)} if \(V_\theta(E_\theta) - t_\theta \geqslant V_\theta(E_B) - t_B - \alpha\) for every \(\theta \in \Theta\) and every \(B \in \mathcal B \setminus \{(\theta)\}\). When \(\mathcal B\) is the set of all finite bundles, we simply say that \(\mathcal M\) is \emph{\(\alpha\)-AF}. Thus \(0\)-AF coincides with AF.

The next proposition is an extension of the classical revelation principle to the setting with composite experiments, and it shows that we can focus on direct menus without loss of generality. The omitted proofs in this section are deferred to Appendix~\ref{app:omitted-proofs-sec2}.

\begin{proposition}[Revelation principle] \label{prop:bundle-revelation}
 For any menu $\mathcal{M}$, there is a direct menu $\mathcal{M}'$ that has the same expected seller revenue and is AF.
\end{proposition}

Besides AF, we also need to impose obedience constraints for the intended types to restrict the number of signals that the seller needs to design for each experiment. We call a menu \textit{responsive} if for every type $\theta$, every signal $s \in S_{E_{\theta}}$ of the intended experiment $E_\theta$ leads type $\theta$ to a different optimal action. The following proposition shows that we can focus on responsive direct menus without loss of generality.

\begin{proposition}[Generalization of Proposition~1 of \cite{bergemann2018design}] \label{prop:obedience-revelation}
 For any AF direct menu $\mathcal{M}$, there is a responsive direct menu $\mathcal{M}'$ that has the same expected seller revenue and is AF.
\end{proposition}

The proof is included in Appendix~\ref{app:omitted-proofs-sec2}. This proposition implies that for each type \(\theta\), the seller only needs to design an experiment \(E_\theta\) whose signals can be identified with recommended actions. If the menu is responsive, then the buyer will always follow the recommendation of the intended experiment, and we call the buyer \textit{obedient}.

\begin{remark}[Relation to \cite{bergemann2018design,cai2020sell,babaioff2012optimal,chen2020selling,li2025sell}]
 Our model is closest to the menu-based sale of experiments in \cite{bergemann2018design,cai2020sell,li2025sell}. We assume that buyers have different utility functions, but share a common prior distribution over the state of the world, which is the same assumption as in \cite{li2025sell}. In contrast, \cite{bergemann2018design,cai2020sell} assume that all buyers have the same utility function, but may have different prior distributions over the state of the world. As \cite{li2025sell} points out, so long as the changes of measure $\textup{d}\mu_{\theta} / \textup{d}\mu$ exist, the problem with different prior distributions $\mu_{\theta}$ and utility functions $u_{\theta}$ is equivalent to the problem with a shared prior $\mu$ and utility functions $u_{\theta}(\omega, a) \textup{d}\mu_{\theta} / \textup{d}\mu$.
 
 The works \cite{babaioff2012optimal,chen2020selling} study generic interactive protocols for selling information. Such protocols are more general than experiment menus: they may sell either an experiment ex ante or a realized signal ex post. When all buyer types have the same prior, namely $\mu_\theta = \mu, \forall \theta \in \Theta$, the optimal generic interactive protocol reduces to a menu of experiments, so our model captures that case as well. In contrast, when buyers have heterogeneous priors, the optimal generic interactive protocol need not be representable as a menu of experiments. Extending AF constraints to this more general setting is an interesting direction for future work.
\end{remark}

\subsection{Arbitrage-free data pricing} \label{subsec:arbitrage}

In Section~\ref{sec:program}, we will show that the general arbitrage-free information selling problem is hard to solve directly. Thus in Section~\ref{sec:blackwell}, we will focus on a special case where the arbitrage constraints can be reduced to the following intuitive and more tractable condition.

\paragraph{Query pricing.} Let \(\mathcal D\) be a database-instance space (all possible database instances), a (deterministic) query \(Q\) is a map $Q: \mathcal D\to\mathcal A_Q$. For every database instance \(D\), the query $Q$ can be written as an experiment \(E_Q\) with signal space \(\mathcal A_Q\) and kernel
\[ K_Q(a\mid D) = \mathbf{1}\{Q(D) = a\}, \forall a \in \mathcal A_Q, D \in \mathcal D. \]

To state the arbitrage constraints in query pricing, we need to compare the informativeness of different queries. A natural way is to look at the partitions of \(\mathcal D\) induced by the queries \cite{deep2017design}.

\begin{definition}[query-induced partition]
 Define the equivalence relation
 \[ D \sim_Q D' \iff Q(D) = Q(D'), \]
 and let \(P_Q\) denote the corresponding partition of \(\mathcal D\).
\end{definition}

If \(P_{Q_1}\) refines \(P_{Q_2}\), then \(Q_1\) is more informative than \(Q_2\). This is quite intuitive. For example, consider the query that returns the entire database, which is the most informative query. The induced partition places each database instance in its own singleton cell; hence it is the finest partition. By contrast, consider a constant query that returns a fixed value on every database instance. Such a query conveys no information at all, and its induced partition consists of a single cell containing the entire database-instance space; hence it is the coarsest partition.

\cite{deep2017design} prove that a pricing function \(p\) is arbitrage-free if and only if it satisfies the following conditions for every pair of queries \(Q_1, Q_2, Q\):
\begin{align*}
 P_{Q_1} \preceq P_{Q_2} &\Longrightarrow p(Q_1) \leqslant p(Q_2), \\
 P_Q = P_{Q_1} \vee P_{Q_2} &\Longrightarrow p(Q) \leqslant p(Q_1) + p(Q_2),
\end{align*}
where \(P_{Q_1}\preceq P_{Q_2}\) means that \(P_{Q_2}\) refines \(P_{Q_1}\), and \(P_{Q_1}\vee P_{Q_2}\) is the coarsest common refinement of \(P_{Q_1}\) and \(P_{Q_2}\).

\paragraph{Model pricing.} Let \(\theta^*\in\mathbb R^d\) denote the trained model. In model pricing \cite{chen2019towards,liu2021dealer}, the seller versions this model by releasing Gaussian noisy instances of \(\theta^*\): a model product with noise level \(\delta>0\) returns
\[ \hat \theta = \theta^* + Z, \qquad Z \sim \mathcal N(0, \delta I_d / d), \]
where $I_d$ is the $d$-dimensional identity matrix. This noisy model can be viewed as an experiment \(E_\delta\) with signal space \(\mathbb R^d\) and kernel
\[ K_{E_\delta}(\hat \theta \mid \theta^*) = \frac{1}{(2\pi\delta/d)^{d/2}} \exp\left(-\frac{d}{2\delta} \|\hat \theta - \theta^*\|^2\right), \forall \hat \theta, \theta^* \in \mathbb R^d. \]

To state the arbitrage constraints in model pricing, denote the precision of noise as \(x = 1 / \delta\). If \(x_1 \geqslant x_2\), then the model with precision \(x_1\) is intuitively more informative than the model with precision \(x_2\). \cite{chen2019towards} prove that a pricing function \(p\) is arbitrage-free if and only if it satisfies the following conditions for every pair of precisions \(x_1,x_2\):
\begin{align*}
 x_1 \leqslant x_2 &\Longrightarrow p(x_1) \leqslant p(x_2), \\
 p(x_1+x_2) &\leqslant p(x_1)+p(x_2).
\end{align*}

\paragraph{Blackwell order.} As discussed above, query and model pricing are special cases of information pricing since queries and noisy models can be represented as experiments. In Section~\ref{sec:blackwell}, we will show that the arbitrage constraints in query pricing and model pricing can be unified through the Blackwell order on experiments.

\begin{definition}[Blackwell dominance]
 Experiment \(E\) Blackwell-dominates experiment \(F\), written \(E\succeqB F\), if there is a Markov kernel \(\Gamma:S_E\to\Delta(S_F)\) such that
 \[ K_F(t \mid\omega) = \sum_{s \in S_E} \Gamma(t \mid s)K_E(s \mid \omega), \forall t \in S_F, \omega \in \Omega. \]
 We write \(E \simB F\) when both \(E \succeqB F\) and \(F \succeqB E\) hold.
\end{definition}

Blackwell's theorem \cite{blackwell1953equivalent} says that \(E\succeqB F\) if and only if \(E\) gives weakly higher value than \(F\) for every prior and decision problem. In particular,
\[ E\succeqB F \Longrightarrow W_\theta(E)\geqslant W_\theta(F), \forall\theta \in \Theta. \]

The following example illustrates that if we ignore arbitrage constraints, then a menu that is optimal under standard information selling constraints may admit arbitrage behavior of buyers, thereby reducing the seller's revenue.

\begin{example} \label{ex:screening-arbitrage-gap}
Let \(\Omega = \{\omega_1, \omega_2, \omega_3, \omega_4\}\), \(A = \{a_1, a_2, a_3, a_4\}\), and let every type share the same uniform prior. Define the three type-dependent payoff matrices directly as
\[
u_A=
\begin{array}{c|cccc}
& a_1&a_2&a_3&a_4\\\midrule
\omega_1&1&0&1/10&0\\
\omega_2&0&0&0&0\\
\omega_3&1/30&0&1/3&0\\
\omega_4&0&0&0&0
\end{array},
\qquad
u_B=
\begin{array}{c|cccc}
& a_1&a_2&a_3&a_4\\\midrule
\omega_1&1&0&1/10&0\\
\omega_2&0&1/3&0&1/30\\
\omega_3&0&0&0&0\\
\omega_4&0&0&0&0
\end{array},
\]
\[
u_C=
\begin{array}{c|cccc}
& a_1&a_2&a_3&a_4\\\midrule
\omega_1&8/15&0&4/75&0\\
\omega_2&0&4/15&0&2/75\\
\omega_3&2/75&0&4/15&0\\
\omega_4&0&2/75&0&4/15
\end{array}.
\]

Let \(E_1=(\{s_1,s_2\},K_{E_1})\) and \(E_2=(\{t_1,t_2\},K_{E_2})\), where the rows of the following matrices are indexed by signals and the columns are indexed by states \(\omega_1,\omega_2,\omega_3,\omega_4\):
\[
K_{E_1}=
\begin{pmatrix}
1&1&0&0\\
0&0&1&1
\end{pmatrix},
\qquad
K_{E_2}=
\begin{pmatrix}
1&0&1&0\\
0&1&0&1
\end{pmatrix}.
\]
Thus \(E_1\) reveals whether the state lies in \(\{\omega_1,\omega_2\}\) or \(\{\omega_3,\omega_4\}\), and \(E_2\) reveals whether the state lies in \(\{\omega_1,\omega_3\}\) or \(\{\omega_2,\omega_4\}\). Let \(E_3\) be the fully informative experiment. A direct calculation of posterior optimal actions gives the following surplus table under the common prior:
\[
\begin{array}{c|ccc}
& E_1&E_2&E_3\\\midrule
A&3/40&0&3/40\\
B&0&1/12&1/12\\
C&3/50&11/150&29/150
\end{array}.
\]
Set prices \(t_1=3/40\), \(t_2=1/12\), and \(t_3=29/150\). Type \(A\) buys \(E_1\), type \(B\) buys \(E_2\), and type \(C\) buys \(E_3\), each with zero net surplus under the intended assignment. This is the optimal menu without AF constraint due to the IR constraint.

However, \(E_1\otimes E_2\) is Blackwell equivalent to \(E_3\), and
\[
t_1+t_2=\frac{3}{40}+\frac{1}{12}=\frac{19}{120}<\frac{29}{150}=t_3.
\]
Thus a buyer can synthesize the fully informative experiment more cheaply than its posted price. The violation arises under a common prior; the heterogeneity is entirely in the payoff functions.
\end{example}

\section{General Arbitrage-free Information Selling} \label{sec:program}

In this section, we study the general arbitrage-free information selling problem. We first formulate the mathematical program and give hardness results for the problem. We then give an additive FPTAS for the regime in which the number of buyer types and buyer actions are constants. Finally, we study simple approximation algorithms obtained by restricting the experiments offered for sale.

\subsection{The mathematical program and hardness results} \label{subsec:unrestricted-program}

Equation~\eqref{eq:AF} requires that, for every type, no bundle of experiments yields higher utility than the intended experiment. In general, we should consider infinitely many possible bundle deviations: repeated purchases of the same experiment may strictly increase informativeness, as the following example shows.

% , while Section~\ref{subsec:simple-approx} shows that restricting attention to special experiment classes that simplify the AF constraints may incur poor approximation guarantees

\begin{example} \label{ex:repeated-purchase-strict}
 Let \(\Omega = \{0,1\}\), let the signal space be \(S = \{0, 1\}\), and consider the binary symmetric experiment
 \[ K_E = \begin{pmatrix} 1 - \varepsilon & \varepsilon \\ 
   \varepsilon & 1 - \varepsilon
  \end{pmatrix}, \qquad 0 < \varepsilon < 1 / 2. \]
 Consider the decision problem in which the buyer wishes to guess the state, receives payoff one for a correct guess and zero otherwise, and has a uniform prior. Then for every odd \(m\), $E^{\otimes(m + 2)} \succ_B E^{\otimes m}$.
\end{example}

The omitted proofs of this section are deferred to Appendix~\ref{app:omitted-proofs-sec3}. Based on the revelation principle and the discussion above, we then formulate the mathematical program for the general arbitrage-free information selling problem, where the seller aims to maximize expected revenue subject to AF and obedience constraints.

Since all signals are conditionally independent and exchangeable, it is convenient to index bundles by count vectors rather than ordered lists. Let $n = |\Theta|$. For a count vector \(c = (c_1, \ldots, c_n) \in \mathbb Z_+^\Theta\), where $c_{\theta}$ is the number of times that experiment \(E_\theta\) is purchased, we denote
\[ E_c := \bigotimes_{\theta \in \Theta}E_{\theta}^{\otimes c_{\theta}}, \quad c \cdot t := \sum_{\theta \in \Theta}c_{\theta}t_{\theta}. \]
Let the unit vector \(e_\theta\) denote the intended purchase of product \(\theta\), and \(c = 0\) denotes the outside option.

By Proposition~\ref{prop:obedience-revelation}, we may focus on responsive direct menus. Let \(m = |A|\), denote \(A = \{a_1, \ldots, a_m\}\), and let $\pi_{\omega, a}(\theta)$ be the probability that product \(\theta\) recommends action \(a\) in state \(\omega\). Then the expected utility of type \(\theta\) from purchasing the intended product can be written as
\[ W_\theta = \sum_{a \in A}\sum_{\omega \in \Omega} \mu(\omega)\pi_{\omega, a}(\theta) u_\theta(\omega, a). \]

For a count vector \(c\), denote $z_{\theta, a}$ as the number of times that \(E_\theta\) recommends action \(a\) in the bundle \(c\). We can then define the count space
\[ \mathcal Z_\Theta(c) := \left\{z = (z_{\theta, a})_{\theta, a} \in \mathbb Z_+^{\Theta \times A}: \sum_{a \in A}z_{\theta, a} = c_{\theta}, \forall \theta \in \Theta\right\}. \]
For every \(z\in\mathcal Z_\Theta(c)\), define the probability that bundle \(c\) produces the count vector \(z\) in state \(\omega\)
\begin{equation} \label{eq:count-prob}
 P(z \mid c,\omega) = \prod_{\theta \in \Theta} \left[ \frac{c_{\theta}!}{\prod_{a \in A}z_{\theta, a}!} \prod_{a \in A}\pi_{\omega, a}(\theta)^{z_{\theta, a}} \right].
\end{equation}
The AF information selling problem can now be written as the following mathematical program:
\begin{equation} \label{eq:unrestricted-program}
 \begin{aligned}
  \max_{\pi, t, y} \quad & \sum_{\theta \in \Theta} \lambda_\theta t_\theta \\
  \text{s.t.} \quad & W_\theta - t_\theta \geqslant \sum_{z \in \mathcal Z_\Theta(c)} y_{\theta, c, z} - c \cdot t, && \forall \theta \in \Theta,\ \forall c \in \mathbb Z_+^\Theta \setminus \{e_\theta\}, \\
  & y_{\theta, c, z} \geqslant \sum_{\omega \in \Omega} \mu(\omega)P(z \mid c, \omega) u_\theta(\omega, a), && \forall \theta, c, z, \ \forall a \in A, \\
  & \sum_{\omega \in \Omega}\mu(\omega)\pi_{\omega, a}(\theta)u_\theta(\omega, a) \geqslant \sum_{\omega \in \Omega} \mu(\omega) \pi_{\omega, a}(\theta) u_\theta(\omega, a'), && \forall \theta, \ \forall a, a' \in A,\\
  &\sum_{a \in A}\pi_{\omega, a}(\theta) = 1, \quad \pi_{\omega, a}(\theta)\geqslant 0, && \forall \theta, \omega, a, \\
  & t_\theta \geqslant 0, y_{\theta, c, z} \geqslant 0, && \forall \theta, c, z
 \end{aligned}
\end{equation}
Note that $y_{\theta, c, z} = \max_{a \in A} \sum_{\omega \in \Omega} \mu(\omega)P(z \mid c, \omega) u_\theta(\omega, a)$ is the value of deviation bundle \(c\) after observing count vector \(z\). We use $y_{\theta, c, z}$ as an auxiliary variable to linearize the program. The objective is to maximize expected revenue. The first line of constraints is AF against any count vector. The third line is the obedience constraint. The last two lines are feasibility constraints. This program is infinite because \(c\) ranges over \(\mathbb Z_+^\Theta\) and because, in some data applications, \(\Omega\) itself may be continuous (the summations can be replaced by integrals in that case).

Although the most immediate difficulties in Program~\eqref{eq:unrestricted-program} come from the infinite family of AF constraints and the possibility of a continuous state space, the first part of the next theorem (and its proof), shows that even when both the state space and the AF constraints are finite, no PTAS exists. The second part shows that one can construct a succinct implicit description of an AF information selling problem for which, unless \(P = NP\), Program~\eqref{eq:unrestricted-program} admits no polynomial time constant factor approximation algorithm.

\begin{theorem} \label{thm:af-revenue-hardness}
 \begin{enumerate}[label=(\roman*)]
  \item Unless \(P = NP\), there is no polynomial time approximation scheme (PTAS) for \eqref{eq:unrestricted-program}.
  \item There exists a succinct high dimensional description of the general AF information selling problem such that, unless \(P = NP\), there is no polynomial time constant factor approximation algorithm for \eqref{eq:unrestricted-program}.
 \end{enumerate}
\end{theorem}

The complete proof is in Appendix~\ref{app:omitted-proofs-sec3}. The main idea in part~(i) is a reduction from Maximum Independent Set on bounded degree graphs, which is known to be APX-hard \cite{alimonti2000some}. The construction turns vertices and edges into buyer types and uses deterministic subset queries so that a buyer's value depends only on whether the union of the purchased queries covers its target vertex or edge. We then derive lower and upper bounds of the optimal revenue that are close to each other. Based on the bounds, we can give a polynomial time reduction from a PTAS for Maximum Independent Set to a PTAS for the AF information selling problem, which yields the desired contradiction. Part~(ii) uses a succinctly represented high-dimensional prior and utilities, and the optimal revenue is exactly the probability of a satisfying assignment to a SAT instance, distinguishing zero from positive revenue then rules out any polynomial time constant factor approximation.

\subsection[An FPTAS for constant |Theta| and |A|]{An FPTAS for constant \(|\Theta|\) and \(|A|\)} \label{subsec:grid-price}

The hardness results established in the previous subsection indicate that it is difficult to obtain an algorithm that is both computationally efficient and close to optimal. We therefore consider, in this subsection, a special but common setting in real data markets, where the number of buyer types $n = |\Theta|$ and the number of buyer actions $m = |A|$ are both small. In fact, the buyer's action space is small in many applications, e.g., a binary decision \cite{liu2021optimal}. A lender either approves or rejects a loan application; a retailer either orders a new batch or does not; an advertiser either shows an advertisement to a potential consumer or withholds it. The buyer's type space is often small as well: buyers of facial data, for instance, may naturally be grouped into researchers, governmental agencies
and companies \cite{chen2022selling}; alternatively, buyer's profitability or risk tolerance may be discretized into a few levels.

In this regime, we treat $n$ and $m$ as constants, while allowing the state space $\Omega$ to be very large or even continuous, which is practical in query and model pricing, and obtain an FPTAS.

\begin{theorem} \label{thm:constant-additive-fptas}
 Fix constants \(n = |\Theta|\) and \(m = |A|\). Given accuracy \(\varepsilon \in (0, 1]\), there exists an additive FPTAS that computes a direct AF menu that has revenue at least \(\OPT - \varepsilon\) in \((1/\varepsilon)^{O(nm)}\) time, where \(\OPT\) is the optimal revenue of Program~\eqref{eq:unrestricted-program}.
\end{theorem}

The key is to eliminate the dependence on the large (even continuous) state space \(\Omega\), and to replace the infinitely many AF constraints by finitely many constraints. This addresses the two main obstacles in Program~\eqref{eq:unrestricted-program}. The construction has four steps. First, we partition \(\Omega\) into finitely many cells by grouping states with similar utility profiles. Second, we reduce the infinite family of AF constraints to a finite family by removing low price products. Third, we replace the state-dependent variables \(\pi_{\omega, a}(\theta)\) by cellwise recommendation profiles; combining the first three steps yields a finite dimensional convex program. Finally, we convert the menu computed by the convex program into an AF direct menu by removing low price products, scaling prices slightly, and applying revelation principle.

\paragraph{Step I: Partition the state space.} For each state define its utility profile
\[ u(\omega) = (u_\theta(\omega, a))_{\theta \in \Theta, a \in A} \in [0, 1]^{nm}. \]
For $\beta > 0$, partition each coordinate interval \([0, 1]\) into \(\lceil 1 / \beta\rceil\) subintervals of length at most \(\beta\), thereby partitioning \([0, 1]^{nm}\) into at most \(\lceil 1 / \beta\rceil^{nm}\) cells of \(\ell_\infty\)-diameter at most \(\beta\), thus yielding a finite measurable partition \(\mathcal G_\beta\) of \(\Omega\) such that
\begin{equation} \label{eq:rank-cell-diameter}
 \|u(\omega) - u(\omega')\|_\infty \leqslant \beta, \quad \forall \omega, \omega' \in g \in \mathcal G_\beta,
\end{equation}
where $g$ is a cell of the partition. Let $N_\beta = |\mathcal G_\beta|$. By construction,
\begin{equation} \label{eq:rank-cover-upper}
 N_\beta \leqslant \left\lceil \frac{1}{\beta} \right\rceil^{nm} \leqslant \left(1 + \frac{1}{\beta}\right)^{nm}.
\end{equation}
For each cell \(g\), type \(\theta\), and action \(a\), define
\[ M_{g, \theta, a} = \sum_{\omega \in g}\mu(\omega)u_\theta(\omega, a). \]
Thus \(M_{g, \theta, a}\) is the contribution of cell \(g\) to type \(\theta\)'s expected payoff from action \(a\). 

\paragraph{Step II: Reduce the infinite AF constraints to a finite family.} The unrestricted AF constraints range over all \(c \in \mathbb Z_+^\Theta\). To reduce this infinite family to a finite family, we choose a price threshold \(\delta\) and remove all products priced below \(\delta\). Then any bundle with more than $H_\delta = \lceil1/\delta\rceil$ items costs more than one, while the value of any information is at most one. Such a bundle cannot be profitable. Therefore, it is enough to check
\[ \mathcal C_\delta(\Theta) = \{c \in \mathbb Z_+^\Theta: \|c\|_1 \leqslant H_\delta\}. \]
Let $\mathcal Z_{H_\delta} = \left\{z = (z_{\theta, a})_{\theta, a} \in \mathbb Z_+^{\Theta \times A}: \sum_{\theta, a}z_{\theta, a} \leqslant H_\delta \right\}$. Let \(B_\delta = \sum_{c\in\mathcal C_\delta(\Theta)}|\mathcal Z_\Theta(c)|\). Then
\begin{equation} \label{eq:feature-dimension}
 B_\delta = \left|\mathcal Z_{H_\delta}\right| = \binom{nm + H_\delta}{H_\delta}.
\end{equation}
The first equality identifies \(B_\delta\) with the size of \(\mathcal Z_{H_\delta}\): specifying a pair \((c, z)\) with \(c \in \mathcal C_\delta(\Theta)\) and \(z \in \mathcal Z_\Theta(c)\) is equivalent to specifying a nonnegative integer vector \(z \in \mathbb Z_+^{\Theta \times A}\) of total mass at most \(H_\delta\), and this is exactly an element of \(\mathcal Z_{H_\delta}\). The second equality is the standard stars and bars count for nonnegative integer vectors in \(nm\) coordinates whose summation is at most \(H_\delta\).

\paragraph{Step III: Cellwise recommendation profiles and finite dimensional convex program.} After Steps I and II, the remaining dependence on $\Omega$ comes from \(\pi_{\omega, a}(\theta)\). Rewrite $x_{\theta, a}(\omega) = \pi_{\omega, a}(\theta)$, and let $x(\omega) = (x_{\theta, a}(\omega))_{\theta, a} \in \mathcal X: =\prod_{\theta \in \Theta} \Delta(A)$. Hence \(x(\omega)\) is the recommendation profile used at state \(\omega\). Besides, \eqref{eq:count-prob} can be rewritten as
\[ P_x(z \mid c) = \prod_{\theta \in \Theta}\left[\frac{c_\theta!}{\prod_{a \in A}z_{\theta, a}!}\prod_{a \in A}x_{\theta, a}^{z_{\theta, a}}\right]. \] 

Using the index set \(\mathcal Z_{H_\delta}\) from Step II, for \(x \in \mathcal X\), write \(x^z = \prod_{\theta, a}x_{\theta, a}^{z_{\theta, a}}\). Define $\phi_{H_\delta}(x) = (x^z)_{z \in \mathcal Z_{H_\delta}} \in \mathbb R^{B_\delta}$, and define
\begin{equation} \label{eq:coord-body}
 \mathcal R_{H_\delta}(\mathcal X) = \operatorname{conv}\{\phi_{H_\delta}(x): x \in \mathcal X\} \subseteq \mathbb R^{B_\delta}.
\end{equation}
Hence, for each cell \(g\), the vector \(r_g \in \mathcal R_{H_\delta}(\mathcal X)\) is a convex combination of \(\phi_{H_\delta}(x)\) for \(x \in \mathcal X\). In other words, \(r_g\) can be seen as an expected value of \(\phi_{H_\delta}(x(\omega))\) for \(\omega \in g\), and thus it can be interpreted as a recommendation profile for cell \(g\). Note that $r_g$ is independent of the states, thus we call it a cellwise recommendation profile. The convex hull \(\mathcal R_{H_\delta}(\mathcal X)\) is the set of all possible cellwise recommendation profiles.

Denote $\gamma(c, z) = \prod_{\theta \in \Theta}\frac{c_{\theta}!}{\prod_{a \in A} z_{\theta, a}!}$, so that \(P_x(z \mid c) = \gamma(c, z) x^z\). Given \(r = (r_g)_{g \in \mathcal G_\beta}\), define
\begin{equation} \label{eq:compressed-intended}
 \widehat W_\theta(r) = \sum_{g \in \mathcal G_\beta}\sum_{a \in A}M_{g, \theta, a}\,r_{g, \mathbf e_{\theta, a}},
\end{equation}
where \(\mathbf e_{\theta,a}\) is the vector with a one in the \((\theta, a)\) coordinate and zeros elsewhere. Equation~\eqref{eq:compressed-intended} is the state-free utility obtained by following the recommendation of the intended product \(\theta\). For a deviation bundle \(c \in \mathcal C_\delta(\Theta) \setminus \{e_\theta\}\), define
\begin{equation} \label{eq:compressed-deviation-action}
 \widehat V_{\theta, c, z, a}(r) = \sum_{g \in \mathcal G_\beta}M_{g, \theta, a}\, \gamma(c, z)\, r_{g, z}
\end{equation}
and
\begin{equation} \label{eq:compressed-deviation}
 \widehat V_{\theta, c}(r) = \sum_{z \in \mathcal Z_\Theta(c)}\max_{a \in A}\widehat V_{\theta, c, z, a}(r).
\end{equation}
Then \(\widehat V_{\theta, c, z, a}(r)\) is the state-free utility from first observing count vector \(z\) after buying \(c\), and then choosing the action \(a\). Equation~\eqref{eq:compressed-deviation} then takes the maximum over the action for each realized count vector \(z\) and sums over all possible count vectors \(z\), so it is the state-free utility of deviation bundle \(c\). Using the quantities above, we obtain the finite dimensional convex program
\begin{equation} \label{eq:coord-price}
 \begin{aligned}
  \operatorname{FinPrice}^{\sigma}_{\delta, \beta}: \qquad
  \max_{r, t} \quad &\sum_{\theta \in \Theta} \lambda_\theta t_\theta \\
  \text{s.t.} \quad &r_g \in \mathcal R_{H_\delta}(\mathcal X), && \forall g \in \mathcal G_\beta,\\
  & t_\theta \geqslant 0, && \forall \theta \in \Theta, \\
  & \widehat W_\theta(r) - t_\theta \geqslant \widehat V_{\theta, c}(r) - c \cdot t - \sigma, && \forall \theta \in \Theta, \ \forall c \in \mathcal C_\delta(\Theta) \setminus \{e_\theta\}.
 \end{aligned}
\end{equation}
The first line of constraints is the cellwise recommendation profile constraint. The last line is the state-free AF and obedience constraint: the intended utility of type \(\theta\) must dominate the value of every deviation bundle \(c \in \mathcal C_\delta(\Theta)\), up to slack \(\sigma\). The feasible region is convex because \(\mathcal R_{H_\delta}(\mathcal X)\) is convex by definition, \(\widehat W_\theta\) and \(\widehat V_{\theta, c, z, a}\) are linear in \(r\), and \(\widehat V_{\theta, c}\) is a sum of pointwise maxima of linear functions. The program is also finite dimensional: the variable \(r\) has \(N_\beta B_\delta\) coordinate entries, and \(t\) has \(n\) prices.

The next lemma shows the correspondence between program~\eqref{eq:coord-price} and the original AF information selling program~\eqref{eq:unrestricted-program}. 

\begin{lemma} \label{lem:coord-compression}
 Let \(\Gamma_\beta = 2\beta\).
 \begin{enumerate}[label=(\roman*)]
  \item Every exact AF direct menu induces \(r_g \in \mathcal R_{H_\delta}(\mathcal X)\) such that the same prices satisfy all constraints of \(\operatorname{FinPrice}^{\Gamma_\beta}_{\delta, \beta}\).
  \item Every feasible \((r, t)\) of \(\operatorname{FinPrice}^{\alpha}_{\delta, \beta}\) admits a finite support menu implementation that is \((\alpha + \Gamma_\beta)\)-AF against \(\mathcal C_\delta(\Theta)\). If \(\Omega\) is discrete, the menu is defined on an extension \(\Omega^+\) of \(\Omega\).
 \end{enumerate}
\end{lemma}

By Carath\'eodory's theorem \cite{rockafellar1997convex}, every feasible \(r_g \in \mathcal R_{H_\delta}(\mathcal X) \subseteq \mathbb R^{B_\delta}\) has a decomposition with \(K_g \leqslant B_\delta + 1\) components, i.e., there exist \(x_{g, \ell} \in \mathcal X\) and \(\alpha_{g, \ell} \geqslant 0\) for \(\ell = 1, \ldots, K_g\) such that
\[ r_g = \sum_{\ell = 1}^{K_g} \alpha_{g, \ell} \phi_{H_\delta}(x_{g, \ell}), \quad \sum_{\ell = 1}^{K_g} \alpha_{g, \ell} = 1. \]

This decomposition determines the finite support menu. If the state space is continuous, then for every cell \(g\) we choose a measurable partition \(g = \bigsqcup_{\ell = 1}^{K_g} g_\ell\) with \(\mu(g_\ell) = \alpha_{g, \ell}\mu(g)\), and use recommendation profile \(x_{g, \ell}\) on \(g_\ell\). Thus product \(\theta\) is the experiment \(E_\theta = (A, K_\theta)\) with \(K_\theta(a \mid \omega) = x_{g, \ell, \theta, a}\) whenever \(\omega \in g_\ell\).

If \(\Omega\) is discrete, the cell cannot in general be split with arbitrary weights. We then implement the same decomposition on an extension \(\Omega^+\) of \(\Omega\). For every \(\omega \in g\), create substates \((\omega, \ell)\), \(\ell = 1, \ldots, K_g\), and define $\Omega^+ = \{(\omega, \ell): \omega \in g \in \mathcal G_\beta,\ \ell = 1, \ldots, K_g\}$, $\mu^+(\omega, \ell) = \mu(\omega) \alpha_{g, \ell}$, with payoff \(u^+_\theta((\omega, \ell), a) = u_\theta(\omega, a)\). The buyer faces the menu \(\mathcal M^+ = \{(E^+_\theta, t_\theta)\}_{\theta \in \Theta}\) on \(\Omega^+\), where \(E^+_\theta = (A, K^+_\theta)\) and $K^+_\theta(a \mid \omega, \ell) = x_{g, \ell, \theta, a}$, $\forall \omega \in g$. Note that the index \(\ell\) is part of the realized extended state and is therefore fixed across all purchases; repeated purchases are independent conditional on \((\omega, \ell)\). Consequently, for every count vector \(c\), every \(z \in \mathcal Z_\Theta(c)\), and every action \(a\), the utility contribution of cell \(g\) is
\[
 \sum_{\omega \in g}\sum_{\ell = 1}^{K_g}\mu^+(\omega, \ell)\gamma(c, z)x_{g, \ell}^{z}u^+_\theta((\omega, \ell), a)
 = M_{g, \theta, a}\gamma(c, z)r_{g, z}.
\]
Thus the decomposition gives the menu with the same utility as in~\eqref{eq:compressed-deviation-action}.

\begin{remark}
 Although the quantities \(M_{g, \theta, a}\) are dependent on \(\Omega\), they can be computed efficiently by numerical methods to arbitrarily high precision. Hence, the precision of the numerical evaluation of \(M\) can be made negligible relative to the final \(\varepsilon\)-guarantee. In particular, if we compute \(M\) to accuracy \(\kappa / (2N_\beta)\), then every AF expression in~\eqref{eq:coord-price} changes by at most \(\kappa\). We will therefore treat the numerical evaluation of \(M\) as exact for simplicity.
\end{remark}

We now describe the separation oracle used to solve~\eqref{eq:coord-price}. First, consider an AF constraint for a fixed pair \((\theta, c)\). Given \((r, t)\), compute \(a_z^* \in \arg\max_a \widehat V_{\theta, c, z, a}(r)\) for every count vector \(z\). If $\widehat W_\theta(r) - t_\theta < \sum_{z \in \mathcal Z_\Theta(c)} \widehat V_{\theta, c, z, a_z^*}(r) - c \cdot t - \sigma$, then the linear inequality
\begin{equation} \label{eq:benders-cut}
 \widehat W_\theta(r) - t_\theta \geqslant \sum_{z \in \mathcal Z_\Theta(c)} \widehat V_{\theta, c, z, a_z^*}(r) - c \cdot t - \sigma
\end{equation}
is valid for every feasible solution and is violated by the current point. 

Second, we need separation for \(\mathcal R_{H_\delta}(\mathcal X)\). Separating a candidate point \(\bar r\) from \(\mathcal R_{H_\delta}(\mathcal X)\) reduces to evaluating the support function of \(\mathcal R_{H_\delta}(\mathcal X)\): for any direction \(q \in \mathbb R^{B_\delta}\), we compare \(q \cdot \bar r\) with \(\max_{r \in \mathcal R_{H_\delta}(\mathcal X)} q \cdot r\). Because \(\mathcal R_{H_\delta}(\mathcal X)\) is the convex hull of the vectors \(\phi_{H_\delta}(x)\), linear optimization over this body is attained at an extreme point, hence
\begin{equation} \label{eq:coord-body-optimization}
 \max_{r \in \mathcal R_{H_\delta}(\mathcal X)} q \cdot r = \max_{x \in \mathcal X} q \cdot \phi_{H_\delta}(x) = \max_{x \in \mathcal X}\sum_{z \in \mathcal Z_{H_\delta}} q_z x^z .
\end{equation}
Equation~\eqref{eq:coord-body-optimization} is a degree-\(H_\delta\) polynomial optimization problem over the product of simplexes \(\mathcal X\), and can be handled by standard methods from real algebraic geometry \cite{basu2006algorithms}. Combined with the standard equivalence between weak optimization and weak separation \cite{grotschel1988geometric}, this yields the required separation oracle for \(\mathcal R_{H_\delta}(\mathcal X)\).

\begin{proposition} \label{prop:coord-price}
 Fix \(\delta, \beta, \sigma, \tau > 0\). If \(\sigma \geqslant \Gamma_\beta\), then one can compute a finite support menu with revenue at least \(\OPT - \tau\) that is \((\sigma + \tau + \Gamma_\beta)\)-AF against \(\mathcal C_\delta(\Theta)\).
 And the time complexity is $\operatorname{poly}\left(N_\beta B_\delta, \log\frac1\tau\right) \cdot \left(H_\delta\right)^{O(nm)}$.
\end{proposition}

\paragraph{Step IV: Remove low price products and repair AF menu.} Step III returns a menu that is approximately AF against \(\mathcal C_\delta(\Theta)\). We then turn it into an exact AF direct menu.

\begin{lemma} \label{lem:threshold-low-prices}
 Let $\{E_\theta, t_\theta\}_{\theta \in \Theta}$ be a menu that is \(\alpha\)-AF against \(\mathcal C_\delta(\Theta)\), and define $\Theta_\delta = \{\theta \in \Theta: t_\theta \geqslant \delta\}$. Remove all products outside \(\Theta_\delta\). The remaining menu has revenue at least \(\sum_\theta\lambda_\theta t_\theta - \delta\) and is \(\alpha\)-AF.
\end{lemma}

After low price removal, every retained product has price at least \(\delta\). We will then scale every retained price to \((1 - \eta)t_\theta\), and apply Proposition~\ref{prop:bundle-revelation}, and the resulting direct menu is exactly AF. The next lemma shows that the revenue loss from scaling is at most \(\alpha / \eta\).

\begin{lemma} \label{lem:exact-repair}
 Let $\{E_\theta, t_\theta\}_{\theta \in \Theta}$ be a menu on a retained product set \(\widehat\Theta \subseteq \Theta\). Assume that the menu is \(\alpha\)-AF against all finite bundles of retained products. Fix \(\eta \in (0, 1)\), scale $t'_\theta = (1 - \eta) t_\theta, \ \forall \theta \in \widehat\Theta$. Then applying Proposition~\ref{prop:bundle-revelation} yields an exactly AF direct menu with revenue $(1 - \eta)\left(\sum_{\theta \in \widehat\Theta} \lambda_\theta t_\theta - \frac{\alpha}{\eta}\right)$.
\end{lemma}

Finally, Algorithm~\ref{alg:af-fptas} summarizes the FPTAS, and we prove Theorem~\ref{thm:constant-additive-fptas}.

\begin{algorithm}[htbp]
 \caption{\textsc{AF-FPTAS}$(\varepsilon)$} \label{alg:af-fptas}
 \begin{algorithmic}[1]
  \State Set \(\delta \gets \varepsilon/5\), \(\eta \gets 2\varepsilon/5\), and \(\bar\alpha \gets 7\varepsilon^2/50\).
  \State Set \(\beta \gets 5\bar\alpha/21\), \(\sigma \gets 10\bar\alpha/21\), and \(\tau \gets \bar\alpha/21\).
  \State Set \(H_\delta \gets \lceil1/\delta\rceil\), construct \(\mathcal G_\beta\), \(\mathcal C_\delta(\Theta)\), and the count spaces \(\mathcal Z_\Theta(c)\).
  \State Compute the cell payoff aggregates \(M_{g, \theta, a}\) for all \(g \in \mathcal G_\beta\), \(\theta \in \Theta\), and \(a \in A\).
  \State Solve \(\operatorname{FinPrice}^{\sigma}_{\delta, \beta}\) to additive objective tolerance \(\tau\), and recover for every cell \(g\) a decomposition \(r_g = \sum_{\ell = 1}^{K_g}\alpha_{g, \ell}\phi_{H_\delta}(x_{g, \ell})\). Interpret these decompositions as a finite support menu and let \(\widehat t\) be the recovered price vector.
  \State Let \(\widehat\Theta = \{\theta: \widehat t_\theta \geqslant \delta\}\), and discard every product outside \(\widehat\Theta\).
  \State Scale every retained primitive price to \((1 - \eta)\widehat t_\theta\).
  \State For every type \(\theta \in \Theta\), enumerate all count vectors \(c \in \mathbb Z_+^{\widehat\Theta}\) with \(\|c\|_1 \leqslant \lceil1/((1 - \eta)\delta)\rceil\), together with \(c = 0\), and choose a bundle that maximizes utility. Apply Proposition~\ref{prop:bundle-revelation} to the chosen bundles and return the resulting direct menu.
 \end{algorithmic}
\end{algorithm}

\begin{proof}[Proof of Theorem~\ref{thm:constant-additive-fptas}]
 Recall the parameter choices in Algorithm~\ref{alg:af-fptas}. Since \(\Gamma_\beta = 2\beta\), we have \(\Gamma_\beta = \sigma = 10\bar\alpha/21\), hence \(\sigma + \tau + \Gamma_\beta = \bar\alpha\). Since \(\sigma \geqslant \Gamma_\beta\), Proposition~\ref{prop:coord-price} returns a finite support menu with revenue at least \(\OPT - \tau\) and price vector \(\widehat t\) that is \(\bar\alpha\)-AF against \(\mathcal C_\delta(\Theta)\). Let \(\widehat\Theta = \{\theta : \widehat t_\theta \geqslant \delta\}\) as in Algorithm~\ref{alg:af-fptas}. By Lemma~\ref{lem:threshold-low-prices}, after discarding products outside \(\widehat\Theta\) the retained primitive menu has revenue at least $\OPT - \tau - \delta$, and is \(\bar\alpha\)-AF. Applying Lemma~\ref{lem:exact-repair} with \(\alpha = \bar\alpha\) yields an exactly AF direct menu with revenue at least $(1 - \eta)\left(\OPT - \tau - \delta - \frac{\bar\alpha}{\eta}\right)$. Since \(\OPT \leqslant 1\), $\OPT - (1 - \eta)\left(\OPT - \tau - \delta - \frac{\bar\alpha}{\eta}\right) = \eta\OPT + (1 - \eta)\left(\tau + \delta + \frac{\bar\alpha}{\eta}\right) \leqslant \eta + \tau + \delta + \frac{\bar\alpha}{\eta}$. Substituting the parameter values gives $\eta + \tau + \delta + \frac{\bar\alpha}{\eta} = \frac{2\varepsilon}{5} + \frac{\varepsilon^2}{150} + \frac{\varepsilon}{5} + \frac{7\varepsilon}{20} < \varepsilon$, $\forall \varepsilon \in (0, 1]$. Hence the returned direct menu has revenue at least \(\OPT - \varepsilon\).
 
 It remains to bound the running time. Since \(\delta = \Theta(\varepsilon)\), we have $H_\delta = \left\lceil\frac{1}{\delta}\right\rceil = O(1 / \varepsilon)$. Since \(n\) and \(m\) are constants, $B_\delta = \binom{nm + H_\delta}{H_\delta} = (1 / \varepsilon)^{O(nm)}$. Moreover, since \(\beta = \Theta(\varepsilon^2)\), \eqref{eq:rank-cover-upper} yields $N_\beta \leqslant \left(1 + 1 / \beta\right)^{nm} = (1 / \varepsilon)^{O(nm)}$. Hence, the time complexity in Proposition~\ref{prop:coord-price} is $(1 / \varepsilon)^{O(nm)}$. For Step~IV, any bundle has size at most $\left\lceil 1 / ((1 - \eta)\delta)\right\rceil = O(1 / \varepsilon)$. Thus the relevant bundles are indexed by count vectors \(c \in \mathbb Z_+^{\widehat\Theta}\) with \(\|c\|_1 \leqslant O(1 / \varepsilon)\). By the stars and bars count, the number of such vectors is $\binom{|\widehat\Theta| + O(1 / \varepsilon)}{O(1 / \varepsilon)}$. Since \(|\widehat\Theta| \leqslant n\), this is \((1 / \varepsilon)^{O(n)}\), hence also \((1 / \varepsilon)^{O(nm)}\). This proves the theorem.
\end{proof}

We also empirically validate the algorithm on realistic synthetic data trading scenarios in Appendix~\ref{app:experiments}. Across a broad range of practical parameter settings, the algorithm exhibits solid computational efficiency and competitive revenue performance.

\subsection{Approximation algorithms with restricted experiment types} \label{subsec:simple-approx}

The FPTAS in Section~\ref{subsec:grid-price} has running time exponential in the parameters \(m=|A|\) and \(n=|\Theta|\), and therefore may still be computationally infeasible when these parameters are large. We therefore study simpler approximation algorithms in this subsection.

A natural idea is to sell only the complete information experiment $F$ (the experiment that reveals the true state in every realized state) with a single posted price. For each type define $v_\theta = V_\theta(F)$, and assume that at least one \(v_\theta\) is positive. Let $v_{\min} = \min\{v_\theta:v_\theta>0\}$ and $v_{\max} = \max_{\theta \in \Theta} v_\theta$. Then seller maximizes revenue by posting the best price \(p\) and yielding revenue
\[ \Rfull = \max_{p \geqslant 0} p\sum_{\theta: v_\theta \geqslant p}\lambda_\theta. \]

Let $k = \min\{n, 1 + \ln(v_{\max}/v_{\min})\}$. The next proposition shows that computing \(\Rfull\) gives a \(k\)-approximate algorithm, and that this approximation factor is tight.
\begin{proposition} \label{prop:full-info-approx}
 If \(v_{\min}>0\), computing \(\Rfull\) gives a \(k\)-approximate algorithm, and the approximation ratio is tight. Moreover, a menu attaining \(\Rfull\) can be computed in time \(O(|\Theta|\log|\Theta|)\) if the values \(v_\theta\) are known.
\end{proposition}

Note that when the state space is finite, the values \(v_\theta\) can be computed in time polynomial in \(|\Omega|\). When the state space is infinite, the values \(v_\theta\) can be estimated to arbitrary accuracy by sampling. Hence, selling only complete information yields a simple algorithm with a provable approximation guarantee and efficient running time. Based on the result, selling complete information is most appropriate for markets with a small number of buyer types, or for data products whose values are comparable across buyer types.

% A natural question is whether there exists a simple algorithm that gives better approximation ratio than selling only complete information. The simplification principle is to restrict the experiments in the menu so that the AF constraints in Program~\eqref{eq:unrestricted-program} collapse to a much smaller family. In Appendix~\ref{sec:more-simple-approx}, we present two ways to simplify the AF constraints. However, neither can beat the \(k\)-approximation lower bound in Proposition~\ref{prop:full-info-approx}. This result indicates that the difficulty induced by the infinite AF constraints in Program~\eqref{eq:unrestricted-program} is difficult to address by restricting the types of experiments in the menu.

\section{Blackwell arbitrage-free under Threshold Utilities} \label{sec:blackwell}

Section~\ref{sec:program} studies the general AF menu design problem. However, the AF constraints are difficult to analyze. In this section, we focus on a special case in which the AF constraints can be reduced to the more intuitive and tractable Blackwell arbitrage-free condition in some realistic settings.

\subsection{Threshold utilities collapse AF to Blackwell AF} \label{subsec:threshold}

\begin{definition}[threshold utility] \label{def:threshold-utility}
 A buyer type \(\theta\) has \emph{threshold utility} if there exist a target experiment \(T_\theta\) and a value \(v_\theta > 0\) such that
 \[ V_\theta(E) = \begin{cases}
 v_\theta, & E \succeqB T_\theta,\\
 0, & E \not\succeqB T_\theta.
 \end{cases} \]
\end{definition}

Threshold utilities model buyers who value a product only when it meets a concrete operational threshold. It is a natural and common utility structure in many applications. A natural question is whether threshold utilities can be explicitly realized in a decision problem. Proposition~\ref{prop:certified-deployment-threshold} gives a positive answer for finite deterministic queries and a negative answer for unrestricted continuous experiment families. The intuition behind the negative answer is that in the continuous setting, an experiment can approximate a target experiment arbitrarily well without Blackwell dominating it, while the value of every finite action Bayesian decision problem changes continuously along such an approximation.

Despite the negative answer, threshold utilities are still a useful abstraction for many applications, thus we study them in this section. For a dataset example, a lender may need both income and debt variables for every applicant in a target pool before a screening rule is usable: if one required field is missing, the workflow cannot be deployed. For a model example, a firm may buy a prediction API or an LLM service only if the product reaches a minimum accuracy or safety benchmark needed for deployment. In all of these cases, the buyer values information that is sufficient for the task, and treats information that is still insufficient as simply not useful for that task.

Based on the definition of threshold utility, we then show that the AF constraints for any type with threshold utility are equivalent to Blackwell AF defined below.

\begin{definition}[Blackwell AF]
 Let \(\mathcal E\) be a family of experiments closed under garbling and composition. A price function \(p:\mathcal E\to\mathbb R_+\) is \textit{Blackwell arbitrage-free (Blackwell AF)} if there do not exist experiments \(E_1, \ldots, E_k, F \in \mathcal E\) such that $E_1 \otimes \cdots \otimes E_k \succeqB F$ and $\sum_{r = 1}^k p(E_r) < p(F)$.
\end{definition}

The next proposition gives a characterization of Blackwell AF that is analogous to the characterizations of query and model AF in Section~\ref{subsec:arbitrage}. The proof is trivial and is deferred to Appendix~\ref{app:omitted-proofs-sec5}.

\begin{proposition} \label{prop:unified}
 A price function \(p\) is Blackwell AF if and only if it satisfies:
 \begin{enumerate}
 \item \textbf{Blackwell monotonicity}: if \(E \succeqB F\), then \(p(E) \geqslant p(F)\);
 \item \textbf{parallel subadditivity}: for every \(E_1, E_2\), \(p(E_1 \otimes E_2)\leqslant p(E_1) + p(E_2)\).
 \end{enumerate}
\end{proposition}

The next theorem is the main result of this section. It says that for any type with threshold utility, the AF constraints are exactly the Blackwell AF constraints.

\begin{theorem} \label{thm:threshold-blackwell}
 Fix a direct menu \(M=\{(E_\theta, t_\theta)\}_{\theta\in\Theta}\), and suppose every type \(\theta\) has threshold utility in the sense of Definition~\ref{def:threshold-utility} where $E_\theta = T_\theta$. Assume also that the assigned experiment is individually rational, so \(t_\theta\leqslant v_\theta\). Then for each type \(\theta\), the following are equivalent:
 \begin{enumerate}
  \item \(\theta\) has no profitable bundle deviation from its assigned experiment;
  \item every bundle \(B\) satisfying \(E_B \succeqB E_\theta\) has total price at least \(t_\theta\).
 \end{enumerate}
\end{theorem}

This theorem shows that, under threshold utilities, which are common in practice, the AF constraints can be reduced to the more intuitive Blackwell AF condition. Outside the threshold utility regime, AF is a strictly stronger constraint than Blackwell AF. The following example illustrates this point.

\begin{proposition} \label{prop:utility-specific-gap}
 There is an information selling instance in which the optimal revenue achieved by a Blackwell AF mechanism is strictly greater than the optimal revenue achieved by an AF mechanism.
\end{proposition}

\subsection{Query and model AF as special cases of Blackwell AF}

Next, we show that the AF definitions for query and model pricing in the previous literature (Section~\ref{subsec:arbitrage}) are both special cases of Blackwell AF. On the one hand, this establishes Blackwell AF as a unifying framework; on the other hand, it indicates that the query and model AF notions in the existing literature do not necessarily align with general AF.

We first show that query AF is exactly the restriction of Blackwell AF to deterministic queries.

\begin{proposition} \label{prop:query-blackwell}
 For deterministic queries \(Q_1, Q_2, Q\):
 \begin{enumerate}
  \item $E_{Q_1} \succeqB E_{Q_2} \Longleftrightarrow P_{Q_1} \succeq P_{Q_2}$, where \(P_{Q_1} \succeq P_{Q_2}\) means that \(P_{Q_1}\) refines \(P_{Q_2}\);
  \item $E_Q \simB E_{Q_1}\otimes E_{Q_2} \Longleftrightarrow P_Q = P_{Q_1} \vee P_{Q_2}$.
 \end{enumerate}
\end{proposition}

Consequently, Proposition~\ref{prop:unified} restricted to deterministic queries reproduces exactly the query AF condition.

We then consider model AF. For \(J\in\mathbb S_{++}^d\), let \(G_J\) denote the Gaussian experiment \(Y = \theta + \varepsilon, \varepsilon\sim \mathcal N(0,J^{-1})\). Note that this setting generalizes the setting of \cite{chen2019towards}. If we restrict $J$ to \(J = dx I_d\), where \(x\) is the scalar precision, then we recover the AF definition of \cite{chen2019towards}.

\begin{proposition} \label{prop:gaussian-matrix}
 For Gaussian location experiments:
 \begin{enumerate}
  \item \(G_{J_1}\succeqB G_{J_2}\) if and only if \(J_1\succeq J_2\) in the Loewner order (i.e., \(J_1 - J_2\) is positive semidefinite);
  \item \(G_{J_1}\otimes G_{J_2}\simB G_{J_1+J_2}\).
 \end{enumerate}
\end{proposition}

Hence, a price rule on this Gaussian family is Blackwell AF if and only if \(\bar p(J) := p(G_J)\) is monotone in the Loewner order and subadditive under matrix addition.

\paragraph{Information-theoretic interpretation.} The mutual information between the underlying state and the signal of an experiment is a natural measure of informativeness. The next proposition shows that mutual information pricing is Blackwell AF, and hence is a valid AF pricing rule under threshold utilities.

\begin{proposition} \label{prop:mi-noarb}
 Fix a prior \(\mu\in\Delta(\Omega)\), and let \(X \sim \mu\) denote the random state. For any experiment \(E\), let \(Y_E\) denote the random signal generated by \(E\) conditional on \(X\), and define $p_\mu(E) = I(X; Y_E)$. Then \(p_\mu\) is Blackwell monotone and subadditive under independent composition.
\end{proposition}

Corollary~\ref{cor:mi-query-gaussian} gives the explicit form of mutual information pricing for query and model.

\begin{corollary} \label{cor:mi-query-gaussian}
 Under the price rule \(p_\mu(E)=I(X;Y_E)\) from Proposition~\ref{prop:mi-noarb}:
 \begin{enumerate}
  \item if the random state is the database random variable \(D \sim \mu\) and \(Q\) is a deterministic query, then $p_\mu(E_Q) = I(D; Q(D)) = H(Q(D))$, where \(H(Q(D))\) is the Shannon entropy of the query answer under the prior \(\mu\).
  \item if the random state is \(\theta \sim \mathcal N(0, \Sigma_0)\), then for every Gaussian experiment \(G_J\), $p_\mu(G_J) = \frac12 \log\det(I + \Sigma_0 J)$.
 \end{enumerate}
\end{corollary}

Note that $p_\mu(E_Q)$ is consistent with the Shannon entropy pricing discussed in \citep{deep2017design}. Moreover, the map $J \mapsto \frac12 \log\det(I + \Sigma_0 J)$ is monotone in the Loewner order and subadditive under matrix addition.

\subsection{Pricing under structured menus} \label{subsec:menu-restrictions}

We now study the optimal pricing problem under threshold utilities. We first consider the general problem, and then study the problem under some realistic and tractable structured menus.

We consider a general setting where the target experiments of different types might be the same. Let target products be indexed by $i \in [M]$, type $\theta \in \Theta$ buyer has target experiment $\tau(\theta) \in [M]$, and threshold value $v_\theta > 0$. Let $\Theta_i = \{\theta \in \Theta: \tau(\theta) = i\}$ and define $D_i(x) = \sum_{\theta \in \Theta_i: \ v_\theta \geqslant x}\lambda_\theta$ as the demand function for target $i$. For a count vector $c \in \mathbb Z_+^M$, write $T^c = \bigotimes_{j = 1}^M T_j^{\otimes c_j}$, where $T_j^{\otimes 0}$ is the uninformative experiment. For each target $i$, define the set of count vectors that can synthesize it: $\mathcal C_i = \{c \in \mathbb Z_+^M: T^c \succeqB T_i\}$.

The following proposition shows that the optimal revenue under threshold utilities can be achieved by a direct menu that posts one price for each target experiment. Note that this is not in conflict with individual rationality, since for type $\theta$ with price $p_{\tau(\theta)} > v_\theta$, we can equivalently assign the uninformative experiment and zero price to this type.

\begin{proposition} \label{prop:threshold-target-lp}
 For threshold types $(\tau(\theta), v_\theta)_{\theta \in \Theta}$, the optimal revenue equals the value of
 \begin{equation} \label{eq:threshold-target-demand-program}
  \begin{aligned}
   \max_{p \geqslant 0} \quad & \sum_{i = 1}^M p_i D_i(p_i) \\
   \text{s.t.} \quad & p_i \leqslant c \cdot p, && \forall i \in[M], \ \forall c \in \mathcal C_i .
  \end{aligned}
 \end{equation}
\end{proposition}

However, Program~\eqref{eq:threshold-target-demand-program} is APX-hard. The proof idea is similar to the proof of Theorem~\ref{thm:af-revenue-hardness} (i), and we defer the discussion to Proposition~\ref{prop:threshold-target-lp-hardness}. Intuitively, the main difficulty is that we need to choose which buyer types receive nonempty information, thus the problem is naturally combinatorial. The hardness results rule out a generic algorithm for threshold types. We therefore turn to structured target families, for which the problem admits efficient exact or approximate solutions.

\paragraph{Full surplus extraction.} The first structure guarantees full buyer surplus extraction, thus we can compute the optimal revenue without choosing which buyer types receive nonempty information. Let \(R_{tbs} = \sum_{\theta \in \Theta} \lambda_\theta v_\theta\) denote the total buyer surplus under threshold utilities, which is an upper bound on revenue for any individually rational mechanism. Extracting \(R_{tbs}\) is possible under special structures of the target experiments and the target values. We give two examples of such structures below.
\begin{enumerate}
 \item The target experiments are independent and the threshold value of the each target experiment is unique. Formally, for every target \(T_i\), no finite bundle of the other targets \(\{T_j: j \neq i\}\) Blackwell dominates \(T_i\). Hence any bundle that obtains \(T_i\) must include \(T_i\) itself, so charging \(p_i= v_i\) creates no arbitrage. In real applications, this structure arises when the target experiments are for different regions or different tasks. For example, a lender may need to buy a dataset for each region, and the datasets are independent because the applicants in different regions are disjoint.
 
 \item The value functions are normalized, monotone, and subadditive. Formally, if each target experiment $T_S$ gives a subset $S$ of $U$, and the value is \(v_S = r(S)\), where \(r: 2^U \to \mathbb R_+\) satisfies \(r(\emptyset) = 0\), \(A \subseteq B \Rightarrow r(A) \leqslant r(B)\), and \(r(A \cup B) \leqslant r(A) + r(B), \ \forall A, B \subseteq U\). If a finite bundle \(T_{S_1}, \ldots, T_{S_k}\) Blackwell dominates \(T_S\), then \(S\subseteq \bigcup_{\ell = 1}^k S_\ell\), thus \(r(S) \leqslant r(\bigcup_{\ell = 1}^k S_\ell) \leqslant \sum_{\ell = 1}^k r(S_\ell)\). Therefore charging \(p_S = v_S\) is Blackwell AF. In real applications: (i) $U$ is the complete database, and $T_S$ is the experiment that reveals the row subset $S$; (ii) $U = \{[n]: n \in \mathbb N\}$ is the set of numbers of API calls, and $T_S$ is the experiment that provide all numbers of data API calls in $S$. In these two applications, the value functions are naturally subadditive due to the law of diminishing marginal utility.
\end{enumerate}

\paragraph{Polynomial-time Blackwell AF violation detection.} The following proposition shows that if we can efficiently detect Blackwell AF violations, i.e., given a price vector \(p\) and a target \(i\), we can efficiently check whether there exists a count vector \(c \in \mathcal C_i\) such that \(c \cdot p < p_i\), then we have a polynomial time approximation algorithm.

\begin{proposition} \label{prop:arbitrage-detection-log-approx}
 There is no general polynomial time algorithm for detecting a Blackwell AF violation. However, if a Blackwell AF violation can be detected in polynomial time, assume all positive threshold values $v_\theta$ lie in $\{1, 2, \ldots, V\}$, then Program~\eqref{eq:threshold-target-demand-program} admits a polynomial-time \((1 + \lfloor\log_2 V\rfloor)\)-approximation.
\end{proposition}

The assumption that the threshold values are integers is without loss of generality, since the values can be scaled without changing the approximation ratio, and is naturally discrete when measured in currency units. 

We next give two realistic data product families where the detection problem is polynomial time solvable. The first is the \textbf{interval query}. For example, a lender may need to buy a dataset that contains all applicants in a certain income range. Formally, let every target query be an interval $I = [a, b] \subseteq \mathbb{R}$. The experiment $T_I$ reveals the fields in $I$. A bundle synthesizes $T_I$ if and only if the union of purchased intervals covers $I$. The following proposition shows that the Blackwell AF violation detection problem can be solved by a shortest path algorithm, thus giving a polynomial time approximation algorithm for interval menus.

\begin{proposition} \label{prop:interval-shortest-path}
  For an interval menu, whether a Blackwell AF violation exists can be solved by a shortest path algorithm in time $O(M^2\log M)$.
\end{proposition}

The second family is a \textbf{fixed-dimensional quality-based menu}. For example, a model service may be graded by reasoning, coding, vision, and safety quality. Formally, let every target product be a quality vector $x \in \mathcal L$, where $\mathcal L$ is the set of quality vectors. The experiment $T_x$ reveals the data product with quality $x$. Higher coordinates mean a more informative product, so $T_x \succeqB T_y$ exactly when $x \geqslant y$ coordinatewise, and \( T_x \otimes T_y \simB T_{x \vee y}, \) where $x \vee y$ is the coordinatewise maximum. The following proposition shows that the Blackwell AF violation detection problem can also be solved by a shortest path algorithm, thus giving a polynomial time approximation algorithm for fixed-dimensional quality-based menus.

\begin{proposition} \label{prop:quality-based-detection}
  For fixed dimension $d$, whether arbitrage exists in a quality-based menu can be solved by a shortest path algorithm in time $O(M^{d+1}\log M)$.
\end{proposition}

\paragraph{Structures with dynamic programming solutions.} To choose buyer types to receive nonempty information efficiently, a natural idea is to use dynamic programming under some hierarchical structure of the target experiments. The first structure is \textbf{tree-structured dataset menus}. In real applications, the demand of data buyers can be organized as a hierarchy such as national, regional, city, and neighborhood data. We can model such menus by trees.

Let \(U\) be a finite set of elementary fields (e.g., all rows in a database), and the products form a tree-structured family \(\mathcal F \subseteq 2^U\), represented as a rooted tree with node \(u\) corresponding to \(S_u \in \mathcal F\). The children \(\mathrm{Ch}(u)\) are disjoint subsets whose union is \(S_u\), and \(T_u\) reveals exactly the fields in \(S_u\). The following proposition (and its proof) gives a natural dynamic programming algorithm for tree-structured dataset menus.

\begin{proposition} \label{prop:tree-dp}
  For a tree-structured dataset menu, suppose all buyer values lie in $\{1, \ldots, V\}$, then Program~\eqref{eq:threshold-target-demand-program} can be computed in time $O(|\mathcal F|V^3)$ by dynamic programming.
\end{proposition}

The second structure is \textbf{data products with Gaussian noise}. In real applications, a data product may have privacy restrictions, and the seller may add Gaussian noise to ensure differential privacy, or the seller may add Gaussian noise to a model's predictions to create different accuracy levels. In these cases, the products can be modeled as Gaussian experiments with different precision matrices. Formally, a product $G_{J}$ with precision matrix $J \in \PD^d$ releases $Y = \theta^\star + \varepsilon$, where $\theta^\star \in \mathbb R^d$ is the original data and $\varepsilon \sim \mathcal N(0, J^{-1})$. The setting of \citet{chen2019towards} is the special case $J = dxI_d$, where $x$ is the scalar precision.

Consider a finite menu $G_{J_1},\ldots,G_{J_M}$, where each $J_i \in \PD^d$. By Proposition~\ref{prop:gaussian-matrix}, a price vector $p \in \mathbb R_+^M$ is Blackwell AF if and only if each target's posted price must equal the cheapest bundle cost that reaches its precision:
\begin{equation} \label{eq:matrix-af-counts}
  p_i = \min_{\substack{c \in \mathbb Z_+^M: \sum_{j = 1}^M c_j J_j \succeq J_i}} \sum_{j = 1}^M c_jp_j, \ \forall i \in [M].
\end{equation}

Solving the matrix variable problem is difficult since the Loewner order is only a partial order. To handle this, we reduce matrix precision to a one-dimensional measure of accuracy. Fix a positive definite weight matrix $W\in\PD^d$, and define the weighted accuracy of product $i$ by $s_i = \operatorname{tr}(WJ_i)>0$. Intuitively, $W$ says which directions matter when we summarize a matrix precision by a single number. After sorting the products so that $s_1 \leqslant \cdots \leqslant s_M$, we call a price vector $p$ \emph{$W$-regular} if $0 \leqslant p_1 \leqslant \cdots \leqslant p_M$ and $p_1 / s_1 \geqslant p_2 / s_2 \geqslant \cdots \geqslant p_M / s_M$. 

\begin{proposition} \label{prop:spectral-gauge-af}
  For every \(W\in\PD^d\), every \(W\)-regular price vector satisfies~\eqref{eq:matrix-af-counts}. Moreover, the revenue maximizing \(W\)-regular price vector \(p^W\) can be computed in polynomial time by dynamic programming.
\end{proposition}

The intuition behind $W$-regularity is that it transforms the constraint Equation~\eqref{eq:matrix-af-counts} into a one-dimensional problem, which can be naturally solved by dynamic programming. However, the best $W$-regular price vector may not be optimal for the original matrix problem. For two targets $i$ and $k$, define $\kappa_{ik}:=\left\lceil \lambda_{\max}\!\left(J_k^{-1 / 2}J_iJ_k^{-1 / 2}\right)\right\rceil$, meaning that $\kappa_{ik}$ copies of product $k$ are enough to synthesize product $i$. For a fixed $W$, define $\Gamma_W = \max_{i, k} \left(\kappa_{ik} \min\{1, s_k / s_i\}\right)$. This number measures how well the weighted accuracies $s_i$ capture the true matrix synthesis relations.

\begin{proposition} \label{prop:spectral-gauge-approx}
  If $\OPT$ is the optimal revenue over all price vectors satisfying~\eqref{eq:matrix-af-counts}, then $p^W$ given by Proposition~\ref{prop:spectral-gauge-af} achieves revenue at least $\OPT / \Gamma_W$.
\end{proposition}

This result generalizes the setting of \citet{chen2019towards}, where $J_i = dx_iI_d$. Take $W = I_d$, then $s_i = d^2x_i$, $\kappa_{ik} = \lceil x_i / x_k\rceil$, and hence $\Gamma_W \leqslant 2$, Proposition~\ref{prop:spectral-gauge-approx} recovers the $2$-approximation in \cite{chen2019towards}. In general matrix setting, the approximation ratio is governed by $\Gamma_W$. Thus the seller should choose appropriate weight matrix $W$ so that the weighted accuracies $s_i = \operatorname{tr}(WJ_i)$ reflect the true synthesis relations well.

\section{Conclusion and Future Work}

We study optimal data pricing when buyers value data through Bayesian decision making and may arbitrage by combining multiple purchases. Modeling data products as statistical experiments lets us unify query pricing, model pricing, and information selling. We show that general AF information selling problem is computationally hard, but admits a practical additive FPTAS when the numbers of types and actions are small constants. We also identify threshold utilities as a tractable regime in which AF reduces to Blackwell dominance, recovering existing query- and model-pricing conditions and enabling efficient algorithms for structured menus.

Several directions are worth exploring. First, it would be natural to extend the framework to interactive information selling mechanisms in \cite{babaioff2012optimal,chen2020selling}. Second, an important algorithmic direction is to develop approximation schemes for richer structured menus (e.g., when buyers need high-dimensional interval queries or overlapping datasets) and richer utility classes.

\bibliographystyle{plainnat}
\bibliography{reference}

\appendix

\section{Omitted Proofs in Section~\ref{sec:preliminary}} \label{app:omitted-proofs-sec2}

\begin{proof}[Proof of Proposition~\ref{prop:bundle-revelation}]
 Let the original menu be $\mathcal M = \{(E_j, t_j)\}_{j = 1}^N$. For each
 type $\theta$, let $B_\theta = (i_1^\theta, \ldots, i_{\ell_\theta}^\theta)$ be the bundle selected by type $\theta$ from $\mathcal M$. Thus $B_\theta \in \arg\max_B\{V_\theta(E_B) - t_B\}$. If several bundles maximize the buyer's utility, the buyer selects one that maximizes the seller's revenue. Define a direct menu $\mathcal M'$ as follows:
 \[ \mathcal M' = \{(E'_\theta, t'_\theta)\}_{\theta \in \Theta}, \quad E'_\theta := E_{B_\theta}, \quad t'_\theta := t_{B_\theta}. \]
 When type $\theta$ chooses its intended entry in $\mathcal M'$, it receives
 exactly the composite experiment and makes exactly the total payment that it
 selected in $\mathcal M$. Consequently, the intended choices in $\mathcal M'$ generate the same expected seller revenue as the original menu.

 It remains to prove that $\mathcal M'$ is AF. This is true since any bundle deviation in $\mathcal M'$ corresponds to a bundle deviation in $\mathcal M$, and the intended entry is already optimal against all deviations in $\mathcal{M}$. Thus $\mathcal M'$ is AF, and the proof is complete.
\end{proof}

\begin{proof}[Proof of Proposition~\ref{prop:obedience-revelation}]
 Fix an AF direct menu \(\mathcal M = \{(E_\theta, t_\theta)\}_{\theta \in \Theta}\), and write \(E_\theta = (S_\theta, K_\theta)\). For each type \(\theta\), fix a tie-breaking rule so that after every signal \(s \in S_\theta\), type \(\theta\) chooses a single action $a_\theta(s) \in \arg\max_{a \in A}\sum_{\omega \in \Omega}\mu(\omega)K_\theta(s \mid \omega)u_\theta(\omega, a)$. For every action \(a \in A\), let $S_a^\theta := \{s \in S_\theta: a_\theta(s) = a\}$. Thus the sets \(\{S_a^\theta\}_{a\in A}\) form a partition of \(S_\theta\). Following the construction in Proposition~1 of \cite{bergemann2018design}, define a new experiment \(E'_\theta=(S'_\theta,K'_\theta)\) for type \(\theta\) by taking
 \[ S'_\theta:=\{a\in A:S_a^\theta\neq\varnothing\}, \quad K'_\theta(a\mid\omega) := \sum_{s \in S_a^\theta} K_\theta(s \mid \omega), \forall a \in S'_\theta, \omega \in \Omega. \]
 Thus the new signal is simply the recommended action label. Equivalently, define the deterministic garbling kernel \(\Gamma_\theta:S_\theta\to\Delta(S'_\theta)\) by $\Gamma_\theta(a\mid s):=\mathbf 1\{s\in S_a^\theta\}$. Then, for every \(a\in S'_\theta\) and \(\omega\in\Omega\),
 \[
 K'_\theta(a\mid\omega)=\sum_{s\in S_\theta}\Gamma_\theta(a\mid s)K_\theta(s\mid\omega),
 \]
 so \(E_\theta\succeqB E'_\theta\). Then we need to show that the new menu \(\mathcal M'\) is AF. As \cite{bergemann2018design} show, $\mathcal{M'}$ preserves the intended type's value. So it suffices to show that any bundle deviation in $\mathcal{M'}$ is a garbling of a bundle deviation in $\mathcal{M}$, which is already suboptimal by the AF property of $\mathcal{M}$. Let $B = (i_1, \ldots, i_\ell)$ be any bundle. Write $E_B = \bigotimes_{r = 1}^{\ell}E_{i_r}, E'_B = \bigotimes_{r = 1}^{\ell}E'_{i_r}$ with signal spaces $S_B = \prod_{r = 1}^{\ell}S_{i_r}, S'_B = \prod_{r = 1}^{\ell}S'_{i_r}$. Define a kernel \(\Gamma_B:S_B\to\Delta(S'_B)\) as:
 \[ \Gamma_B\left((a_1, \ldots, a_\ell) \mid (s_1, \ldots, s_\ell)\right) := \prod_{r = 1}^{\ell}\Gamma_{i_r}(a_r \mid s_r), \]
 for every \((a_1,\ldots,a_\ell)\in S'_B\) and \((s_1,\ldots,s_\ell)\in S_B\). We claim that \(\Gamma_B\) garbles \(E_B\) into \(E'_B\). Indeed, for every \((a_1,\ldots,a_\ell)\in S'_B\) and every \(\omega\in\Omega\),
 \begin{align*}
  &\quad\sum_{(s_1, \ldots, s_\ell) \in S_B}
  \Gamma_B \left((a_1, \ldots, a_\ell) \mid (s_1, \ldots, s_\ell)\right) 
  K_{E_B}(s_1, \ldots, s_\ell \mid \omega) \\
  &= \sum_{(s_1, \ldots, s_\ell) \in S_B}
  \prod_{r = 1}^{\ell}\Gamma_{i_r}(a_r \mid s_r)K_{i_r}(s_r \mid \omega) = \sum_{s_1 \in S_{i_1}} \cdots \sum_{s_\ell \in S_{i_\ell}}
  \prod_{r = 1}^{\ell}\Gamma_{i_r}(a_r \mid s_r)K_{i_r}(s_r \mid \omega) \\
  &= \left(\sum_{s_1 \in S_{i_1}}\Gamma_{i_1}(a_1 \mid s_1)K_{i_1}(s_1 \mid \omega)\right)\cdots
  \left(\sum_{s_\ell \in S_{i_\ell}}\Gamma_{i_\ell}(a_\ell \mid s_\ell)K_{i_\ell}(s_\ell \mid \omega)\right) \\
  &= \prod_{r = 1}^{\ell}\sum_{s_r \in S_{i_r}}\Gamma_{i_r}(a_r \mid s_r)K_{i_r}(s_r \mid \omega) = \prod_{r = 1}^{\ell}K'_{i_r}(a_r \mid \omega)
  = K_{E'_B}(a_1, \ldots, a_\ell \mid \omega).
 \end{align*}
 Hence \(E_B\succeqB E'_B\). Thus \(\mathcal M'\) is AF and responsive. Since prices are unchanged and each intended type's value is unchanged, \(\mathcal M'\) has the same expected seller revenue as \(\mathcal M\).
\end{proof}

\section{Omitted Proofs in Section~\ref{sec:program}} \label{app:omitted-proofs-sec3}

\begin{proof}[Proof of Example~\ref{ex:repeated-purchase-strict}]
 Fix an odd integer \(m = 2r + 1\). For every \(t \geqslant 1\), write $E^{\otimes t}=(S^t, K_t)$, $K_t((s_1, \ldots, s_t) \mid \omega) = \prod_{j = 1}^t K_E(s_j \mid \omega)$. For odd \(t\), define the majority rule \(a_t: S^t \to \{0, 1\}\) by
 \[ a_t(s_1, \ldots, s_t) = \begin{cases}
  1, & \sum_{j = 1}^t s_j \geqslant (t + 1) / 2,\\
  0, & \sum_{j = 1}^t s_j \leqslant (t - 1) / 2.
  \end{cases} \]
 We first show that \(a_t\) is optimal for \(E^{\otimes t}\). Let \(k = \sum_{j = 1}^t s_j\). Under the uniform prior,
 \[ \frac{\Pr(\omega = 1 \mid s_1, \ldots, s_t)}{\Pr(\omega = 0 \mid s_1, \ldots, s_t)} = \left(\frac{1 - \varepsilon}{\varepsilon}\right)^{2k - t}.\]
 Since \(0 < \varepsilon < 1 / 2\), this ratio is greater than one if and only if \(k > t / 2\). Because \(t\) is odd, there is no tie, so the optimal action is exactly the majority rule \(a_t\).

 Let \(q_t\) denote the expected utility from \(a_t\). By symmetry between states \(0\) and \(1\), it suffices to condition on \(\omega=1\). Then the number of correct signals follows \(\mathrm{Bin}(t, 1 - \varepsilon)\), so for \(m = 2r + 1\),
 \[ q_m = \sum_{k = r + 1}^{2r + 1}\binom{2r + 1}{k}(1 - \varepsilon)^k\varepsilon^{2r + 1 - k}. \]

 To compare \(q_m\) and \(q_{m + 2}\), let \(C_m\sim\mathrm{Bin}(2r + 1, 1 - \varepsilon)\) be the number of correct signals among the first \(m\) copies, and let \(Y\sim\mathrm{Bin}(2, 1 - \varepsilon)\) be the number of correct signals among the last two copies. Under \(\omega = 1\), the variables \(C_m\) and \(Y\) are independent. The decisions under \(E^{\otimes m}\) and \(E^{\otimes(m + 2)}\) differ only in the following two cases: (1) \(C_m = r\), in which case the \(m\)-copy majority is wrong and the \((m + 2)\)-copy majority becomes correct exactly when \(Y = 2\); (2) \(C_m = r + 1\), in which case the \(m\)-copy majority is correct and the \((m + 2)\)-copy majority becomes wrong exactly when \(Y = 0\). Hence $q_{m + 2} - q_m = \Pr(C_m = r, Y = 2) - \Pr(C_m = r + 1, Y = 0)$. Using independence and the binomial formulas,
 \begin{align*}
  q_{m + 2} - q_m
  & = \binom{2r + 1}{r}(1 - \varepsilon)^r\varepsilon^{r + 1}(1 - \varepsilon)^2 - \binom{2r + 1}{r + 1}(1 - \varepsilon)^{r + 1}\varepsilon^r\varepsilon^2\\
  & = \binom{2r + 1}{r}(1 - \varepsilon)^{r + 1}\varepsilon^{r + 1}(1 - 2\varepsilon) > 0.
 \end{align*}
 
 To show \(E^{\otimes(m+2)}\succeq_B E^{\otimes m}\), we only need to define the garbling kernel \(\Gamma: S^{m + 2} \to \Delta(S^m)\) by $\Gamma((s_1, \ldots, s_m) \mid (s'_1, \ldots, s'_{m + 2})) = \mathbf 1\{(s_1, \ldots, s_m) = (s'_1, \ldots, s'_m)\}$. If the reverse dominance \(E^{\otimes m}\succeq_B E^{\otimes(m + 2)}\) also held, then the value of every decision problem would be weakly higher under \(E^{\otimes m}\) than under \(E^{\otimes(m + 2)}\), contradicting \(q_{m + 2} > q_m\). Therefore $E^{\otimes(m + 2)} \succ_B E^{\otimes m}$.
\end{proof}

\begin{proof}[Proof of Theorem~\ref{thm:af-revenue-hardness}, part~(i)]
  We reduce from Maximum Independent Set (MIS) on graphs with maximum vertex degree being a fixed constant \(\Delta \geqslant 3\). This problem has no PTAS unless \(P = NP\) \citep{alimonti2000some}. Let \(G = (V, E)\), \(N := |V|\), \(M := |E|\), and let \(\alpha\) be the size of a maximum independent set in \(G\). Since every vertex has degree at most \(\Delta\), we have $M \leqslant \frac{\Delta N}{2}$, $\alpha \geqslant \frac{N}{\Delta + 1}$, and therefore
  \begin{equation} \label{eq:mis-edge-alpha-new}
    M \leqslant \frac{\Delta(\Delta + 1)}2\,\alpha.
  \end{equation}

  \paragraph{Step 1: Construct the information selling instance.} Let $\Omega = \{\varnothing\} \cup \{\{v\}: v \in V\}$ and let the prior be uniform on \(\Omega\), i.e., $\mu(\varnothing) = \mu(\{v\}) = \frac{1}{N + 1}, \forall v \in V$. For every \(R \subseteq V\), define the deterministic subset query \(Q_R\) by $Q_R(\omega) = \omega \cap R, \ \forall \omega \in \Omega$, and let \(E_R\) be the induced experiment. For an edge \(e = \{u, v\}\), write \(E_e = E_{\{u,v\}}\).

  The buyer type space is $\Theta = V \cup E$. Each vertex \(v \in V\) gives one vertex type, and each edge \(e \in E\) gives one edge type. The action space is $A = \{a^0\} \cup \{a^v: v \in V\} \cup \{a^e: e \in E\}$. Fix an arbitrary constant \(\zeta \in (0, 1)\). Choose an integer $C \geqslant \max\left\{2, \frac{20(\Delta + 1)^2}{\zeta^2}\right\}$, and define $a = \frac{1}{3C(N + 1)}$. For a vertex type \(v\in V\), define
  \[ u_v(\omega, a^0) = 1 - \frac{1}{3C}, \forall \omega \in \Omega \quad \text{and} \quad u_v(\omega, a^v) = \begin{cases}
    1, & \omega = \varnothing, \\
    0, & \omega = \{v\}, \\
    1 - \frac{1}{3C}, & \omega = \{w\},\ w \neq v.
  \end{cases} \]
  For every \(a \notin \{a^0, a^v\}\), set \(u_v(\omega, a) = 0\). For an edge type \(e = \{u,v\}\), define
  \[ u_e(\omega, a^0) = \frac23, \forall \omega \in \Omega \quad \text{and} \quad u_e(\omega, a^e)= \begin{cases}
    1, & \omega = \varnothing, \\
    0, & \omega \in \bigl\{\{u\}, \{v\}\bigr\}, \\
    \frac23, & \omega = \{w\},\ w \notin \{u, v\}.
  \end{cases} \]
  For every \(a \notin \{a^0, a^e\}\), set \(u_e(\omega, a) = 0\).

  For a vertex type \(v\), the expected payoff difference between \(a^v\) and \(a^0\) under the prior is $\left(\frac{2}{3C} - 1\right) / (N + 1) < 0$. For an edge type \(e\), the expected payoff difference between \(a^e\) and \(a^0\) under the prior is $-\frac{1}{N + 1} < 0$. Thus every type chooses $a^0$ under the prior. Then we can compute the value of deterministic subset queries. For every \(R\subseteq V\),
  \begin{equation} \label{eq:deterministic-values-new}
    V_v(E_R) = \begin{cases}
      a, & v \in R, \\
      0, & v \notin R,
    \end{cases} \quad V_e(E_R) = \begin{cases}
      Ca, & e \subseteq R, \\
      0, & e \nsubseteq R.
    \end{cases}
  \end{equation}

  We assign type mass \(1 /(N + M)\) to every type. To simplify notation, throughout the proof we use the unnormalized revenue $\widetilde R(\mathcal M) = \sum_{\theta \in \Theta} p_\theta$, and denote \(\widetilde R^*(G)\) as the optimal unnormalized revenue in the instance constructed from \(G\). The normalized revenue is \(\widetilde R(\mathcal M) / (N + M)\). 

  \paragraph{Step 2: Lower bound from an independent set.} Let \(I\subseteq V\) be an independent set. Consider the direct menu
  \[ \mathcal M_I = \{(E_{\{v\}}, a): v \in I\} \cup \{(E_{\emptyset}, 0): v \in V \setminus I\} \cup \{(E_e, Ca): e \in E\}. \]
  We claim that \(\mathcal M_I\) is AF. For a vertex type \(v \in I\), any bundle with positive value must contain some query whose subset contains \(v\). Such a bundle costs at least \(a\), and the intended product costs exactly \(a\). For an edge type \(e = \{u, v\}\), a bundle has positive value for this type only if the union of the purchased subsets contains both \(u\) and \(v\). Since \(I\) is independent, the menu does not contain both singleton queries \(E_{\{u\}}\) and \(E_{\{v\}}\). Hence any such bundle must contain at least one edge query, and therefore costs at least \(Ca\). Therefore \(\mathcal M_I\) is AF and $\widetilde R(\mathcal M_I) = CaM + a|I|$. Taking \(I\) to be a maximum independent set gives
  \begin{equation} \label{eq:mis-lower-bound-new}
    \widetilde R^*(G) \geqslant CaM + a\alpha.
  \end{equation}

  \paragraph{Step 3: Upper bound for an arbitrary AF menu.} Let $\mathcal M = \{(F_\theta, p_\theta)\}_{\theta \in \Theta}$ be any direct AF menu. Since the maximum value of any experiment is at most \(a\) for a vertex type and at most \(Ca\) for an edge type, individual rationality gives $p_v \leqslant a$, $p_e \leqslant Ca$. Define $x_v = \frac{p_v}{a} \in [0,1]$, $\forall v \in V$. Fix a vertex \(v\). Let \(K_v\) be the kernel of \(F_v\), and let \(S_v\) be the set of signals after which action \(a^v\) is weakly better than $a^0$ for type \(v\). Equivalently,
  \[ S_v = \left\{s: \frac{1}{3C}\mu(\varnothing)K_v(s \mid \varnothing) \geqslant \left(1 - \frac{1}{3C}\right)\mu(\{v\})K_v(s \mid \{v\})\right\}. \]
  Let $r_v = K_v(S_v \mid \varnothing)$. For type \(v\), the incremental value of \(F_v\) relative to $a^0$ is
  \[ V_v(F_v) = \sum_{s\in S_v} \left(\frac{1}{3C}\mu(\varnothing)K_v(s \mid  \varnothing)- \left(1 - \frac{1}{3C}\right)\mu(\{v\})K_v(s \mid \{v\})\right) \leqslant a \, K_v(S_v \mid \varnothing) = ar_v. \]
  Since \(p_v \leqslant V_v(F_v)\), we obtain
  \begin{equation} \label{eq:vertex-certificate-new}
    x_v \leqslant r_v.
  \end{equation}
  Moreover, by the definition of \(S_v\), we have
  \begin{equation} \label{eq:vertex-contamination-new}
    \left(1 - \frac{1}{3C}\right)\mu(\{v\})K_v(S_v \mid \{v\}) \leqslant ar_v.
  \end{equation}

  Now fix an edge \(e = \{u,v\}\). Consider the deviation in which type \(e\) buys \(F_u\otimes F_v\). For $w \in \{u, v\}$, define $\rho_w = x_w / r_w$ if $r_w > 0$, and $\rho_w = 0$ if $r_w = 0$. Then $\rho_w \in [0, 1]$ by~\eqref{eq:vertex-certificate-new}. We lower bound the value of
  the deviation by considering the following randomized decision rule: after observing \((s_u, s_v)\), take action \(a^e\) with probability \(\rho_u \rho_v\) if \(s_u \in S_u\) and \(s_v \in S_v\), and take $a^0$ otherwise. Note that the randomization is only for analysis: the buyer's optimal pure action after each signal pair gives payoff at least that of any lottery over pure actions.

  Under state \(\varnothing\), the probability of taking $a^e$ is \(x_u x_v\), so the positive incremental contribution is \(Ca x_u x_v\). Under state \(\{u\}\), its negative contribution
  is at most $\frac23 \, \mu(\{u\}) \rho_u K_u(S_u \mid \{u\}) \leqslant \frac{\frac23\,a}{1 - \frac{1}{3C}}x_u$, where the inequality uses~\eqref{eq:vertex-contamination-new}. Similarly, the negative contribution under \(\{v\}\) is at most \(\frac{\frac23 a}{1 - \frac{1}{3C}}x_v\). All other states give zero incremental contribution for type \(e\). Hence
  \begin{equation} \label{eq:edge-deviation-value-new}
    V_e(F_u \otimes F_v) \geqslant \left[Ca x_u x_v - \frac{\frac23 a}{1 - \frac{1}{3C}}(x_u + x_v)\right]_+,
  \end{equation}
  where \([z]_+ = \max\{0, z\}\) for any \(z \in \mathbb R\). The AF constraint for type \(e\) gives $V_e(F_e) - p_e \geqslant V_e(F_u \otimes F_v) - p_u - p_v$. Using \(V_e(F_e) \leqslant Ca\), \(p_u + p_v = a(x_u + x_v)\), \eqref{eq:edge-deviation-value-new}, and \(p_e \leqslant Ca\), we get
  \begin{equation} \label{eq:edge-price-bound-new}
    p_e \leqslant Ca - a\left[Cx_ux_v - K(x_u + x_v)\right]_+, \quad K := 1 + \frac{\frac23}{1 - \frac{1}{3C}}.
  \end{equation}
  Since \(C \geqslant 2\), we have \(K \leqslant 2\). Summing the vertex prices and the edge price bounds~\eqref{eq:edge-price-bound-new}, every direct AF menu satisfies
  \begin{equation} \label{eq:fractional-upper-bound-new}
    \widetilde R(\mathcal M) \leqslant CaM + a F_C(x),
  \end{equation}
  where $F_C(x) = \sum_{v \in V}x_v - \sum_{\{u, v\} \in E} \left[Cx_ux_v - K(x_u + x_v)\right]_+$.

  \paragraph{Step 4: Final reduction.}
  Let $\tau = \sqrt{5 / C}$, $H = \{v \in V: x_v \geqslant \tau\}$. Denote $E[H] = \{\{u, v\} \in E: u, v \in H\}$ as the set of edges in the induced subgraph \(G[H]\). Then for every edge \(\{u, v\} \in E[H]\), we have $Cx_ux_v - K(x_u + x_v) \geqslant C\tau^2 - 2K \geqslant 1$. Therefore
  \[ F_C(x) \leqslant \tau N + \sum_{v \in H}x_v - |E[H]| \leqslant \tau N + |H| - |E[H]|. \]
  Deleting one endpoint from each edge in the induced subgraph \(G[H]\) leaves an independent set of size at least \(|H| - |E[H]|\). Hence $|H| - |E[H]| \leqslant \alpha$. Using \(N \leqslant (\Delta + 1)\alpha\), we obtain
  \begin{equation} \label{eq:rounding-upper-new}
    F_C(x) \leqslant \bigl(1 + \tau(\Delta + 1)\bigr)\alpha.
  \end{equation}
  Combining~\eqref{eq:mis-lower-bound-new}, \eqref{eq:fractional-upper-bound-new}, and \eqref{eq:rounding-upper-new}, we have
  \begin{equation} \label{eq:opt-sandwich-new}
    CaM + a\alpha \leqslant \widetilde R^*(G) \leqslant CaM + a\bigl(1 + \tau(\Delta + 1)\bigr)\alpha.
  \end{equation}
  
  Then we construct an independent set from any feasible menu. Given any feasible menu \(\mathcal M\), form \(H = \{v: x_v \geqslant \tau\}\) and delete one endpoint from every edge in \(G[H]\). Let the resulting independent set be \(I(\mathcal M)\). Then
  \begin{equation} \label{eq:extraction-bound-new}
    |I(\mathcal M)| \geqslant |H|-|E[H]| \geqslant F_C(x) - \tau N.
  \end{equation}

  Suppose that Program~(3) admits a PTAS. Let $D_\Delta = \frac{\Delta(\Delta + 1)}2$. Run this PTAS on the instance constructed above with accuracy parameter $\eta \leqslant \frac{\zeta}{2(CD_\Delta + 1)}$. Let \(\mathcal M\) be the returned direct AF menu, and let $\widehat\alpha = 
  \frac{\widetilde R(\mathcal M) - CaM}{a}$. By the approximation guarantee and the lower bound~\eqref{eq:mis-lower-bound-new}, $\widetilde R(\mathcal M) \geqslant (1 - \eta)\widetilde R^*(G) \geqslant (1 - \eta)(CaM + a\alpha)$. This implies $\widehat\alpha \geqslant \alpha - \eta(CM + \alpha) \geqslant \bigl(1 - \eta(CD_\Delta + 1)\bigr)\alpha$, where the last inequality uses~\eqref{eq:mis-edge-alpha-new}. On the other hand, \eqref{eq:fractional-upper-bound-new} gives $\widehat\alpha \leqslant F_C(x)$. Construct \(I(\mathcal M)\) from the menu as in Step~4. By~\eqref{eq:extraction-bound-new}, $|I(\mathcal M)| \geqslant \widehat\alpha - \tau N \geqslant \bigl(1 - \eta(CD_\Delta + 1) - \tau(\Delta + 1)\bigr)\alpha$. By the choice of \(C\), we have \(\tau(\Delta + 1) \leqslant \zeta / 2\), and by the choice of \(\eta\), we have \(\eta(CD_\Delta + 1) \leqslant \zeta / 2\). Hence $|I(\mathcal M)| \geqslant (1 - \zeta)\alpha$. Thus a PTAS for Program~(3) would imply a PTAS for MIS on bounded degree graphs. This contradicts the APX-hardness of bounded-degree MIS unless \(P = NP\). Therefore Program~(3) admits no PTAS.
\end{proof}

\begin{proof}[Proof of Theorem~\ref{thm:af-revenue-hardness}, part~(ii)]
 Given a Boolean formula \(\varphi(x_1, \ldots, x_n)\) on \(n \geqslant 1\). Let $\Omega = \{0, 1\}^{2n}$, and denote a state as \(\omega = (x, y)\) with \(x, y \in \{0, 1\}^n\). The prior is uniform on \(\Omega\), and can be sampled in polynomial time by drawing independent uniform bit vectors \(x, y \in \{0, 1\}^n\).

 There is one buyer type and two actions $A = \{0, 1\}$. Define the event $E_\varphi = \{(x, y) \in \Omega: \varphi(x) = 1,\ y = 0^n\}$. The utility function is
 \[ u((x, y), 1) = \mathbf 1\{(x, y) \in E_\varphi\}, \quad u((x, y),0) = \mathbf 1\{(x, y) \notin E_\varphi\}. \]
 Thus action \(1\) is correct exactly on \(E_\varphi\), and action \(0\) is correct exactly outside \(E_\varphi\).

 Let $q = \Pr(E_\varphi)$. Since \(y = 0^n\) has probability \(2^{-n}\), we have $q \leqslant 2^{-n} \leqslant \frac{1}{2}$. Therefore, without any information, the buyer optimally chooses action \(0\) and obtains utility $U^0 = 1 - q$. Now consider any experiment \(F\). The buyer's payoff after observing \(F\) is at most \(1\), so the value of \(F\) satisfies $V(F) \leqslant 1 - (1 - q) = q$. On the other hand, the complete information experiment achieves payoff \(1\) in every state, so its value is exactly $1 - (1 - q) = q$. Hence $\OPT = q$, where \(\OPT\) is the optimal revenue.

 If \(\varphi\) is unsatisfiable, then \(E_\varphi=\varnothing\), so \(q=0\) and therefore $\OPT = 0$. If \(\varphi\) is satisfiable, then some \(x^\star \in \{0, 1\}^n\) satisfies \(\varphi(x^\star) = 1\). The state \((x^\star, 0^n)\) then belongs to \(E_\varphi\), and therefore $q \geqslant 2^{-2n} > 0$. Thus $\varphi \text{ is satisfiable} \iff \OPT > 0$.

 Suppose that a polynomial time constant factor approximation algorithm existed. Let \(c \geqslant 1\) be its approximation factor, then (1) if \(\varphi\) is unsatisfiable, then every feasible menu has revenue \(0\), so the algorithm must output revenue \(0\); (2) if \(\varphi\) is satisfiable, then the algorithm must output revenue at least \(\OPT / c = q / c > 0\). Hence the algorithm decides satisfiability by checking whether the output revenue is zero or positive, which is impossible unless \(P = NP\).
\end{proof}

\begin{proof}[Proof of Lemma~\ref{lem:coord-compression}]
 We first give a basic approximation property of a payoff cell. Fix a cell \(g\), a type \(\theta\), and an action \(a\). If \(\mu(g) > 0\), let
 \[ \bar u_{g, \theta, a} = \frac{M_{g, \theta, a}}{\mu(g)} = \frac{1}{\mu(g)} \sum_{\omega \in g}\mu(\omega)u_\theta(\omega, a). \]
 According to Equation~\ref{eq:rank-cell-diameter}, we have $\left|u_\theta(\omega, a) - \bar u_{g, \theta, a}\right| \leqslant \beta, \forall \omega \in g$. Therefore, for every function \(f: g \to [0, 1]\), we have
 \begin{equation} \label{eq:cell-product-error}
  \left|\sum_{\omega \in g}\mu(\omega)f(\omega)u_\theta(\omega, a) - \bar u_{g, \theta, a}\sum_{\omega \in g}\mu(\omega)f(\omega)\right| \leqslant \beta\sum_{\omega \in g}\mu(\omega)f(\omega).
 \end{equation}
 If \(\mu(g) = 0\), every term above is zero, so~\eqref{eq:cell-product-error} is trivial.

 We now prove part~(i). Start from an AF responsive direct menu. Let \(x(\omega)\in\mathcal X\) be the recommendation profile used at state \(\omega\). For each cell \(g\) with \(\mu(g) > 0\), define
 \[ r_{g, z} = \frac{1}{\mu(g)} \sum_{\omega \in g}\mu(\omega)x(\omega)^z, \quad z \in \mathcal Z_{H_\delta}. \]
 This vector belongs to \(\mathcal R_{H_\delta}(\mathcal X)\). For cells with \(\mu(g)=0\), choose any point in \(\mathcal R_{H_\delta}(\mathcal X)\).

 Recall that $V_\theta(E_\theta) = \sum_{g \in \mathcal G_\beta}\sum_{a \in A}\sum_{\omega \in g} \mu(\omega)x_{\theta, a}(\omega)u_\theta(\omega, a)$ (we neglect the constant term \(U_\theta^0\) since it is independent of the menu). For every pair \((g,a)\), apply~\eqref{eq:cell-product-error} with \(f(\omega) = x_{\theta, a}(\omega)\). Since
 \[ \bar u_{g, \theta, a} \sum_{\omega \in g}\mu(\omega)x_{\theta, a}(\omega) = \frac{M_{g, \theta, a}}{\mu(g)} \sum_{\omega \in g}\mu(\omega)x_{\theta, a}(\omega) = M_{g, \theta, a}\,r_{g,\mathbf e_{\theta, a}}, \]
 summing over \(g\) and \(a\) gives
 \begin{equation} \label{eq:intended-compression-error}
  \left|V_\theta(E_\theta) - \widehat W_\theta(r)\right| \leqslant \sum_{g \in \mathcal G_\beta}\sum_{a \in A} \beta\sum_{\omega \in g}\mu(\omega)x_{\theta, a}(\omega) = \beta\sum_{g \in \mathcal G_\beta}\sum_{\omega \in g}\mu(\omega)\sum_{a \in A}x_{\theta, a}(\omega) = \beta.
 \end{equation}

 Next fix a count vector \(c \in \mathcal C_\delta(\Theta)\) and an action rule \(s: \mathcal Z_\Theta(c) \to A\). The expected utility of this rule after buying \(c\) is $V_{\theta, s}(E_c) = \sum_{g \in \mathcal G_\beta}\sum_{z \in \mathcal Z_\Theta(c)} \sum_{\omega \in g} \mu(\omega)P_{x(\omega)}(z \mid c)\, u_\theta(\omega, s(z))$. Denote
 \[ \widehat V_{\theta, c, s}(r) = \sum_{z \in \mathcal Z_\Theta(c)} \widehat V_{\theta, c, z, s(z)}(r) = \sum_{g \in \mathcal G_\beta}\sum_{z \in \mathcal Z_\Theta(c)} M_{g, \theta, s(z)}\, \gamma(c, z)\, r_{g, z}. \]
 For each pair \((g, z)\), apply~\eqref{eq:cell-product-error} with \(a = s(z)\) and \(f(\omega) = P_{x(\omega)}(z \mid c)\). Similar to Equation~\eqref{eq:intended-compression-error}, we have $\left|V_{\theta, s}(E_c) - \widehat V_{\theta, c, s}(r)\right| \leqslant \beta$. Let \(s^\star\) be an optimal action rule for the bundle value, so \(V_\theta(E_c) = V_{\theta, s^\star}(E_c)\). Then $V_\theta(E_c) \leqslant \widehat V_{\theta, c, s^\star}(r) + \beta \leqslant \widehat V_{\theta, c}(r) + \beta$. Conversely, let \(\widehat V_{\theta, c}(r) = \widehat V_{\theta, c, \widehat s^\star}(r)\). Then $\widehat V_{\theta, c}(r) \leqslant
 V_{\theta, \widehat s^\star}(E_c) + \beta \leqslant V_\theta(E_c) + \beta$. Therefore
 \begin{equation} \label{eq:deviation-compression-error}
  \left|V_\theta(E_c) - \widehat V_{\theta, c}(r)\right| \leqslant \beta.
 \end{equation}
 AF of the original responsive menu gives $V_\theta(E_\theta) - t_\theta \geqslant V_\theta(E_c) - c \cdot t$, $\forall \theta,\ \forall c \in \mathbb Z_+^\Theta \setminus \{e_\theta\}$. Combining this with~\eqref{eq:intended-compression-error} and~\eqref{eq:deviation-compression-error}, we obtain
 \[ \widehat W_\theta(r) - t_\theta \geqslant V_\theta(E_\theta) - \beta - t_\theta \geqslant V_\theta(E_c) - \beta - c \cdot t \geqslant \widehat V_{\theta, c}(r) - 2\beta - c \cdot t. \]
 Thus \((r, t)\) is feasible for \(\operatorname{FinPrice}^{2\beta}_{\delta,\beta}\), which proves part~(i).

 We now prove part~(ii). Suppose \((r, t)\) is feasible for \(\operatorname{FinPrice}^{\alpha}_{\delta, \beta}\). By the Carath\'eodory theorem \cite{rockafellar1997convex}, for each cell \(g\), we choose $r_g = \sum_{\ell = 1}^{K_g}\alpha_{g, \ell}\phi_{H_\delta}(x_{g, \ell})$ with $K_g \leqslant B_\delta + 1$. Then we use the implementation described in the main text to get the final finite support menu. Let \(W_\theta^{\rm rec}\) be the payoff from following the intended recommendation in this implemented menu. Then similar to~\eqref{eq:intended-compression-error}, we have $\left|W_\theta^{\rm rec} - \widehat W_\theta(r)\right| \leqslant \beta$. The actual value of the intended product can only be larger, because the buyer may ignore the recommendation, so $V_\theta(E_\theta) \geqslant W_\theta^{\rm rec} \geqslant \widehat W_\theta(r) - \beta$. Likewise, for every bundle \(c \in \mathcal C_\delta(\Theta)\), $V_\theta(E_c) \leqslant \widehat V_{\theta, c}(r) + \beta$. Since \((r, t)\) is feasible for \(\operatorname{FinPrice}^{\alpha}_{\delta, \beta}\), we have $\widehat W_\theta(r) - t_\theta \geqslant \widehat V_{\theta, c}(r) - c \cdot t - \alpha$. Combining the above discussion yields
 \[ V_\theta(E_\theta) - t_\theta \geqslant \widehat W_\theta(r) - \beta - t_\theta \geqslant \widehat V_{\theta, c}(r) - \beta - c \cdot t - \alpha \geqslant V_\theta(E_c) - c \cdot t - \alpha - 2\beta. \]
 Since \(\Gamma_\beta = 2\beta\), the implemented menu is \((\alpha + \Gamma_\beta)\)-AF against \(\mathcal C_\delta(\Theta)\).
\end{proof}

\begin{proof}[Proof of Proposition~\ref{prop:coord-price}]
 By Proposition~\ref{prop:bundle-revelation}, there exists an optimal exact AF direct menu. Since \(\sigma \geqslant \Gamma_\beta\), Lemma~\ref{lem:coord-compression}(i) implies that this menu induces a feasible point of \(\operatorname{FinPrice}^{\sigma}_{\delta, \beta}\). Therefore the optimum value of \(\operatorname{FinPrice}^{\sigma}_{\delta, \beta}\) is at least \(\OPT\).

 We solve~\eqref{eq:coord-price} by the ellipsoid method and binary search over the objective value. The search interval is bounded: using the IR constraint, \(t_\theta \leqslant 1 + \sigma\), and hence the objective is at most \(1 + \sigma\). For a candidate revenue level \(R\), we add the linear constraint \(\sum_{\theta \in \Theta}\lambda_\theta t_\theta \geqslant R\) and test feasibility. Binary search to additive accuracy \(\tau\) then gives objective value at least \(\OPT - \tau\).

 For AF separation, fix \((\theta, c)\), we test whether the corresponding cut~\eqref{eq:benders-cut} is violated. For each fixed \((\theta,c,z,a)\), evaluating \(\widehat V_{\theta,c,z,a}(r)\) requires summing over \(g\in\mathcal G_\beta\), so it takes \(O(N_\beta)\) arithmetic operations. The total number of triples \((c,z,a)\) is at most \(mB_\delta\), hence one full AF separation pass over all \(\theta\in\Theta\) costs polynomial time in \(N_\beta\) and \(B_\delta\).

 For separation over \(\mathcal R_{H_\delta}(\mathcal X)\), fix a direction \(q\). Let $\rho = \frac{\tau}{m + m^{H_\delta}}$. By~\eqref{eq:coord-body-optimization}, support function evaluation reduces to optimizing a degree-\(H_\delta\) polynomial over the product of simplexes \(\mathcal X\). Equivalently, for a threshold \(\lambda\), it suffices to decide whether there exists \(x \in \mathcal X\) such that \(\sum_{z \in \mathcal Z_{H_\delta}} q_z x^z \geqslant \lambda\). This is an existential formula over the reals with \(nm\) variables, \(O(nm)\) simplex constraints, and one polynomial inequality of degree \(H_\delta\). Standard real algebraic geometry algorithms solve this problem to accuracy \(\rho\) in time \(\operatorname{poly}(\log(1 / \rho)) \cdot H_\delta^{O(nm)} = \operatorname{poly}(\log(1 / \tau)) \cdot H_\delta^{O(nm)}\) \cite{basu2006algorithms}. Hence the weak optimization--weak separation reduction \cite{grotschel1988geometric} yields a weak separation oracle for \(\mathcal R_{H_\delta}(\mathcal X)\) with accuracy \(\rho\) in time $\operatorname{poly}(B_\delta, \log(1 / \tau)) \cdot H_\delta^{O(nm)}$. Checking all $g \in \mathcal G_\beta$ therefore costs $\operatorname{poly}(N_\beta, B_\delta, \log(1 / \tau)) \cdot H_\delta^{O(nm)}$ time. Thus the ellipsoid method returns a profile that can be replaced by an implementable profile \(\widetilde r_g \in \mathcal R_{H_\delta}(\mathcal X)\) such that \(\|r_g - \widetilde r_g\|_\infty \leqslant \rho\).

 We now bound the effect of this replacement on the AF constraints. For every type \(\theta\), since \(0 \leqslant M_{g,\theta,a} \leqslant \mu(g)\) and \(\sum_g M_{g,\theta,a} \leqslant 1\), we have $\left|\widehat W_\theta(r) - \widehat W_\theta(\widetilde r)\right| \leqslant \rho\sum_{a \in A}\sum_{g \in \mathcal G_\beta}M_{g, \theta, a} \leqslant m \rho$.
 Similarly, for fixed \((\theta, c, z, a)\), $\left|\widehat V_{\theta, c, z, a}(r) - \widehat V_{\theta, c, z, a}(\widetilde r)\right| \leqslant \gamma(c, z)\rho$. Using \(|\max_a f_a - \max_a f'_a| \leqslant \max_a |f_a-f'_a|\) and the multinomial theorem, $\left|\widehat V_{\theta, c}(r) - \widehat V_{\theta, c}(\widetilde r)\right| \leqslant \rho \sum_{z \in \mathcal Z_\Theta(c)}\gamma(c, z) = \rho m^{\|c\|_1} \leqslant \rho m^{H_\delta}$. Therefore, if \((r, t)\) satisfies the explicit AF constraints of \(\operatorname{FinPrice}^{\sigma}_{\delta,\beta}\), then \((\widetilde r,t)\) satisfies
 \[ \widehat W_\theta(\widetilde r) - t_\theta \geqslant \widehat V_{\theta, c}(\widetilde r) - c \cdot t - \sigma - (m + m^{H_\delta})\rho = \widehat V_{\theta, c}(\widetilde r) - c \cdot t - \sigma - \tau \]
 for every \(\theta \in \Theta\) and every \(c \in \mathcal C_\delta(\Theta) \setminus \{e_\theta\}\). Thus \((\widetilde r, t)\) is feasible for \(\operatorname{FinPrice}^{\sigma + \tau}_{\delta, \beta}\).

 By Carath\'eodory's theorem \cite{rockafellar1997convex}, each implementable profile \(\widetilde r_g \in \mathcal R_{H_\delta}(\mathcal X) \subseteq \mathbb R^{B_\delta}\) has a decomposition \(\widetilde r_g = \sum_{\ell = 1}^{K_g} \alpha_{g, \ell}\phi_{H_\delta}(x_{g, \ell})\) with \(K_g \leqslant B_\delta + 1\). This yields a finite support menu with the same price vector \(t\), so its revenue is at least \(\OPT - \tau\). Lemma~\ref{lem:coord-compression}(ii) applied with \(\alpha = \sigma + \tau\) shows that this recovered menu is \((\sigma + \tau + \Gamma_\beta)\)-AF against \(\mathcal C_\delta(\Theta)\).

 Finally, the binary search uses \(O(\log(1 / \tau))\) feasibility tests, and the ellipsoid works in dimension \(N_\beta B_\delta + n\) with polynomially many oracle calls. Hence the total running time is $\operatorname{poly}\left(N_\beta B_\delta, \log\frac1\tau\right) \cdot \left(H_\delta\right)^{O(nm)}$. This proves both the guarantee and time complexity.
\end{proof}

\begin{proof}[Proof of Lemma~\ref{lem:threshold-low-prices}]
 Since every deleted product has price below \(\delta\), we have $\sum_{\theta \notin \Theta_\delta} \lambda_\theta t_\theta \leqslant \delta\sum_{\theta \notin \Theta_\delta} \lambda_\theta \leqslant \delta$. Therefore the remaining revenue is at least \(\sum_\theta\lambda_\theta t_\theta - \delta\). To prove that the retained menu is \(\alpha\)-AF, fix any type \(\theta \in \Theta\) and any finite bundle \(B\) of retained products, and let \(c(B)\) be its count vector. If \(\|c(B)\|_1 \leqslant H_\delta\), then \(c(B) \in \mathcal C_\delta(\Theta)\), so the desired inequality follows directly from the assumption that the original menu is \(\alpha\)-AF against \(\mathcal C_\delta(\Theta)\). If \(\|c(B)\|_1 > H_\delta\), then \(t_B > H_\delta \delta \geqslant 1\). Since \(V_\theta(E_B) \leqslant 1\), we get \(V_\theta(E_B) - t_B < 0\). On the other hand, \(0 \in \mathcal C_\delta(\Theta) \setminus \{e_\theta\}\), so the assumption with \(c = 0\) gives \(V_\theta(E_\theta) - t_\theta \geqslant -\alpha\). Therefore \(V_\theta(E_\theta) - t_\theta \geqslant V_\theta(E_B) - t_B - \alpha\). Thus the retained menu is \(\alpha\)-AF against all finite bundles of retained products.
\end{proof}

\begin{proof}[Proof of Lemma~\ref{lem:exact-repair}]
 Fix a retained type \(\theta \in \widehat\Theta\), and define \(U_\theta = V_\theta(E_\theta) - t_\theta\). After scaling, the intended product of type \(\theta\) is offered at price \((1 - \eta)t_\theta\), so its net utility becomes $V_\theta(E_\theta) - (1 - \eta)t_\theta = U_\theta + \eta t_\theta$. Let \(B_\theta\) be the bundle that maximizes type \(\theta\)'s utility under the scaled prices, we have $V_\theta(E_{B_\theta}) - (1 - \eta)p_\theta^\star \geqslant U_\theta + \eta t_\theta$, where \(p_\theta^\star\) is the original unscaled total price of \(B_\theta\). Since the menu is \(\alpha\)-AF before scaling, we have $V_\theta(E_{B_\theta}) - p_\theta^\star \leqslant U_\theta + \alpha$. Adding \(\eta p_\theta^\star\) to both sides yields $V_\theta(E_{B_\theta}) - (1 - \eta)p_\theta^\star \leqslant U_\theta + \alpha + \eta p_\theta^\star$. Thus we have $U_\theta + \eta t_\theta \leqslant U_\theta + \alpha + \eta p_\theta^\star$, hence $p_\theta^\star \geqslant t_\theta - \frac{\alpha}{\eta}$.

 Proposition~\ref{prop:bundle-revelation} applied to the scaled menu yields an exactly AF direct menu with the same revenue as that scaled menu. Applying the lower bound \(p_\theta^\star \geqslant t_\theta - \alpha / \eta\) for retained types, we obtain
 \[ \sum_{j \in \Theta}\lambda_j p_j^\star \geqslant \sum_{\theta \in \widehat\Theta}\lambda_\theta \left(t_\theta - \frac{\alpha}{\eta}\right) = \sum_{\theta \in \widehat\Theta}\lambda_\theta t_\theta - \frac{\alpha}{\eta}\sum_{\theta \in \widehat\Theta}\lambda_\theta \geqslant \sum_{\theta \in \widehat\Theta}\lambda_\theta t_\theta - \frac{\alpha}{\eta}. \]
 Multiplying by \((1 - \eta)\) gives the claimed revenue bound.
\end{proof}

\begin{proof}[Proof of Proposition~\ref{prop:full-info-approx}]
 \textbf{Part 1: \(1 + \ln(v_{\max} / v_{\min})\)-approximation.} Let $D(p) = \sum_{\theta: v_\theta \geqslant p}\lambda_\theta$ be the demand for complete information at price \(p\), and let $S = \sum_\theta \lambda_\theta v_\theta$ be the total full information surplus. Therefore \(\OPT \leqslant S\). Also \(\Rfull = \max_{p \geqslant 0} pD(p)\), so \(D(p) \leqslant \Rfull / p\) for every \(p > 0\). By the tail integration formula, $S = \int_0^{v_{\max}} D(p)\,\textup{d}p$. For \(0 \leqslant p \leqslant v_{\min}\), we have
 \[ \int_0^{v_{\min}} D(p)\,\textup{d}p = v_{\min}D(v_{\min}) \leqslant \Rfull. \]
 For \(p \in [v_{\min}, v_{\max}]\), we have
 \[ \int_{v_{\min}}^{v_{\max}}D(p)\, \textup{d}p \leqslant \int_{v_{\min}}^{v_{\max}}\frac{\Rfull}{p}\, \textup{d}p = \Rfull\ln(v_{\max} / v_{\min}). \]
 Combining the two inequalities gives $\OPT \leqslant S \leqslant \Rfull\left(1 + \ln(v_{\max} / v_{\min})\right)$.

 \medskip\noindent\textbf{Part 2: Tightness of \(1 + \ln(v_{\max} / v_{\min})\) factor.} Fix an integer \(h \geqslant 2\). Let the state space be $\Omega = \{0, 1\}^h$ with the uniform prior. For each \(i \in [h]\), define the deterministic coordinate query $T_i(\omega) = \omega_i$. Let the common action space be $A = \{a^0\} \cup \{a_i^0, a_i^1: i \in [h]\}$. Type \(i\) has mass $\lambda_i = \frac1h$ and utility function $u_i(\omega, a^0) = 1 - \frac1{2i}$, $u_i(\omega, a_i^b) = 1$ if \(\omega_i = b\) and $u_i(\omega, a_i^b) = 0$ otherwise, and $u_i(\omega, a_j^b) = 0$ for every \(j \neq i\). Without information, every type \(i\) chooses $a^0$ and \(U_i^0 = 1 - \frac1{2i}\). For any bundle of the coordinate queries, if the bundle reveals \(\omega_i\), then type \(i\) chooses the matching action in \(\{a_i^0, a_i^1\}\) and obtains payoff \(1\); otherwise the posterior of \(\omega_i\) remains uniform, so $a^0$ remains optimal. Therefore type \(i\) has threshold utility with target \(T_i\) and value $v_i = \frac1{2i}$.

 Consider the direct menu that sells each target query \(T_i\) at price \(1 / (2i)\). We claim that this menu is AF. Fix a type \(i\), and consider an arbitrary finite bundle \(B\) of menu items. First, repeated copies of the same deterministic query do not create any new information. Then, let \(J \subseteq[h]\) be the set of coordinates purchased in the resulting bundle. If \(i \in J\), then the bundle contains \(T_i\), and the bundle cost is at least the posted price of \(T_i\). If \(i \notin J\), then type \(i\) obtains value zero from the bundle \(B\).

 The revenue of this menu is $\sum_{i = 1}^h \lambda_i v_i = \sum_{i = 1}^h \frac1h \cdot \frac1{2i} = \frac{H_h}{2h}$, where $H_h = \sum_{i = 1}^h \frac1i$. On the other hand, the total full information surplus is also $\sum_{i = 1}^h \lambda_i v_i = \frac{H_h}{2h}$. No individually rational menu can extract more than the total full information surplus. Hence $\OPT = \frac{H_h}{2h}$.

 It remains to compute \(\Rfull\). If \(p \in (1 / (2(k + 1)), 1 / (2k)]\) for some \(k < h\), then exactly the types \(1, \ldots, k\) buy, so the revenue is at most $p \cdot \frac{k}{h} \leqslant \frac1{2k} \cdot \frac{k}{h} = \frac1{2h}$. If \(p \leqslant 1 / (2h)\), then all \(h\) types buy and the revenue is at most $p \cdot 1 \leqslant \frac1{2h}$. If \(p > 1 / 2\), then no type buys. Therefore $\Rfull = \frac1{2h}$. Thus we have $\frac{\OPT}{\Rfull} = H_h$. Since $H_h = (1 - o(1))(1 + \ln h)$ and \(v_{\max} / v_{\min} = h\), thus proving that the logarithmic factor is tight.

 \medskip\noindent\textbf{Part 3: \(n\)-approximation.} Let \(\mathcal{M} = \{(E_\theta, t_\theta)\}_{\theta \in \Theta}\) be any AF direct menu. Individual rationality implies that \(t_\theta \leqslant V_\theta(E_\theta)\). Since complete information Blackwell dominates \(E_\theta\), we have \(V_\theta(E_\theta) \leqslant V_\theta(F) = v_\theta\). Therefore the revenue of \(\mathcal{M}\) satisfies $R(\mathcal{M}) = \sum_\theta \lambda_\theta t_\theta \leqslant \sum_\theta \lambda_\theta v_\theta$.

 For every type \(\theta\), the complete information price \(p = v_\theta\) sells at least to that type, so \(\Rfull \geqslant \lambda_\theta v_\theta\). Summing over all \(n = |\Theta|\) types yields \(\sum_\theta\lambda_\theta v_\theta \leqslant n\Rfull\), thus $R(\mathcal{M}) \leqslant n\Rfull$. Take the supremum over all AF direct menus \(\mathcal{M}\) to conclude \(\OPT \leqslant n\Rfull\).

 \medskip\noindent\textbf{Part 4: Tightness of the \(n\) factor.}
 Fix \(n\) and \(\rho \in (0, 1)\). Let \(\Omega = \{0, 1\}^n\) with the uniform prior. For each \(i \in [n]\), let \(T_i(\omega) = \omega_i\), and create a threshold type similar to Part~2 of this proof with target \(T_i\), value \(v_i = \frac12\rho^{i - 1}\), and mass $\lambda_i = \frac{1 / v_i}{\sum_{j = 1}^n 1 / v_j}$. Let \(c = (\sum_{j = 1}^n 1 / v_j)^{-1}\), so \(\lambda_i v_i = c\) for all \(i\). Sell each target \(T_i\) at price \(v_i\). According to the Part~2 argument, this menu is AF. The revenue of this menu is $\sum_{i = 1}^{n}\lambda_i v_i = \sum_{i = 1}^{n} c = nc = \OPT$.
 
 It remains to bound \(\Rfull\). The values satisfy \(v_1 > v_2 > \cdots > v_n\). If \(p \in (v_{k + 1}, v_k]\) for some \(k < n\), then only types \(1, \ldots, k\) buy, and the revenue is at most $v_k\sum_{i = 1}^k \lambda_i = \sum_{i = 1}^k c\frac{v_k}{v_i} = c \sum_{i = 1}^k \rho^{k - i} \leqslant \frac{c}{1 - \rho}$. If \(p \leqslant v_n\), the revenue is at most \(v_n \leqslant c / (1 - \rho)\); if \(p > v_1\), the revenue is zero. Thus \(\Rfull \leqslant c / (1 - \rho)\), and $\frac{\OPT}{\Rfull} \geqslant \frac{nc}{c / (1 - \rho)} = n(1 - \rho)$. Choosing \(\rho \to 0\) gives the claimed lower bound.

 \medskip\noindent\textbf{Part 5: running time.} First, sorting the values $v_{\theta}$ takes \(O(|\Theta|\log|\Theta|)\), and then scanning the resulting candidate prices is linear in \(|\Theta|\) since we only need to check prices \(p = v_\theta\) for \(\theta \in \Theta\).
\end{proof}

\subsection{Experimental Validation of the FPTAS} \label{app:experiments}

\paragraph{Data selling instance construction.}
\label{app:section32-modeling}
We first test the algorithm on two abstract markets. Let \(\Theta=[n]\), \(A=[m]\), and draw the type masses \(\lambda\) from a Dirichlet distribution.  For the state set, the finite market uses \(\Omega=[S]\) and the continuous market uses \(\Omega=[0,1]^d\). Before the experiment starts, the buyer's prior distribution over the market condition is uniform.

We first consider \textbf{finite latent feature markets}. In realistic scenarios, high-dimensional markets often have low-dimensional intrinsic structure.  For example, a large dataset may contain many attributes and possible states, but its value is often driven by a few common quality dimensions such as freshness, coverage, and privacy risk. So buyer utilities can be represented through low-dimensional market features:
\[
    u_{s\theta a}=\operatorname{sigmoid}\{b_{\theta a}+\rho(p_\theta^\top g_s+q_a^\top g_s)\},
\]
where \(b_{\theta a}\) is type \(\theta\)'s preference or cost for executing action \(a\); the factor \(g_s\) represents common conditions such as demand, quality, or risk; and \(p_\theta,q_a\) determine how sensitive type \(\theta\) and action \(a\) are to these conditions. We draw \(b_{\theta a}\), \(g_s\), \(p_\theta\), and \(q_a\) from Gaussian distributions.

Next we assume the market condition (state space) to be continuous, such as time and location. We formulate this as \textbf{continuous random markets} with \(\Omega=[0,1]^d\) and \(\omega=x\), where \(d\) is the dimension of the continuous state. The buyer's utility is
\[
    u_{\theta a}(x)=\operatorname{sigmoid}\{b_{\theta a}+\rho f_{\theta a}(x)+\ell_{\theta a}^{\top}x\},
\]
where \(f_{\theta a}(x)\) describes how the continuous market condition \(x\) changes the utility of action \(a\) for type \(\theta\) in a nonlinear way, and \(\ell_{\theta a}^{\top}x\) captures directional effects such as utility increasing with quality or decreasing with risk. We use different constructions for \(f_{\theta a}\):
\begin{itemize}[leftmargin=2em,itemsep=1pt,topsep=1pt]
    \item \emph{Radial basis function:} \(f^{\rm RBF}_{\theta a}(x)=\sum_k\alpha_{\theta ak}\exp(-\|x-c_k\|^2/w_k^2)\), which creates local demand or quality hot spots around centers \(c_k\).
    \item \emph{Fourier function:} \(f^{\rm F}_{\theta a}(x)=\sum_k\alpha_{\theta ak}\cos(2\pi r_k^\top x+\phi_k)\), which creates smooth global variation across the state space.
\end{itemize}

Now we consider two more concrete data market applications.  In this part, \(\Omega=[S]\), \(\Theta=[n]\), and \(A=[m]\), and draw the type masses \(\lambda\) from a Dirichlet distribution. The common prior over the \(S\) market states is uniform.

We first consider \textbf{dataset purchase markets}, modeled as an advertiser buying audience data before choosing an advertising template.  There are \(K\) user groups.  A state is a vector
\(\omega_s=(r_{s1},\ldots,r_{sK})\), where \(r_{sk}\in[0,1]\) is the conversion rate of group \(k\) in market snapshot \(s\).  Each conversion rate is drawn from one of two beta distributions, where the two distributions represent high-conversion and low-conversion regimes, respectively.  An action \(a\) is one of \(m\) advertising templates; \(x_{ak}\in[0,1]\) denotes the coverage of template \(a\) on user group \(k\).

A type \(\theta\) is an advertiser with audience weights \(w_\theta=(w_{\theta1},\ldots,w_{\theta K})\), drawn from a Dirichlet distribution.  If advertiser \(\theta\) chooses template \(a\) in state \(\omega_s\), its utility is
\[
    u_\theta(\omega_s,a)
    =\frac{\sum_{k=1}^K w_{\theta k}x_{ak}r_{sk}}
    {\sum_{k=1}^K w_{\theta k}x_{ak}+\kappa},
\]
where \(\kappa>0\) is a small numerical constant preventing division by zero.

We then consider \textbf{model-versioning markets}, motivated by a retailer that uses a demand-forecasting model to choose a replenishment or inventory level.  The state \(\omega\in\mathbb R^d\) is the true parameter of the demand model and is drawn from a Gaussian distribution.  A model version sold by the seller releases a noisy parameter \(z=\omega+\xi\); after observing \(z\), the buyer chooses one of \(m\) discrete inventory levels, where action \(a\) corresponds to the inventory level \(q_a=(a-1)/(m-1)\) on \([0,1]\).  

For type \(\theta\), given parameter \(z\), the model outputs \(y_\theta(z)=\operatorname{sigmoid}(w_\theta^\top z+b_\theta)\), where \(w_\theta\) and \(b_\theta\) are drawn from Gaussian distributions.  The retailer chooses the inventory level \(q_{\hat a_\theta(z)}\), where \(\hat a_\theta(z)\in\arg\min_{a\in[m]}|q_a-y_\theta(z)|\).  Realized performance is evaluated against the best inventory level predicted by the true parameter \(\omega\):
\[
    u_\theta(\omega,z)=1-\left(q_{\hat a_\theta(z)}-y_\theta(\omega)\right)^2.
\]

Finally, we include sanity checks that follow the two tightness constructions in Proposition~\ref{prop:full-info-approx}.  For each \((n,m)\), we compare the two bounds \(n\) and \(1+\ln(v_{\max}/v_{\min})\), then use the corresponding construction from the proposition for the case where the logarithmic bound is the tighter bound and for the case where the \(n\)-bound is the tighter bound.  The purpose is to check whether the computed revenue reaches the active upper bound predicted by the proposition, and to illustrate that in some regimes our menu can substantially outperform selling only complete information at a single price.

\paragraph{Experimental results.}

In the experiments, denote \(n=|\Theta|\) and \(m=|A|\), we sweep \(n\in\{2,4,6,8\}\) and \(m\in\{2,4\}\).  Finite latent feature and dataset purchase use a large state space with \(S=200000\) states; continuous smooth markets and model versioning markets use state dimension \(d=2\). In the FPTAS running time \((1/\varepsilon)^{O(nm)}\), the largest reported cell \((n,m)=(8,4)\) already has \(nm=32\), which is already a significant exponential term. \(\varepsilon\) is chosen as a small value that keeps the running time reasonable while making the algorithm perform as well as possible. For reporting, \(R_{\rm alg}\) is the revenue of the FPTAS, \(R_{\rm tbs}\) is the total buyer surplus, which is a upper bound on \(\OPT\), and \(\Rfull\) is the best posted price revenue from selling only the complete information experiment. The sanity checks on tight instances use the two tightness constructions in Proposition~\ref{prop:full-info-approx}; for these instances, \(R_{\rm tbs}=k\Rfull\), where \(k=\min\{n,1+\ln(v_{\max}/v_{\min})\}\).

\begin{table}[H]
\centering
\caption{Market performance ratios. Rows index the \((n,m)\), columns index the market family, and each entry reports \(R_{\rm alg}/\Rfull\pm\) s.d.; \(R_{\rm tbs}/\Rfull\pm\) s.d., where the standard deviation is over 5 repeated runs when available.}
\label{tab:section32-main-experiments}
{\scriptsize
\setlength{\tabcolsep}{2pt}
\resizebox{\linewidth}{!}{%
\begin{tabular}{@{}lcccc@{}}
\toprule
\((n,m)\) & Finite latent-feature & Continuous smooth & Dataset purchase & Model versioning \\
\midrule
\((2,2)\) & \(1.096\pm0.119;\,1.162\pm0.114\) & \(1.035\pm0.055;\,1.040\pm0.055\) & \(0.962\pm0.584;\,1.210\pm0.231\) & \(1.114\pm0.127;\,1.136\pm0.132\) \\
\((2,4)\) & \(1.044\pm0.024;\,1.065\pm0.034\) & \(1.096\pm0.061;\,1.097\pm0.062\) & \(1.157\pm0.234;\,1.177\pm0.222\) & \(1.171\pm0.065;\,1.182\pm0.057\) \\
\((4,2)\) & \(1.193\pm0.137;\,1.360\pm0.226\) & \(1.348\pm0.201;\,1.377\pm0.229\) & \(1.049\pm0.240;\,1.240\pm0.080\) & \(1.153\pm0.095;\,1.204\pm0.080\) \\
\((4,4)\) & \(1.127\pm0.097;\,1.162\pm0.087\) & \(1.453\pm0.282;\,1.490\pm0.298\) & \(1.438\pm0.288;\,1.456\pm0.298\) & \(1.223\pm0.113;\,1.310\pm0.118\) \\
\((6,2)\) & \(1.235\pm0.116;\,1.458\pm0.130\) & \(1.443\pm0.237;\,1.475\pm0.220\) & \(1.349\pm0.345;\,1.429\pm0.302\) & \(1.197\pm0.069;\,1.288\pm0.093\) \\
\((6,4)\) & \(1.162\pm0.120;\,1.247\pm0.184\) & \(1.558\pm0.381;\,1.641\pm0.399\) & \(1.704\pm0.340;\,1.718\pm0.345\) & \(1.352\pm0.109;\,1.546\pm0.116\) \\
\((8,2)\) & \(1.245\pm0.156;\,1.415\pm0.219\) & \(1.247\pm0.235;\,1.301\pm0.267\) & \(1.571\pm0.301;\,1.643\pm0.322\) & \(1.363\pm0.101;\,1.575\pm0.224\) \\
\((8,4)\) & \(1.195\pm0.055;\,1.327\pm0.094\) & \(1.624\pm0.380;\,1.727\pm0.419\) & \(1.243\pm0.016;\,1.569\pm0.064\) & \(1.390\pm0.210;\,1.667\pm0.270\) \\
\bottomrule
\end{tabular}}
}
\end{table}

\begin{table}[H]
\centering
\caption{Running time. Rows index the \((n,m)\), columns index the market family, and each entry reports \(\varepsilon;\) running time in seconds.}
\label{tab:section32-main-runtime}
{
\setlength{\tabcolsep}{3pt}
\resizebox{0.85\linewidth}{!}{%
\begin{tabular}{@{}lcccc@{}}
\toprule
\((n,m)\) & Finite latent-feature & Continuous smooth & Dataset purchase & Model versioning \\
\midrule
\((2,2)\) & \(.001;\,1.93\mathrm{s}\) & \(.005;\,1.79\mathrm{s}\) & \(.02;\,0.23\mathrm{s}\) & \(.01;\,1.72\mathrm{s}\) \\
\((2,4)\) & \(.001;\,6.08\mathrm{s}\) & \(.005;\,5.36\mathrm{s}\) & \(.02;\,0.93\mathrm{s}\) & \(.01;\,7.40\mathrm{s}\) \\
\((4,2)\) & \(.001;\,4.26\mathrm{s}\) & \(.005;\,3.78\mathrm{s}\) & \(.02;\,0.83\mathrm{s}\) & \(.01;\,3.41\mathrm{s}\) \\
\((4,4)\) & \(.001;\,9.47\mathrm{s}\) & \(.005;\,12.96\mathrm{s}\) & \(.02;\,6.21\mathrm{s}\) & \(.01;\,27.87\mathrm{s}\) \\
\((6,2)\) & \(.001;\,4.78\mathrm{s}\) & \(.01;\,9.17\mathrm{s}\) & \(.02;\,4.59\mathrm{s}\) & \(.01;\,12.07\mathrm{s}\) \\
\((6,4)\) & \(.001;\,13.00\mathrm{s}\) & \(.01;\,19.90\mathrm{s}\) & \(.02;\,5.82\mathrm{s}\) & \(.02;\,37.60\mathrm{s}\) \\
\((8,2)\) & \(.001;\,11.68\mathrm{s}\) & \(.01;\,27.49\mathrm{s}\) & \(.02;\,17.91\mathrm{s}\) & \(.02;\,24.03\mathrm{s}\) \\
\((8,4)\) & \(.001;\,20.98\mathrm{s}\) & \(.01;\,38.07\mathrm{s}\) & \(.1;\,185.89\mathrm{s}\) & \(.10;\,145.52\mathrm{s}\) \\
\bottomrule
\end{tabular}}
}
\end{table}

These results show that the proposed algorithm is computationally efficient across all tested market families: even for the largest reported instances with \((n,m)=(8,4)\).

\begin{table}[H]
\centering
\caption{Sanity checks on tight instances where \(1+\ln(v_{\max}/v_{\min})<n\)}
\label{tab:section32-log-small}
{\normalsize
\setlength{\tabcolsep}{5pt}
\begin{tabular}{@{}lcccccccc@{}}
\toprule
\((n,m)\) & \((2,2)\) & \((2,4)\) & \((4,2)\) & \((4,4)\) & \((6,2)\) & \((6,4)\) & \((8,2)\) & \((8,4)\) \\
\midrule
\(1+\ln(v_{\max}/v_{\min})\) & 1.303 & 1.256 & 1.531 & 1.384 & 1.675 & 1.553 & 1.978 & 1.724 \\
\(R_{\rm alg}/\Rfull\) & 1.262 & 1.226 & 1.487 & 1.360 & 1.631 & 1.524 & 1.913 & 1.688 \\
\(R_{\rm alg}/R_{\rm ub}\) & 0.968 & 0.976 & 0.971 & 0.983 & 0.974 & 0.981 & 0.967 & 0.979 \\
\bottomrule
\end{tabular}
}
\end{table}

\begin{table}[H]
\centering
\caption{Sanity checks on tight instances where \(1+\ln(v_{\max}/v_{\min})>n\)}
\label{tab:section32-n-small}
{\normalsize
\setlength{\tabcolsep}{5pt}
\begin{tabular}{@{}lcccccccc@{}}
\toprule
\((n,m)\) & \((2,2)\) & \((2,4)\) & \((4,2)\) & \((4,4)\) & \((6,2)\) & \((6,4)\) & \((8,2)\) & \((8,4)\) \\
\midrule
\(1+\ln(v_{\max}/v_{\min})\) & 5.255 & 5.820 & 11.949 & 14.530 & 20.817 & 27.492 & 27.406 & 33.974 \\
\(R_{\rm alg}/\Rfull\) & 1.972 & 1.984 & 3.896 & 3.956 & 5.886 & 5.970 & 7.816 & 7.928 \\
\(R_{\rm alg}/R_{\rm ub}\) & 0.986 & 0.992 & 0.974 & 0.989 & 0.981 & 0.995 & 0.977 & 0.991 \\
\bottomrule
\end{tabular}
}
\end{table}

Under the market instances we construct, \(R_{\rm alg}\)  gives a clear improvement over \(\Rfull\) in many cases. However, there are still some cases where the ratio is relatively small, because \(\Rfull\) is already very close to \(R_{\rm tbs}\). For the sanity checks on tight instances, the \(R_{\rm alg}/R_{\rm tbs}\) ratios are high in both tables, which shows that our constructions are correct.

\section{Omitted Results and Proofs in Section~\ref{sec:blackwell}} \label{app:omitted-proofs-sec5}

\begin{proposition} \label{prop:certified-deployment-threshold}
 The following two statements hold.
 \begin{enumerate}
  \item Let \(U\) be a finite set of fields. For each \(R \subseteq U\), let \(E_R\) be the deterministic experiment that reveals the coordinates in \(R\). Fix nonempty target sets \(R_1, \ldots, R_n \subseteq U\) and values \(v_1, \ldots, v_n \in (0, 1)\). For every finite family \(\mathcal R \subseteq 2^U\) of products that the seller may offer, there exists a Bayesian decision problem with a finite action space and payoffs in \([0, 1]\) such that, for every type \(j \in [n]\) and every finite bundle \(B\) from \(\mathcal R\),
  \[ V_j(E_B) = \begin{cases}
   v_j, & E_B \succeqB E_{R_j}, \\
   0, & E_B \not\succeqB E_{R_j}.
  \end{cases} \]
  \item Let \(\Omega = \{0, 1\}\) have a full support prior, and let \(T\) be the experiment that reveals \(\omega\). For \(\sigma > 0\), let \(E_\sigma\) be the experiment that sends \(Y = \omega + \sigma Z\), where \(Z \sim N(0, 1)\). There do not exist a finite action Bayesian decision problem with payoffs in \([0, 1]\) and a value \(v > 0\) such that for every \(E \in \{T\} \cup \{E_\sigma : \sigma > 0\}\),
  \[ V(E) = \begin{cases}
   v, & E \succeqB T, \\
   0, & E \not\succeqB T
  \end{cases} \]
 \end{enumerate}
\end{proposition}

The finite restriction in the positive part is natural: the types of  appropriate deterministic queries that the seller can offer in the real data market are limited, and repeated purchases of a deterministic product do not create new information, so a finite deterministic menu induces only finitely many Blackwell equivalence classes of bundles.

\begin{proof}
 We first prove part~(1). Let \(L\) be an integer such that \(1 / L \leqslant 1 - v_j\) for every \(j \in [n]\). Let the state space be \(\Omega = [L]^U\) with the uniform prior, and write \(\omega_R\) for the restriction of \(\omega\) to coordinates in \(R\). For each \(R \subseteq U\), the experiment \(E_R\) sends the deterministic signal \(\omega_R\). If a bundle \(B = (R^1, \ldots, R^k)\) is purchased, then \(E_B \simB E_{R(B)}\), where \(R(B) = \bigcup_{\ell = 1}^k R^\ell\), with \(R(B) = \emptyset\) for the empty bundle. Therefore \(E_B \succeqB E_{R_j}\) if and only if \(R_j \subseteq R(B)\).

 The common action space is $A = \{a^0\} \cup \{a_{j, x}: j \in [n],\ x \in [L]^{R_j}\}$. For type \(j\), define \(u_j(\omega, a^0) = 1 - v_j\), \(u_j(\omega, a_{j, x}) = 1\{\omega_{R_j} = x\}\), and \(u_j(\omega, a_{\ell, x}) = 0\) for every \(\ell \neq j\). All payoffs lie in \([0, 1]\). Without information, every action \(a_{j, x}\) succeeds with probability \(L^{-|R_j|} \leqslant 1 / L \leqslant 1 - v_j\), and every action intended for another type gives payoff zero. Hence type \(j\) chooses \(a^0\), and \(U_j^0 = 1 - v_j\).

 Consider any bundle \(B\). If \(R_j \subseteq R(B)\), then after observing \(\omega_{R(B)}\), type \(j\) knows \(\omega_{R_j}\) and chooses the action \(a_{j, \omega_{R_j}}\). Thus \(W_j(E_B) = 1\), and \(V_j(E_B) = v_j\). If \(R_j \not\subseteq R(B)\), then at least one coordinate in \(R_j\) remains unrevealed. Conditional on any realized signal of \(E_B\), the vector \(\omega_{R_j}\) is still uniform over at least \(L\) possible values. Thus every action \(a_{j, x}\) has conditional expected payoff at most \(1 / L \leqslant 1 - v_j\), and $a^0$ remains optimal after every signal. Hence \(W_j(E_B) = 1 - v_j\), and \(V_j(E_B) = 0\). This proves the threshold formula on every finite bundle from \(\mathcal R\).

 We now prove part~(2). Fix any finite action Bayesian decision problem with payoffs in \([0, 1]\). For \(q \in [0, 1]\), define $g(q) = \max_{a \in A}\bigl((1 - q)u(0, a) + q u(1, a)\bigr)$. Since \(A\) is finite, \(g\) is continuous. The complete information payoff is \(W(T) = \mu(0)g(0) + \mu(1)g(1)\).

 Let \(\pi_\sigma(Y) = \Pr(\omega = 1 \mid Y)\) be the posterior under \(E_\sigma\). The payoff from \(E_\sigma\) is \(W(E_\sigma) = \mathbb E[g(\pi_\sigma(Y))]\). The likelihood ratio satisfies
 \[ \frac{\pi_\sigma(y)}{1 - \pi_\sigma(y)} = \frac{\mu(1)}{\mu(0)}\exp\left(\frac{2y - 1}{2\sigma^2}\right). \]
 Conditional on \(\omega = 0\), we have \(Y = \sigma Z\), so the exponent above is \(Z / \sigma - 1 / (2\sigma^2)\), which tends to \(-\infty\) almost surely. Conditional on \(\omega = 1\), we have \(Y = 1 + \sigma Z\), so the exponent is \(1 / (2\sigma^2) + Z / \sigma\), which tends to \(+\infty\) almost surely when \(\sigma \to 0\). Hence \(\pi_\sigma(Y) \to 0\) almost surely conditional on \(\omega = 0\), and \(\pi_\sigma(Y) \to 1\) almost surely conditional on \(\omega = 1\). Since \(g\) is bounded and continuous, dominated convergence gives \(W(E_\sigma) \to W(T)\) as \(\sigma \to 0\). Therefore \(V(E_\sigma) \to V(T)\), because the prior payoff \(U^0\) is independent of \(\sigma\).

 It remains to show that \(E_\sigma \not\succeqB T\) for every \(\sigma > 0\). Indeed, if \(E_\sigma\) Blackwell dominated \(T\), then there would be a measurable function \(h:\mathbb R \to [0, 1]\) such that \(\mathbb E[h(Y) \mid \omega = 0] = 1\) and \(\mathbb E[h(Y) \mid \omega = 1] = 0\). However, the distributions \(N(0, \sigma^2)\) and \(N(1, \sigma^2)\) are mutually absolutely continuous, so the first equality implies \(h(Y) = 1\) almost surely under both states, contradicting the second equality. Thus the noisy continuous experiment never Blackwell dominates the exact target experiment.

 If an exact threshold representation with value \(v > 0\) existed on \(\{T\} \cup \{E_\sigma : \sigma > 0\}\), then \(V(T) = v\) and \(V(E_\sigma) = 0\) for every \(\sigma > 0\). This contradicts \(V(E_\sigma) \to V(T) = v\). Hence no such Bayesian decision problem exists.
\end{proof}

\begin{proof}[Proof of Proposition~\ref{prop:unified}]
 We first prove necessity. Suppose monotonicity fails, so that \(E\succeqB F\) but \(p(E) < p(F)\), which constructs an arbitrage. Suppose instead that subadditivity fails for some \(E_1, E_2\), so that $p(E_1 \otimes E_2) > p(E_1) + p(E_2)$. Let $F = E_1 \otimes E_2$. Then \(E_1 \otimes E_2 \succeqB F\) but \(p(F) > p(E_1) + p(E_2)\), which again constructs an arbitrage.

 For sufficiency, consider any bundle $E_1 \otimes \cdots \otimes E_k \succeqB F$. By Blackwell monotonicity and repeated parallel subadditivity, $p(F) \leqslant p(E_1 \otimes \cdots \otimes E_k) \leqslant p(E_1) + p(E_2 \otimes \cdots \otimes E_k) \leqslant \cdots \leqslant \sum_{r = 1}^k p(E_r)$. Thus no arbitrage exists.
\end{proof}

\begin{proof}[Proof of Theorem~\ref{thm:threshold-blackwell}]
 1 \(\Rightarrow\) 2 is trivial according to the definition of Blackwell dominance. For 2 \(\Rightarrow\) 1, we show that profitable deviations are exactly the cheaper Blackwell substitutes. Under the threshold assumption, if \(E_B \not\succeqB E_\theta\), then \(V_\theta(E_B) = 0\), and since prices are nonnegative and \(t_\theta \leqslant v_\theta\), such a bundle cannot be better than the assigned net utility \(v_\theta - t_\theta\). Hence AF is exactly the requirement that no bundle \(B\) with \(E_B \succeqB E_\theta\) has total price less than \(t_\theta\).
\end{proof}

\begin{proof}[Proof of Proposition~\ref{prop:utility-specific-gap}]
 Let the state of the world be \(\omega = (X, Y, Z) \in \{0, 1\}^3\), where \(X, Y, Z\) are independent bits. There are three buyer types, each buyer type has the same uniform prior on \(\{0, 1\}^3\), and the three type masses are \(\lambda_A = \lambda_B = \lambda_H = 1\). Type \(A\) wants to guess \(X\): it receives payoff \(1/2\) for a correct guess and \(0\) otherwise, so the maximum value of any experiment for type \(A\) is \(1 / 4\). Type \(B\) is symmetric and wants to guess \(Y\), again with maximum value \(1 / 4\). Type \(H\) wants to guess the ordered pair \((X, Y)\). It receives payoff \(1\) if both bits are guessed correctly and \(0\) otherwise, so the maximum value of any experiment for type \(H\) is \(3 / 4\).

 Consider three experiments $E_X, E_Y, F$. Experiment \(E_X\) reveals \(X\), \(E_Y\) reveals \(Y\), and \(F\) reveals the full state \((X, Y, Z)\). The posted prices are $t_X = 1/4, t_Y = 1/4, t_F = 3/4$. If type \(A\) buys \(E_X\), type \(B\) buys \(E_Y\), and type \(H\) buys \(F\), then the Blackwell AF constraints (note that $E_X$ and $E_Y$ cannot be used to construct a cheaper Blackwell substitute for $F$) and IR constraints are satisfied, and the menu obtains total revenue $1/4 + 1/4 + 3/4 = 5/4$. No IR menu can obtain more, because the price for each type cannot exceed its maximum value. Therefore the Blackwell AF optimum is \(5 / 4\).

 However, the AF constraints are not satisfied. Type \(H\) can buy \(E_X\) and \(E_Y\), learn \((X, Y)\), thus obtaining the same value as from \(F\), but pay only \(1/2\). We now prove that this deviation forces a revenue loss for \emph{any} AF menu. Consider any direct menu with assigned experiments \(E_A, E_B, E_H\) and prices \(t_A, t_B, t_H\) satisfying AF. Let \(S = t_A+t_B\). IR gives $t_A \leqslant 1/4, t_B\leqslant 1/4, t_H\leqslant 3/4$. For any experiment \(E\), let \(p_X(E)\) be the maximum probability of correctly guessing \(X\) after observing \(E\) under the common prior. Then the value of \(E\) for type \(A\) is $V_A(E) = (p_X(E)-1/2)/2$.
 Since type \(A\) buys \(E_A\) at price \(t_A\), IR implies \(V_A(E_A)\geqslant t_A\), hence $p_X(E_A)\geqslant 1/2 + 2t_A$. Similarly, if \(p_Y(E_B)\) is the best probability of guessing \(Y\) after \(E_B\), then $p_Y(E_B) \geqslant 1/2 + 2t_B$. 

 Now let type \(H\) buy the bundle \(E_A\otimes E_B\). It can use the optimal rule for guessing \(X\) from \(E_A\) and the optimal rule for guessing \(Y\) from \(E_B\), and then output the pair of guesses. Let \(C_X\) be the event that the first rule guesses \(X\) correctly and \(C_Y\) the event that the second rule guesses \(Y\) correctly. By the union bound, $\Pr[C_X \cap C_Y] \geqslant (1/2 + 2t_A) + (1/2 + 2t_B) - 1 = 2S$. Type \(H\) can also ignore all purchased signals and get outside success probability \(1/4\). Therefore $V_H(E_A\otimes E_B) \geqslant \max\{0,2S-1/4\}$. AF for type \(H\) requires $V_H(E_H)-t_H \geqslant V_H(E_A \otimes E_B)-S$.
 Since \(V_H(E_H) \leqslant 3/4\), we get $3/4 - t_H \geqslant \max\{0, 2S - 1/4\}-S$. If \(S < 1/8\), then total revenue is at most \(S + 3/4 < 7/8 < 1\). If \(S \geqslant 1/8\), then the maximum term equals \(2S - 1/4\), so $3/4 - t_H \geqslant S - 1/4$, hence $t_H \leqslant 1 - S$. Thus $t_A+t_B+t_H=S+t_H\leqslant 1$.
 So every AF menu has total revenue at most \(1\). The bound is tight: post \((E_X,1/4)\), \((E_Y,1/4)\), and \((F,1/2)\). Therefore the AF optimum is \(1\).
\end{proof}

\begin{proof}[Proof of Proposition~\ref{prop:query-blackwell}]
 We first prove part~(1). Assume that $E_{Q_1}\succeqB E_{Q_2}$. Then there is a post processing map from the output of $Q_1$ to the output distribution of $Q_2$. If $Q_1(D) = Q_1(D')$, the input to this post processing map is identical under $D$ and $D'$. Therefore the post processed output distribution must also be identical. Since $E_{Q_2}$ is deterministic, this implies $Q_2(D) = Q_2(D')$. Hence every cell of $P_{Q_1}$ is contained in a cell of $P_{Q_2}$, so $P_{Q_1} \succeq P_{Q_2}$.

 Conversely, suppose $P_{Q_1} \succeq P_{Q_2}$. Then $Q_2$ is constant on every cell of $P_{Q_1}$. Hence there exists a deterministic map $g: \mathcal A_{Q_1} \to \mathcal A_{Q_2}$ such that $g(Q_1(D)) = Q_2(D)$ for every $D\in\mathcal D$. Applying $g$ to the signal of $E_{Q_1}$ produces exactly the signal of $E_{Q_2}$, so $E_{Q_1} \succeqB E_{Q_2}$.

 We now prove part~(2). Define the pair query $\widetilde Q(D) = \left(Q_1(D), Q_2(D)\right)$. We first show that $E_{Q_1} \otimes E_{Q_2} = E_{\widetilde Q}$. Indeed, for every \(D \in \mathcal D\) and every \((a_1, a_2)\in \mathcal A_{Q_1} \times \mathcal A_{Q_2}\),
 \[ \begin{aligned}
   K_{E_{Q_1} \otimes E_{Q_2}}\left((a_1, a_2) \mid D\right) &= K_{Q_1}(a_1 \mid D)\, K_{Q_2}(a_2 \mid D) = \mathbf 1\{Q_1(D) = a_1\}\mathbf 1\{Q_2(D) = a_2\} \\
   &= \mathbf 1\{\widetilde Q(D) = (a_1, a_2)\} = K_{E_{\widetilde Q}}\left((a_1, a_2) \mid D\right).
  \end{aligned} \]

 Next we identify the partition induced by \(\widetilde Q\). For \(D,D'\in\mathcal D\), $D \sim_{\widetilde Q} D' \iff \widetilde Q(D) = \widetilde Q(D') \iff Q_1(D) = Q_1(D')\ \text{and}\ Q_2(D) = Q_2(D')$. This is exactly the condition that \(D\) and \(D'\) lie in the same cell of \(P_{Q_1} \vee P_{Q_2}\). Therefore $P_{\widetilde Q} = P_{Q_1} \vee P_{Q_2}$.
 
 If \(P_Q=P_{Q_1}\vee P_{Q_2}\), then \(P_Q=P_{\widetilde Q}\). Applying part~(1) first to \((Q,\widetilde Q)\) and then to \((\widetilde Q,Q)\), we obtain $E_Q \simB E_{\widetilde Q} = E_{Q_1} \otimes E_{Q_2}$. Conversely, suppose $E_Q \simB E_{Q_1} \otimes E_{Q_2} = E_{\widetilde Q}$. Applying part~(1) in both directions gives \(P_Q = P_{\widetilde Q} = P_{Q_1} \vee P_{Q_2}\), proving part~(2).
\end{proof}

\begin{proof}[Proof of Proposition~\ref{prop:gaussian-matrix}]
 We first prove part~(1). Suppose $J_1 \succeq J_2$. Since inversion reverses the Loewner order on positive definite matrices, we have $J_2^{-1} - J_1^{-1} \succeq 0$. Given a sample $Y_1 = \theta + \varepsilon_1$ from $G_{J_1}$, draw independent noise $\eta \sim \mathcal N(0, J_2^{-1} - J_1^{-1})$ and output $Y_1 + \eta$. Conditional on $\theta$, this output has distribution $\mathcal N(\theta, J_1^{-1} + J_2^{-1} - J_1^{-1}) = \mathcal N(\theta, J_2^{-1})$. Thus $G_{J_2}$ is obtained from $G_{J_1}$ by post processing, and $G_{J_1} \succeqB G_{J_2}$.

 Conversely, suppose \(G_{J_1} \succeqB G_{J_2}\). To prove \(J_1 \succeq J_2\), it is sufficient to show that $v^\top J_1 v \geqslant v^\top J_2 v, \forall v \in \R^d$. Fix \(v \in \R^d\). The case \(v = 0\) is trivial, so assume \(v \neq 0\). Consider the binary decision problem with prior $\Pr(\theta = v) = \Pr(\theta = -v) = \frac12$, action set \(A = \{+, -\}\), and utility $u(\theta, +) = \mathbf 1\{\theta = v\}, u(\theta, -) = \mathbf 1\{\theta = -v\}$. Thus the buyer's objective is simply to identify whether the state is \(v\) or \(-v\).

 For \(r \in \{1, 2\}\), let \(Y_r = \theta + \varepsilon_r\), where \(\varepsilon_r \sim \mathcal N(0, J_r^{-1})\). Write $q_r = v^\top J_r v$. Under the two possible states, the signal \(Y_r\) has distributions $Y_r \mid (\theta = v) \sim \mathcal N(v, J_r^{-1}), Y_r \mid (\theta = -v) \sim \mathcal N(-v, J_r^{-1})$. The log-likelihood ratio between these two Gaussian distributions is
 \[ \log\frac{f_{r, +}(y)}{f_{r, -}(y)} = -\frac12(y - v)^\top J_r(y - v) + \frac12(y + v)^\top J_r(y + v) = 2\, v^\top J_r y. \]
 Since the prior probabilities are equal, a Bayes optimal decision rule chooses \(+\) if and only if \(v^\top J_r y\geqslant 0\). Define $Z_r = v^\top J_r Y_r$. When \(\theta = v\), we have $Z_r = v^\top J_r v + v^\top J_r\varepsilon_r$, so $Z_r \mid (\theta = v) \sim \mathcal N(q_r, q_r)$, since $\operatorname{Var}(v^\top J_r\varepsilon_r) = v^\top J_rJ_r^{-1}J_r v = v^\top J_r v = q_r$. By symmetry, $Z_r \mid (\theta = -v) \sim \mathcal N(-q_r, q_r)$. Therefore the expected utility from \(G_{J_r}\) in this decision problem is $\Pr(Z_r \geqslant 0 \mid \theta = v) = \Pr\!\left(\mathcal N(q_r, q_r) \geqslant 0\right) = \Phi(\sqrt{q_r})$, where \(\Phi\) is the standard normal cdf. Since \(G_{J_1} \succeqB G_{J_2}\), Blackwell's theorem implies that \(G_{J_1}\) gives weakly higher value than \(G_{J_2}\) in every prior and decision problem. Hence $\Phi(\sqrt{q_1}) \geqslant \Phi(\sqrt{q_2})$. Since \(\Phi\) is strictly increasing on \(\R\), this yields \(q_1 \geqslant q_2\), that is, $v^\top J_1 v \geqslant v^\top J_2 v$.

 We next prove part~(2). Let $Y_r = \theta + \varepsilon_r$, where $\varepsilon_r \sim \mathcal N(0, J_r^{-1})$ for $r = 1, 2$, with $\varepsilon_1, \varepsilon_2$ independent conditional on $\theta$. Set $J = J_1 + J_2$ and define the weighted average $T = J^{-1}(J_1Y_1 + J_2Y_2)$. Then $T = \theta + J^{-1}(J_1\varepsilon_1 + J_2\varepsilon_2)$, and the covariance of the noise term is $J^{-1}(J_1J_1^{-1}J_1 + J_2J_2^{-1}J_2)J^{-1} = J^{-1}(J_1 + J_2)J^{-1} = J^{-1}$. Hence $T$ is distributed as $G_J$, so $G_{J_1} \otimes G_{J_2} \succeqB G_J$.

 For the reverse Blackwell domination, define the residual $R = Y_1 - Y_2 = \varepsilon_1 - \varepsilon_2$. The covariance between $T - \theta$ and $R$ is $\operatorname{Cov}\left(J^{-1}(J_1\varepsilon_1 + J_2\varepsilon_2), \varepsilon_1 - \varepsilon_2\right) = J^{-1}(J_1J_1^{-1} - J_2J_2^{-1}) = 0$. Since the vector is jointly Gaussian, $T$ and $R$ are independent conditional on $\theta$. Moreover, $R \sim \mathcal N(0, J_1^{-1} + J_2^{-1})$, which does not depend on $\theta$. Hence, given only a draw $T$ from $G_J$, a garbling can independently sample such an $R$ and reconstruct $Y_1 = T + J^{-1}J_2R, Y_2 = T - J^{-1}J_1R$ (these equations are obtained by solving $T = J^{-1}(J_1Y_1 + J_2Y_2)$ and $R = Y_1 - Y_2$). Therefore $G_J \succeqB G_{J_1} \otimes G_{J_2}$, proving Blackwell equivalence.
\end{proof}

\begin{proof}[Proof of Proposition~\ref{prop:mi-noarb}]

 We first prove Blackwell monotonicity. Suppose \(E \succeqB F\). By definition, there exists a garbling kernel $\Gamma: S_E \to \Delta(S_F)$
 such that we can generate a signal \(Y_E\) from \(E\) conditional on the state \(X\), and then, conditional on \(Y_E = s\), generate \(Y_F\) from \(\Gamma(\cdot \mid s)\). Hence the triple \((X, Y_E, Y_F)\) forms a Markov chain $X \to Y_E \to Y_F$. By the data processing inequality, we have $I(X; Y_E) \geqslant I(X; Y_F)$. Equivalently, $p_\mu(E) \geqslant p_\mu(F)$.

 We next prove subadditivity. Let \(E_1\) and \(E_2\) be two experiments. Conditional on the state \(X = \omega\), let \(Y_1\) be drawn from \(K_{E_1}(\cdot \mid \omega)\) and let \(Y_2\) be drawn independently from \(K_{E_2}(\cdot \mid \omega)\). By the definition of \(E_1 \otimes E_2\), we have $p_\mu(E_1 \otimes E_2) = I(X; Y_1, Y_2)$. Applying the chain rule for mutual information gives $I(X; Y_1, Y_2) = I(X; Y_1) + I(X; Y_2 \mid Y_1)$. It remains to show that $I(X; Y_2 \mid Y_1) \leqslant I(X; Y_2)$.

 Since \(Y_1\) and \(Y_2\) are conditionally independent given \(X\), we have $I(Y_1; Y_2 \mid X) = 0$. Now apply the chain rule to the triple \((Y_2, X, Y_1)\) in two different ways: $I(Y_2; X, Y_1) = I(Y_2; X) + I(Y_2; Y_1 \mid X)$, and $I(Y_2; X, Y_1) = I(Y_2; Y_1) + I(Y_2; X \mid Y_1)$. Since \(I(Y_2; Y_1 \mid X) = 0\), comparison of these two identities yields $I(X; Y_2 \mid Y_1) = I(X; Y_2) - I(Y_1; Y_2) \leqslant I(X; Y_2)$. Thus we have $I(X; Y_1, Y_2) \leqslant I(X; Y_1) + I(X; Y_2)$. Equivalently, $p_\mu(E_1 \otimes E_2) \leqslant p_\mu(E_1) + p_\mu(E_2)$.
\end{proof}

\begin{proof}[Proof of Corollary~\ref{cor:mi-query-gaussian}]
 For part~(1), a deterministic query satisfies \(H(Q(D) \mid D) = 0\). Hence $p_\mu(E_Q) = I(D; Q(D)) = H(Q(D)) - H(Q(D) \mid D) = H(Q(D))$.

 For part~(2), let the signal of \(G_J\) be $Y_{G_J} = \theta + \varepsilon, \varepsilon \sim \mathcal N(0, J^{-1})$, where \(\varepsilon\) is independent of \(\theta\). Then $Y_{G_J} \sim \mathcal N(0, \Sigma_0 + J^{-1})$. Therefore $p_\mu(G_J) = I(\theta; Y_{G_J}) = h(Y_{G_J}) - h(Y_{G_J} \mid \theta)$. Using the Gaussian differential entropy formula $h(Z) = \frac12 \log((2\pi e)^d \det\!\Sigma_Z)$, we obtain $p_\mu(G_J) = \frac12 \log\det\!\left((\Sigma_0 + J^{-1})J\right) = \frac12 \log\det(I + \Sigma_0 J)$.
\end{proof}

\begin{proof}[Proof of Proposition~\ref{prop:threshold-target-lp}]
  By Proposition~\ref{prop:bundle-revelation}, it is enough to consider direct AF menus. Fix such a menu \(\mathcal M=\{(E_\theta, t_\theta)\}_{\theta \in \Theta}\). For each target \(i\), define its indirect cost in \(\mathcal M\) by
  \[ \rho_i = \inf\left\{\sum_{\theta \in \Theta} c_\theta t_\theta: \bigotimes_{\theta \in \Theta} E_\theta^{\otimes c_\theta} \succeqB T_i,\; 
  c \in \mathbb Z_+^\Theta\right\}. \]
  If no finite bundle can synthesize \(T_i\), set \(\rho_i = \infty\). Consider a type \(\theta \in \Theta_i\) that is willing to buy its assigned product \((E_\theta, t_\theta)\) under \(\mathcal M\). Then we have \(E_\theta \succeqB T_i\). According to the definition of $\rho_i$, \(\rho_i \leqslant t_\theta\). If \(\rho_i < t_\theta\), then it contradicts AF. Hence every positive payment made by a type in \(\Theta_i\) equals \(\rho_i\), and the revenue extracted from target class \(i\) is at most \(\rho_iD_i(\rho_i)\), interpreted as zero when \(\rho_i = \infty\).

  We claim that for any direct AF menu \(\mathcal M\), its revenue is no larger than that of another menu that only sells the target products \(T_1, \ldots, T_M\) with price vector \(p\), where \(p_i = \rho_i\) if \(\rho_i < \infty\), and \(p_i = P\) otherwise, for some \(P\) larger than all finite \(\rho_j\) and all values \(v_\theta\). We first prove that \(p\) is feasible for~\eqref{eq:threshold-target-demand-program}: fix \(i \in [M]\) and \(c \in \mathcal C_i\), so \(T^c \succeq_B T_i\). If \(c_j > 0\) for some \(j\) with \(\rho_j = \infty\), then \(c \cdot p \geqslant P \geqslant p_i\), and the constraint holds. Then consider the case where \(c_j > 0\) only for targets \(j\) with \(\rho_j < \infty\). If \(\|c\|_1 = 0\), then the empty bundle Blackwell dominates \(T_i\), so \(p_i = \rho_i = 0 \leqslant c \cdot p\). If \(\|c\|_1 > 0\), then for every \(\varepsilon > 0\) and every \(j\) with \(c_j > 0\), by the definition of \(\rho_j\), there exists a finite bundle \(B^j\) in \(\mathcal M\) such that \(E_{B^j} \succeq_B T_j\) and \(t_{B^j} \leqslant \rho_j + \varepsilon / \|c\|_1\). Taking \(c_j\) copies of each \(B^j\) yields a finite bundle \(B\) in \(\mathcal M\) satisfying \(E_B \succeq_B T^c \succeq_B T_i\), with total price at most \(\sum_j c_j \rho_j + \varepsilon\). Therefore \(\rho_i \leqslant \sum_j c_j \rho_j + \varepsilon\). Letting \(\varepsilon \to 0\), we obtain \(\rho_i \leqslant \sum_j c_j \rho_j\). In particular, \(\rho_i < \infty\), and thus \(p_i = \rho_i \leqslant \sum_j c_j p_j = c \cdot p\). Hence \(p\) satisfies all constraints of Program~(4). Also, the revenue of the target-product menu with price vector \(p\) is \(\sum_i p_i D_i(p_i)\), which is exactly the objective of Program~(4).
  
  Conversely, if \(p\) is feasible for the program~\eqref{eq:threshold-target-demand-program}, then directly selling the target menu \(\{(T_i, p_i): i \in [M]\}\) is AF and its revenue is exactly \(\sum_i p_iD_i(p_i)\). Therefore, the optimization problem over direct Blackwell AF menus is equivalent to the program~\eqref{eq:threshold-target-demand-program}.
\end{proof}

\begin{proposition} \label{prop:threshold-target-lp-hardness}
  Unless \(P = NP\), there is no PTAS for \eqref{eq:threshold-target-demand-program}.
\end{proposition}

\begin{proof}
  Given the equivalence established in Proposition~\ref{prop:threshold-target-lp}, we only need to prove APX-hardness of Program~\eqref{eq:unrestricted-program} under the threshold regime. We reduce from Maximum Independent Set on graphs with maximum vertex degree being a fixed constant \(\Delta \geqslant 3\). Let $G = (V, E), N := |V|, M := |E|$. Let $\alpha$ be the size of a maximum independent set in \(G\).

 \paragraph{Step 1: Construct the threshold instance.} Let the state space be $\Omega = \{\varnothing\} \cup \left\{\{v\}: v \in V\right\}$. For every subset \(R \subseteq V\), define the deterministic subset query $Q_R(\omega) := \omega \cap R, \omega \in \Omega$, and let \(E_R\) be the induced experiment. Let the buyer type space be \(\Theta = V \cup E\). For each type \(\theta \in \Theta\), define its target set \(T_\theta \subseteq V\) and target value \(v_\theta\) by
 \[ T_\theta = \begin{cases}
  \{v\}, & \theta = v \in V,\\
  \{u,v\}, & \theta = \{u, v\} \in E,
  \end{cases} \quad v_\theta = \begin{cases}
  1 / 3, & \theta \in V, \\
  1, & \theta \in E.
  \end{cases} \]
 
 We assume the prior mass of each type is \(1 / (N + M)\). To keep the algebra simple, we sometimes work with the unnormalized revenue objective $\widetilde R = \sum_{\theta \in \Theta} t_\theta$, and denote \(\widetilde R^*(G)\) as the optimal unnormalized revenue in the instance constructed from \(G\).

 \paragraph{Step 2: Lower bound from an independent set.} Let \(I \subseteq V\) be a maximum independent set. Consider the menu $\mathcal M_I = \{(E_{\{v\}}, 1 / 3): v \in I\} \cup \{(E_e, 1): e \in E\}$. We claim that \(\mathcal M_I\) is AF. For a vertex type \(v \in I\), any bundle that reaches the target \(E_{\{v\}}\) must contain \(\{v\}\). Such a bundle must therefore contain either the singleton query \(E_{\{v\}}\), which costs \(1 / 3\), or an edge query containing \(v\), which costs at least \(1\). Hence no bundle is cheaper than the intended purchase. For an edge type \(e = \{u, v\}\), any bundle that reaches the target \(E_{\{u, v\}}\) must contain both \(u\) and \(v\). Since \(I\) is independent, the menu does not contain both singleton queries \(E_{\{u\}}\) and \(E_{\{v\}}\). Therefore any bundle whose union contains \(\{u, v\}\) must contain at least one edge query, and hence costs at least \(1\). The intended edge query \(E_e\) itself costs exactly \(1\), so there is no profitable deviation. Thus \(\mathcal M_I\) is AF and earns unnormalized revenue $\widetilde R(\mathcal M_I) = M + \frac{|I|}{3} = M + \frac{\alpha}{3}$. Hence
 \begin{equation} \label{eq:mis-lower-bound}
  \widetilde R^*(G) \geqslant M + \frac{\alpha}{3}.
 \end{equation}

 \paragraph{Step 3: Upper bound for an arbitrary direct AF menu.} Let \(\mathcal M = \{(F_\theta,p_\theta)\}_{\theta\in\Theta}\) be any direct AF menu. Define $I = \{v \in V: F_v \succeq_B E_{\{v\}}\}$. Write \(E[I] = \{\{u, v\} \in E : u, v \in I\}\) for the set of edges with both endpoints in \(I\), and \(G[I] = (I, E[I])\) for the induced subgraph on \(I\). If \(v \notin I\), then vertex type \(v\) obtains value zero from its intended product, so individual rationality forces \(p_v = 0\). If \(v \in I\), then its value is \(1 / 3\), so individual rationality gives \(p_v \leqslant 1 / 3\). Therefore revenue from vertex types is at most
 \begin{equation} \label{eq:vertex-revenue-upper}
  \sum_{v \in V} p_v \leqslant \frac{|I|}{3}.
 \end{equation}

 Now fix an edge \(e = \{u, v\} \in E\). If both \(u, v \in I\), then $F_u \otimes F_v \succeq_B E_{\{u\}} \otimes E_{\{v\}} \sim_B E_{\{u, v\}}$. Thus edge type \(e\) could buy the bundle \((F_u, F_v)\) and obtain its full threshold value \(1\). AF therefore implies $p_e \leqslant p_u + p_v \leqslant \frac23$. If at least one of \(u, v\) does not belong to \(I\), then we only use individual rationality, which gives $p_e \leqslant 1$. Hence total revenue from edge types is at most
 \begin{equation} \label{eq:edge-revenue-upper}
  \sum_{e \in E} p_e \leqslant \frac23 \,|E[I]| + \left(M - |E[I]|\right).
 \end{equation}
 Combining~\eqref{eq:vertex-revenue-upper} and~\eqref{eq:edge-revenue-upper}, we have $\widetilde R(\mathcal M) \leqslant M + \frac{|I| - |E[I]|}{3}$. From the induced subgraph \(G[I]\), deleting one endpoint from each edge leaves an independent set of size at least \(|I| - |E[I]|\). Therefore $|I| - |E[I]| \leqslant \alpha$, and so
 \begin{equation} \label{eq:mis-upper-bound}
  \widetilde R(\mathcal M) \leqslant M + \frac{\alpha}{3}.
 \end{equation}
 Since \(\mathcal M\) was arbitrary, \eqref{eq:mis-lower-bound} and \eqref{eq:mis-upper-bound} imply $\widetilde R^*(G) = M + \frac{\alpha}{3}$.

 \paragraph{Step 4: Final reduction.} Returning to normalized type masses \(1 / (N + M)\), write \(R^*(G)\) for the optimal normalized revenue of the instance constructed from \(G\). Then $R^*(G) = \frac{M + \alpha / 3}{N + M}$. Suppose a PTAS for revenue maximization under AF existed. Then for every fixed \(\eta > 0\), it would output a revenue \(R\) satisfying $R \geqslant (1 - \eta) R^*(G)$.
 
 Define $\widehat\alpha = 3\left((N + M)R - M\right)$. Since \(R \leqslant R^*(G)\), we have \(\widehat\alpha \leqslant \alpha\). On the other hand,$\widehat\alpha \geqslant 3\left((1 - \eta)\left(M + \frac{\alpha}{3}\right) - M\right) = \alpha - 3\eta M - \eta\alpha$. Using~\eqref{eq:mis-edge-alpha-new}, we have $\widehat\alpha \geqslant \alpha - \eta\left(\frac{3\Delta(\Delta + 1)}2 + 1\right)\alpha$. Thus, by choosing \(\eta\) to be a sufficiently small constant, one would obtain a PTAS for Maximum Independent Set on bounded-degree graphs, a contradiction. This proves the result.
\end{proof}

\begin{proof}[Proof of Proposition~\ref{prop:arbitrage-detection-log-approx}]
  We first prove the hardness by reducing from weighted set cover.  An
  instance consists of $n$ sets $S_1, \ldots, S_n \subseteq U$, integer costs $c_1, \ldots, c_n$, and a budget $K$.  Let the state space be $\Omega = \{\bot\}\cup U$.  For every set $S_j$, create a binary experiment $E_j$ that reveals whether $\omega\in S_j$, and set its price to $c_j / H$, where $H = 4|U|(K + 1)$.  Also create the target experiment $F$, which reveals whether $\omega\in\{\bot\}$ or $\omega \in U$, at price $(K + 1 / 2) / H$.

  Consider bundles $E_J = \otimes_{j \in J}E_j$ for subsets $J \subseteq[n]$. If $\bigcup_{j \in J}S_j = U$, then $E_J$ perfectly distinguishes $\bot$ from all states in $U$, so $E_J \succeqB F$. The bundle is cheaper than $F$ exactly when $\sum_{j \in J} c_j / H < (K + 1 / 2) / H$, which is equivalent to $\sum_{j \in J} c_j \leqslant K$ since the costs are integral. Conversely, if $J$ is not a cover,  $E_J$ cannot Blackwell dominate $F$. Thus a Blackwell AF violation exists if and only if the set cover instance has a cover of cost at most $K$.

  Then consider the approximation algorithm. Note that we can restrict attention to integral price vectors \(p\) in \(\{0, 1, \ldots, V\}^M\) without loss of generality. For a pricing vector \(p = (p_i)_{i = 1}^M\), define \(m_p\) by \(m_p(i) = \min_{a \in \mathcal C_i} a \cdot p\), meaning that \(m_p(i)\) is the minimum total price of a finite bundle of menu targets that Blackwell dominates \(T_i\). Since \(e_i \in \mathcal C_i\), we have \(m_p \leqslant p\). Note that the vector \(m_p\) is Blackwell AF, and any AF vector \(\tilde p \leqslant p\) satisfies \(\tilde p_i \leqslant a \cdot \tilde p \leqslant a \cdot p\) for every \(a \in \mathcal C_i\), and hence \(\tilde p_i \leqslant m_p(i)\).

  Then we show that the vector \(m_p\) can be computed in polynomial time using the AF violation detection oracle. Since \(p\) is integral, \(m_p(i)\) is an integer in \(\{0, \ldots, p_i\}\). For an integer \(x \leqslant p_i\), set \(\tilde p_i = x\) and \(\tilde p_j = p_j\) for \(j \neq i\). A violation for target \(i\) under \(\tilde p\) exists if and only if \(m_p(i) < x\). Therefore binary search over \(x\) computes each \(m_p(i)\) using \(O(\log V)\) detection calls.

  Let \(\mathcal L = \{1, 2, 4, \ldots, 2^{\lfloor\log_2V\rfloor}\}\). For each \(L \in \mathcal L\), set \(p_i^L = \min\{v_\theta: \theta \in \Theta_i, \ v_\theta \geqslant L\}\) if this set is nonempty, and set \(p_i^L = L\) otherwise. If \(m_{p^L}\) is a price vector, its revenue is \(R(L) = \sum_{i = 1}^M m_{p^L}(i)D_i(m_{p^L}(i))\). Let \(L^\star\) maximize \(R(L)\) over \(L \in \mathcal L\). We show that \(m_{p^{L^\star}}\) achieves a \((1 + \lfloor\log_2V\rfloor)\)-approximation.

  Let \(p^\star\) be an integral optimal AF vector for Program~\eqref{eq:threshold-target-demand-program}, and let \(A^\star = \{\theta: p^\star_{\tau(\theta)} \leqslant v_\theta\}\). Fix \(L \in \mathcal L\). If \(\theta \in A^\star\) and \(p^\star_{\tau(\theta)} \geqslant L\), we want to prove that \(L \leqslant m_{p^L}(\tau(\theta)) \leqslant v_\theta\). This means that every type who buys under the optimal price vector at a price at least \(L\) also buys under \(m_{p^L}\), and pays at least \(L\). For the upper bound, since \(v_\theta \geqslant p^\star_{\tau(\theta)} \geqslant L\), we have \(p_{\tau(\theta)}^L \leqslant v_\theta\) by definition. Since \(m_{p^L} \leqslant p^L\), we have \(m_{p^L}(\tau(\theta)) \leqslant v_\theta\). For the lower bound, let \(q_i^L = \min\{p_i^\star, L\}\). The vector \(q^L\) is AF since for every \(a \in \mathcal C_i\), $q_i^L \leqslant \min\{a \cdot p^\star, L\} \leqslant \sum_j a_j \min\{p_j^\star, L\} = a \cdot q^L$. Since \(q_i^L \leqslant L \leqslant p_i^L\) for every \(i\) and \(m_{p^L}\) is the largest AF vector whose every coordinate is at most the corresponding coordinate of \(p^L\), we have \(q^L \leqslant m_{p^L}\). Since \(p^\star_{\tau(\theta)} \geqslant L\), we have \(q_{\tau(\theta)}^L = L\), hence \(L \leqslant m_{p^L}(\tau(\theta))\). 
  
  Therefore, $R(L) \geqslant \sum_{\theta \in A^\star: p^\star_{\tau(\theta)} \geqslant L} \lambda_\theta L$. According to the construction of \(\mathcal L\) and $p^\star_{\tau(\theta)}$ is an integer, we have
  \[ \sum_{L \in \mathcal L} R(L) \geqslant \sum_{\theta \in A^\star} \lambda_\theta \sum_{L \leqslant p^\star_{\tau(\theta)}}L \geqslant \sum_{\theta \in A^\star}\lambda_\theta p^\star_{\tau(\theta)} = \OPT. \]
  Since \(|\mathcal L| = 1 + \lfloor\log_2V\rfloor\), level \(L^\star\) satisfies \(R(L^\star) \geqslant \OPT / (1 + \lfloor\log_2V\rfloor)\).
\end{proof}

Note that if the value range is continuous and we enlarge the \(\mathcal L\) to all values in \([1, V]\), then the same construction can achieve the same approximation ratio as Proposition~\ref{prop:full-info-approx}. Indeed, this construction contains the case of selling only the complete information. For the trade-off between running time and precision, choosing the powers of $2$ is a convenient choice.

\begin{proof}[Proof of Proposition~\ref{prop:interval-shortest-path}]
  Fix a posted price vector \(p\). Let \(X\) be the sorted set of all endpoints appearing in the menu intervals. Build a directed graph on \(X\). For every menu interval \(J = [\ell, r]\), add an edge \(\ell \to r\) with cost \(p_J\). Also add a zero-cost edge from each endpoint to the previous endpoint in the sorted order except the smallest one.

  We claim that, for every target interval \(I = [a, b]\), the shortest path distance from \(a\) to \(b\) equals the minimum cost of a bundle of menu intervals whose union covers \(I\). Let \(\mathcal A\) be any collection of menu intervals covering \([a, b]\), then \(\mathcal A\) contains a subcollection that can be ordered as \([\ell_1, r_1], \ldots, [\ell_t, r_t]\), with \(\ell_1 \leqslant a \leqslant r_1\), \(\ell_{s + 1} \leqslant r_s\) for \(s < t\), and \(r_t \geqslant b\). Starting from \(a\), use zero-cost left edges to move to \(\ell_1\), take the interval edge to \(r_1\), then repeat the same operation for the next interval; at the end, if \(r_t > b\), move left to \(b\). Thus any cover of \([a, b]\) can be viewed as a path from \(a\) to \(b\). Conversely, any path from \(a\) to \(b\) gives a cover. Therefore the shortest path distance is exactly the cheapest covering cost.

  Thus, target \(i\) has a Blackwell AF violation exactly when the shortest path distance from \(a_i\) to \(b_i\) is less than \(p_i\). The graph has \(O(M)\) vertices and \(O(M)\) edges. Running Dijkstra's algorithm from each target left endpoint to the corresponding right endpoint gives total time \(O(M^2\log M)\).
\end{proof}

\begin{proof}[Proof of Proposition~\ref{prop:quality-based-detection}]
  For each coordinate \(r \in [d]\), sort the \(M\) numbers \(x_{1r}, \ldots, x_{Mr}\) and remove repetitions. Let \(L_r\) be this set together with an initial value \(\bot_r\), which represents no purchased product in coordinate \(r\). Then \(|L_r| \leqslant M + 1\). Define \(\mathcal C = \left\{\bigvee_{j \in A}x_j: A \subseteq [M]\right\}\) and \(\bot = (\bot_1, \ldots, \bot_d)\). Then \(N := |\mathcal C| \leqslant \prod_{r = 1}^d |L_r| \leqslant (M + 1)^d\). Build a directed graph on \(\mathcal C\). For every \(s \in \mathcal C\) and every product \(x_j\), add an edge \(s \to s \vee x_j\) of cost \(p_j\). A bundle \(j_1, \ldots, j_k\) corresponds to the path
  \[ \bot \to x_{j_1} \to x_{j_1} \vee x_{j_2} \to \cdots \to \bigvee_{\ell = 1}^k x_{j_\ell} \]
  with the same total cost. Conversely, every path from \(\bot\) records a sequence of purchased products, and its endpoint is the join of those products. Thus the minimum cost of a bundle that Blackwell dominates target \(i\) is \(\min_{s \geqslant x_i}\operatorname{dist}_p(s)\), where \(\operatorname{dist}_p(s)\) is the shortest path distance from \(\bot\) to \(s\). An arbitrage violation exists exactly when \(p_i > \min_{s \geqslant x_i}\operatorname{dist}_p(s)\) for some target \(i\). 

  The graph has \(N\) vertices and at most \(NM\) edges. Dijkstra's algorithm computes all distances from \(\bot\) in \(O(NM\log N)\) time, and checking all $s \geqslant x_i$ for all $i$ takes \(O(MN)\) additional time. Since \(d\) is fixed and \(N \leqslant(M + 1)^d\), the total running time is \(O(M^{d + 1}\log M)\).
\end{proof}

\begin{proof}[Proof of Proposition~\ref{prop:tree-dp}]
  We claim that Blackwell AF is equivalent to the two local constraints: for every child \(w \in \mathrm{Ch}(u)\), \(p_w \leqslant p_u\); and for every internal node \(u\), \(p_u \leqslant\sum_{w \in \mathrm{Ch}(u)}p_w\). These constraints are necessary because the parent dataset \(S_u\) reveals each child dataset \(S_w\), and all child datasets together reveals \(S_u\). They are also sufficient. Fix a node \(u\), and consider any bundle whose union covers \(S_u\). If the bundle contains an ancestor \(a\) of \(u\), then \(p_a \geqslant p_u\). Otherwise, the selected products lie in the descendant of \(u\). Since the children partition \(S_u\), the bundle must cover every \(S_w\) with \(w\in\mathrm{Ch}(u)\). Using induction, if covering \(S_w\) inside the subtree of \(w\) costs at least \(p_w\), the bundle costs at least \(\sum_{w\in\mathrm{Ch}(u)}p_w\geqslant p_u\). Thus the local constraints are exactly the Blackwell AF constraints for the tree.

  Now we optimize over these local constraints. For each node \(u\), recall \(D_u(p)\) is the total weight of types in \(\Theta_u\) whose value is at least \(p\), and set \(r_u(p) = pD_u(p)\). Besides, note that some optimal solution has prices in \(\{0, 1, \ldots, V\}\).

  For a node \(u\), let \(F_u(p_u)\) be the maximum revenue obtainable from the subtree rooted at \(u\). If \(u\) is a leaf, then \(F_u(p_u) = r_u(p_u)\). Otherwise, its child prices \(x_w\) must satisfy \(0 \leqslant p_w \leqslant p_u\) and \(\sum_w p_w \geqslant p_u\). Once these prices are fixed, the child subtrees are independent, so
  \[ F_u(p_u) = r_u(p_u) + \max_{\substack{0 \leqslant p_w \leqslant p_u, \forall w \\ \sum_w p_w \geqslant p_u}} \sum_{w \in \mathrm{Ch}(u)} F_w(p_w). \]

  For fixed \(u\) and \(p_u\), we can compute the inner maximization by a dynamic programming algorithm. Let $\mathrm{Ch}(u) = \{w_1, \ldots, w_q\}$. For each child \(w_\ell\), let \(B_\ell(t)\) be the maximum total revenue from the children $\{w_1, \ldots, w_\ell\}$ when the sum of their prices is \(t\), where sums larger than \(V\) are recorded as \(V\). The algorithm is described in Algorithm~\ref{alg:child-value-maximization}. This takes \(O(|\mathrm{Ch}(u)|V^2)\) time for fixed \(p_u\), and hence \(O(|\mathrm{Ch}(u)|V^3)\) time to compute all entries \(F_u(0), \ldots, F_u(V)\). Summing over all nodes gives \(O(|\mathcal F|V^3)\), and maximizing \(F_r(p_r)\) over the root price \(p_r\) returns the optimal revenue, and storing the maximizing child prices recovers an optimal price vector.
\end{proof}

\begin{algorithm}[htbp]
    \caption{\textsc{Child Value Maximization}} \label{alg:child-value-maximization}
    \begin{algorithmic}[1]
      \State $B_0(0) \gets 0$ and $B_0(t) \gets -\infty$ for all $t = 1, \ldots, V$.
      \For{$\ell = 1, \ldots, q$}
          \State Initialize $B_\ell(t) \gets -\infty$ for all $t = 0, \ldots, V$.
          \For{$t = 0, \ldots, V$}
              \For{$p_{w_\ell} = 0, \ldots, p_u$}
                  \State $t' \gets \min\{V, t + p_{w_\ell}\}$, $B_\ell(t') \gets \max\{B_\ell(t'), B_{\ell - 1}(t) + F_{w_\ell}(p_{w_\ell})\}$.
              \EndFor
          \EndFor
      \EndFor
      \State \Return \(\max_{t \geqslant p_u}B_q(t)\).
    \end{algorithmic}
  \end{algorithm}

\begin{proof}[Proof of Proposition~\ref{prop:spectral-gauge-af}]
  Fix \(W \in \PD^d\), and sort the products so that \(s_i = \operatorname{tr}(WJ_i)\) is nondecreasing. Let \(p\) be \(W\)-regular. Consider any target \(i\) and any count vector \(c \in \mathbb Z_+^M\) such that \(\sum_j c_jJ_j \succeq J_i\). Multiplying by \(W\) and taking traces gives \(\sum_j c_js_j \geqslant s_i\).
  If the bundle contains some product \(j\) with \(s_j \geqslant s_i\), then the bundle costs at least \(p_i\) by monotonicity. Otherwise all purchased products satisfy \(s_j < s_i\). Using \(p_j / s_j \geqslant p_i / s_i\), the bundle cost satisfies \(\sum_j c_jp_j \geqslant (p_i/s_i)\sum_j c_js_j\geqslant p_i\). Thus every bundle that Blackwell dominates \(G_{J_i}\) costs at least \(p_i\). Taking \(c = e_i\) is feasible and has cost \(p_i\), so~\eqref{eq:matrix-af-counts} holds.
  
  Now we optimize over \(W\)-regular prices. Let \(\mathcal V = \bigcup_i\{v_\theta: \theta \in \Theta_i\}\) be the set of buyer values. For each index \(i\), define \(\mathcal P_i = \{0\} \cup \mathcal V \cup \{(s_i / s_k)v: k \in [M], \ v \in \mathcal V\}\). We next show that it suffices to optimize over the finite sets \(\mathcal P_i\). Note that if \(s_i = s_k\) for \(i < k\) (\(i > k\) is similar), then \(W\)-regularity implies \(p_i \leqslant p_k\) and \(p_k \leqslant p_i\), so \(p_i = p_k\). Therefore, we can assume that \(s_1 < s_2 < \cdots < s_M\) without loss of generality. Let \(v_{\max} = \max \mathcal V\), then we can restrict attention to price vectors with \(0 \leqslant p_i \leqslant v_{\max}\) for all \(i\). Among all optimal \(W\)-regular price vectors in this compact set, choose one that maximizes \(\sum_i p_i\). We prove that this price vector satisfies \(p_i \in \mathcal P_i\) for every \(i\) by contradiction.

  Suppose otherwise that \(p_i \notin \mathcal P_i\) for some \(i\). In particular, \(p_i \notin \{0\} \cup \mathcal V\). Let \(p_a / s_a = p_{a + 1} / s_{a + 1} = \cdots = p_i / s_i\), and either \(a = 1\) or \(p_{a - 1} / s_{a - 1} > p_a / s_a\), and let \(p_i = p_{i + 1} = \cdots = p_b\), and either \(b = M\) or \(p_b < p_{b + 1}\). Define \(d_j = s_j / s_i\) for \(a \leqslant j \leqslant i\), \(d_j = 1\) for \(i < j \leqslant b\), and \(d_j = 0\) otherwise. For sufficiently small \(\varepsilon > 0\), let \(p'_j = p_j + \varepsilon d_j\) for every \(j\). We first check that \(p'\) is still \(W\)-regular. For \(j \in \{a, \ldots, i\}\), all unit prices \(p_j / s_j\) are increased by the same amount \(\varepsilon / s_i\), so $p'_a / s_a = \cdots = p'_i / s_i$, and for \(j \in \{i, \ldots, b\}\), $p'_i = \cdots = p'_b$. Since \(s_1 < \cdots < s_M\), we have $p'_a < \cdots < p'_i$ and $p'_i / s_i > \cdots > p'_b / s_b$. For $j < a (a > 1)$ or $j > b (b < M)$, we choose \(\varepsilon\) small enough so that \(p_a / s_a + \varepsilon / s_i \leqslant p_{a - 1} / s_{a - 1}\) and \(p_b + \varepsilon \leqslant p_{b + 1}\), and all other constraints are unchanged or become easier to satisfy. Hence \(p'\) is \(W\)-regular.

  It remains to check that the demand is unchanged for sufficiently small \(\varepsilon\). For every \(j \in \{i,\ldots,b\}\), we have \(p_j = p_i\), so \(p_j \notin \mathcal V\). For every \(j \in \{a,\ldots,i\}\), we have \(p_j = (s_j/s_i)p_i\). If \(p_j \in \mathcal V\), then \(p_i = (s_i/s_j)p_j \in \mathcal P_i\), contradicting \(p_i \notin \mathcal P_i\). Thus no modified price \(p_j\) lies in \(\mathcal V\). Since \(\mathcal V\) is finite, we can choose \(\varepsilon > 0\) small enough so that no modified price crosses any value in \(\mathcal V\), and hence \(D_j(p'_j) = D_j(p_j)\) for every modified index \(j\). Therefore, $\sum_j p'_j D_j(p'_j) = \sum_j p'_j D_j(p_j) \geqslant \sum_j p_j D_j(p_j)$. Moreover, \(d_i = 1\), so \(\sum_j p'_j > \sum_j p_j\). If the revenue strictly increases, this contradicts the optimality of \(p\); if the revenue is unchanged, this contradicts the choice of \(p\) as an optimal \(W\)-regular price vector maximizing \(\sum_j p_j\). Hence no such \(i\) exists, and there is an optimal \(W\)-regular price vector with \(p_i \in \mathcal P_i\) for every \(i\).

  We can now use dynamic programming. Let \(A_i(\alpha)\) be the maximum revenue from products \(1, \ldots, i\), conditional on \(p_i = \alpha\), where \(\alpha \in \mathcal P_i\). Initialize \(A_1(\alpha) = \alpha D_1(\alpha)\). For \(i \geqslant 2\), suppose \(p_i = \alpha\) and \(p_{i - 1} = \beta\). The \(W\)-regular constraints are equivalent to \(\alpha s_{i - 1} / s_i \leqslant \beta \leqslant \alpha\). Thus
  \[ A_i(\alpha) = \alpha D_i(\alpha) + \max_{\beta \in \mathcal P_{i - 1}: \ \alpha s_{i - 1} / s_i \leqslant \beta \leqslant \alpha} A_{i - 1}(\beta). \]
  The optimal \(W\)-regular revenue is \(\max_{\alpha \in \mathcal P_M}A_M(\alpha)\). Since \(|\mathcal P_i| = O(M|\mathcal V|)\), the dynamic program has polynomial size.
\end{proof}

\begin{proof}[Proof of Proposition~\ref{prop:spectral-gauge-approx}]
  Let \(p^\star\) be an optimal price vector that satisfies Equation~\eqref{eq:matrix-af-counts}, so \(\operatorname{Rev}(p^\star) = \OPT\). For the fixed weight matrix \(W\), define \(\widehat p_i = \min_k p_k^\star\max\{1, s_i / s_k\}\). We first show that \(\widehat p\) is \(W\)-regular. If \(s_i \leqslant s_j\), then \(\max\{1, s_i / s_k\} \leqslant \max\{1, s_j / s_k\}\) for every \(k\), so \(\widehat p_i \leqslant \widehat p_j\). Also, \(s^{-1}\max\{1, s / s_k\} = \max\{1 / s, 1 / s_k\}\) is nonincreasing in \(s\), so taking the minimum over \(k\) gives \(\widehat p_i / s_i \geqslant \widehat p_j / s_j\). Thus \(\widehat p\) is \(W\)-regular.

  Choosing \(k = i\) gives \(\widehat p_i \leqslant p_i^\star\). We then need to find the lower bound of \(\widehat p_i\). Since \(\kappa_{ik}J_k \succeq J_i\),  Blackwell AF of \(p^\star\) gives \(p_i^\star \leqslant \kappa_{ik}p_k^\star\) for all \(k\). By the definition of \(\Gamma_W\), \(\kappa_{ik} \leqslant \Gamma_W \max\{1, s_i / s_k\}\). Hence \(p_i^\star \leqslant \Gamma_Wp_k^\star \max\{1, s_i / s_k\}\) for all \(k\), and taking the minimum over \(k\) yields \(\widehat p_i \geqslant p_i^\star / \Gamma_W\).

  For every buyer type \(\theta \in \Theta_i\) who buys the product under \(p^\star\), we have \(p_i^\star \leqslant v_\theta\). Since \(\widehat p_i \leqslant p_i^\star\), $\theta$ will also buy the product under \(\widehat p\), and since \(\widehat p_i \geqslant p_i^\star / \Gamma_W\), its payment is at least a \(1 / \Gamma_W\) fraction of the original payment. Hence \(\operatorname{Rev}(\widehat p) \geqslant \operatorname{Rev}(p^\star) / \Gamma_W = \OPT / \Gamma_W\). Since \(p^W\) maximizes revenue over all \(W\)-regular price vectors, it also achieves revenue at least \(\OPT / \Gamma_W\).
\end{proof}

\end{document}